\PassOptionsToPackage{unicode}{hyperref}
\PassOptionsToPackage{hyphens}{url}
\PassOptionsToPackage{dvipsnames,svgnames,x11names}{xcolor}
\documentclass[
  12pt,
  letterpaper]{article}

\usepackage{amsmath,amssymb}
\usepackage{setspace}
\usepackage{iftex}
\ifPDFTeX
  \usepackage[T1]{fontenc}
  \usepackage[utf8]{inputenc}
  \usepackage{textcomp} 
\else 
  \usepackage{unicode-math}
  \defaultfontfeatures{Scale=MatchLowercase}
  \defaultfontfeatures[\rmfamily]{Ligatures=TeX,Scale=1}
\fi
\usepackage{lmodern}
\ifPDFTeX\else  
\fi
\IfFileExists{upquote.sty}{\usepackage{upquote}}{}
\IfFileExists{microtype.sty}{
  \usepackage[]{microtype}
  \UseMicrotypeSet[protrusion]{basicmath} 
}{}
\makeatletter
\@ifundefined{KOMAClassName}{
  \IfFileExists{parskip.sty}{%
    \usepackage{parskip}
  }{
    \setlength{\parindent}{0pt}
    \setlength{\parskip}{6pt plus 2pt minus 1pt}}
}{
  \KOMAoptions{parskip=half}}
\makeatother
\usepackage{xcolor}
\usepackage[margin=1in]{geometry}
\makeatletter
\ifx\paragraph\undefined\else
  \let\oldparagraph\paragraph
  \renewcommand{\paragraph}{
    \@ifstar
      \xxxParagraphStar
      \xxxParagraphNoStar
  }
  \newcommand{\xxxParagraphStar}[1]{\oldparagraph*{#1}\mbox{}}
  \newcommand{\xxxParagraphNoStar}[1]{\oldparagraph{#1}\mbox{}}
\fi
\ifx\subparagraph\undefined\else
  \let\oldsubparagraph\subparagraph
  \renewcommand{\subparagraph}{
    \@ifstar
      \xxxSubParagraphStar
      \xxxSubParagraphNoStar
  }
  \newcommand{\xxxSubParagraphStar}[1]{\oldsubparagraph*{#1}\mbox{}}
  \newcommand{\xxxSubParagraphNoStar}[1]{\oldsubparagraph{#1}\mbox{}}
\fi
\makeatother

\providecommand{\tightlist}{%
  \setlength{\itemsep}{0pt}\setlength{\parskip}{0pt}}\usepackage{longtable,booktabs,array}
\usepackage{calc} 
\usepackage{etoolbox}
\makeatletter
\patchcmd\longtable{\par}{\if@noskipsec\mbox{}\fi\par}{}{}
\makeatother
\IfFileExists{footnotehyper.sty}{\usepackage{footnotehyper}}{\usepackage{footnote}}
\makesavenoteenv{longtable}
\usepackage{graphicx}
\makeatletter
\newsavebox\pandoc@box
\newcommand*\pandocbounded[1]{
  \sbox\pandoc@box{#1}%
  \Gscale@div\@tempa{\textheight}{\dimexpr\ht\pandoc@box+\dp\pandoc@box\relax}%
  \Gscale@div\@tempb{\linewidth}{\wd\pandoc@box}%
  \ifdim\@tempb\p@<\@tempa\p@\let\@tempa\@tempb\fi
  \ifdim\@tempa\p@<\p@\scalebox{\@tempa}{\usebox\pandoc@box}%
  \else\usebox{\pandoc@box}%
  \fi%
}
\def\fps@figure{htbp}
\makeatother
\NewDocumentCommand\citeproctext{}{}
\NewDocumentCommand\citeproc{mm}{%
  \begingroup\def\citeproctext{#2}\cite{#1}\endgroup}
\makeatletter
 \let\@cite@ofmt\@firstofone
 \def\@biblabel#1{}
 \def\@cite#1#2{{#1\if@tempswa , #2\fi}}
\makeatother
\newlength{\cslhangindent}
\newlength{\csllabelwidth}
\newenvironment{CSLReferences}[2] 
 {\begin{list}{}{%
  \setlength{\itemindent}{0pt}
  \setlength{\leftmargin}{0pt}
  \setlength{\parsep}{0pt}
  \ifodd #1
   \setlength{\leftmargin}{\cslhangindent}
   \setlength{\itemindent}{-1\cslhangindent}
  \fi
  \setlength{\itemsep}{#2\baselineskip}}}
 {\end{list}}
\usepackage{calc}

\usepackage{amsthm}
\ifPDFTeX
  \DeclareUnicodeCharacter{2212}{\ensuremath{-}}   
  \DeclareUnicodeCharacter{00D7}{\ensuremath{\times}}
  \DeclareUnicodeCharacter{2192}{\ensuremath{\rightarrow}}
  \DeclareUnicodeCharacter{0394}{\ensuremath{\Delta}}
  \DeclareUnicodeCharacter{2264}{\ensuremath{\leq}}
  \DeclareUnicodeCharacter{2265}{\ensuremath{\geq}}
\fi
\usepackage{needspace}   
\makeatletter
\@ifpackageloaded{caption}{}{\usepackage{caption}}
\AtBeginDocument{%
\ifdefined\contentsname
  \renewcommand*\contentsname{Table of contents}
\else
  \newcommand\contentsname{Table of contents}
\fi
\ifdefined\listfigurename
  \renewcommand*\listfigurename{List of Figures}
\else
  \newcommand\listfigurename{List of Figures}
\fi
\ifdefined\listtablename
  \renewcommand*\listtablename{List of Tables}
\else
  \newcommand\listtablename{List of Tables}
\fi
\ifdefined\figurename
  \renewcommand*\figurename{Figure}
\else
  \newcommand\figurename{Figure}
\fi
\ifdefined\tablename
  \renewcommand*\tablename{Table}
\else
  \newcommand\tablename{Table}
\fi
}
\@ifpackageloaded{float}{}{\usepackage{float}}
\floatstyle{ruled}
\@ifundefined{c@chapter}{\newfloat{codelisting}{h}{lop}}{\newfloat{codelisting}{h}{lop}[chapter]}
\floatname{codelisting}{Listing}

\usepackage{amsthm}
\theoremstyle{plain}
\newtheorem{proposition}{Proposition}
\theoremstyle{plain}
\newtheorem{lemma}{Lemma}
\theoremstyle{plain}
\newtheorem{corollary}{Corollary}
\theoremstyle{plain}
\newtheorem{theorem}{Theorem}
\theoremstyle{remark}
\AtBeginDocument{}
\newtheorem*{remark}{Remark}

\makeatother
\makeatletter
\@ifpackageloaded{caption}{}{\usepackage{caption}}
\@ifpackageloaded{subcaption}{}{\usepackage{subcaption}}
\makeatother

\usepackage{bookmark}

\IfFileExists{xurl.sty}{\usepackage{xurl}}{} 
\hypersetup{
  pdftitle={Panel Conditioning in Fixed-Effects Models: Identification and Bias Propagation},
  pdfauthor={Shoki Okubo},
  pdfkeywords={panel
conditioning, identification, age-period-cohort, two-way fixed
effects, event study, refreshment samples},
  colorlinks=true,
  linkcolor={blue},
  filecolor={Maroon},
  citecolor={Blue},
  urlcolor={Blue},
  pdfcreator={LaTeX via pandoc}}

\title{Panel Conditioning in Fixed-Effects Models: Identification and
Bias Propagation\thanks{This research benefited from discussions and feedback during presentations at the Institute of Social Science, University of Tokyo, the Japanese Association for Mathematical Sociology, and the panel survey conference at Keio University. Code reproducing every numerical result in this paper, including the downstream identities of Section 6 and every row of Tables 1 and 2, is in the replication archive at \url{https://github.com/sokubo/paper-panel-conditioning-replication} (fixed version: tag \texttt{paper-v1.0}, commit \texttt{f06283e}). No restricted data are used. This work was supported by JSPS KAKENHI Grant Number 22K13525.}}
\author{Shoki Okubo\thanks{Department of Sociology, Toyo University, Tokyo, Japan. Email: okubo080@toyo.jp. Website: sokubo.github.io.}}
\date{September 24, 2026}

\begin{document}
\maketitle
\begin{abstract}
Panel conditioning, the causal effect of prior survey participation on
responses, can vary with tenure. Under an additive model of cell means
in period, entry cohort, and tenure, we characterize which features of
the conditioning path a staggered panel identifies on its observed
support, and how the unidentified component affects common panel
estimators. The identified set of the path is an affine translate of the
tenure projection of the cell design's kernel, and a linear functional
of the path is identified exactly when it annihilates that projection.
It always contains an affine direction and, when the entry cohorts share
a stride, periodic directions, which exhaust it under a connectivity
condition on observed increments; second differences at that stride are
then identified, and ordinary ones generally are not when the stride
exceeds one. Under a recruitment condition, an interrupted schedule such
as the four-eight-four rotation of the Current Population Survey (CPS)
distinguishes a constant increment per interview from one per calendar
month, which no equally spaced schedule can. We give support conditions
for recovery under a plateau, entry-wave negative controls, or bounded
cohort drift. A second set of results links identification to
regression: two-way fixed effects absorb every unidentified direction,
so the remaining conditioning bias is normalization-invariant and itself
identified, and a two-way regression with tenure indicators corrects it
under a residual-rank condition. When event time is aligned with tenure,
conditioning shifts event-study coefficients by a known linear
functional of the path, producing pre-trends without anticipation;
bounds on identified curvature bound those shifts. Simulations verify
the identities, a 19-wave Japanese panel illustrates the support
calculations, and published CPS month-in-sample indices give a
descriptive, not identifying, example.
\end{abstract}


\setstretch{1.5}
\section{Introduction}\label{introduction}

Panel surveys ask the same individuals the same questions repeatedly,
and a long literature documents that the act of being interviewed
changes later answers---\emph{panel conditioning}
(\citeproc{ref-warren2012}{Warren and Halpern-Manners 2012};
\citeproc{ref-halpernmanners2012}{Halpern-Manners and Warren 2012};
\citeproc{ref-halpernmanners2017}{Halpern-Manners, Warren, and Torche
2017}). Formal identification analysis of the phenomenon is rare: Das,
Toepoel, and van Soest (\citeproc{ref-dasetal2011}{2011}) give one for a
binary item in a two-wave design with a refreshment sample, and Feng,
Hu, and Sun (\citeproc{ref-fenghusun2022}{2022}) give one for a
latent-state model of labor-force misclassification in the Current
Population Survey. For the object that most of the applied literature
reports---a mean path of conditioning over tenure, estimated from the
cell means of a staggered panel---practitioners have long behaved as if
a non-identification result existed: the U.S. Bureau of Labor Statistics
normalizes month-in-sample effects to average to one; Krueger, Mas, and
Niu (\citeproc{ref-krueger2017}{2017}) subtitle their study of rotation
group bias ``Will the Real Unemployment Rate Please Stand Up?'' and
conclude that it remains unclear which rotation group provides the more
accurate estimate (\citeproc{ref-krueger2017}{Krueger, Mas, and Niu
2017, 264}); van den Brakel and Krieg
(\citeproc{ref-brakelkrieg2015}{2015}) estimate rotation-group bias only
relative to a reference group. Each of these choices is a normalization.
This paper states the result that forces it, proves it on the panel's
observed cohort--period support, characterizes what it leaves
identified---which depends on the refreshment schedule and on how long
each cohort is followed---and derives what the unidentified component
does to the fixed-effects and event-study estimators that panel data are
used for.

The first half concerns identification. Theorem~\ref{thm-linear} states
what a panel identifies about the conditioning path \(\tau(\cdot)\) on
its observed cohort--period support: the identified set is the translate
of the tenure projection of the cell design's kernel, and a linear
functional of \(\tau\) is identified if and only if it annihilates that
projection. Two kinds of direction are in the kernel of every staggered
design. The affine direction is there because tenure \(s\), calendar
time \(t\), and entry cohort \(e\) satisfy \(s = t - e + 1\); it is the
exact translation into the measurement domain of the age--period--cohort
linear dependency (\citeproc{ref-masetal1973}{Mason et al. 1973};
\citeproc{ref-fosse2019}{Fosse and Winship 2019a}) and of the
underidentification of fully dynamic event-study designs
(\citeproc{ref-borusyak2024}{Borusyak, Jaravel, and Spiess 2024}). When
the entry cohorts share a common stride \(d\)---the greatest common
divisor of their spacings---\(d - 1\) periodic directions are there as
well; these are the measurement-domain form of the additional
non-identifiability of age--period--cohort models with unequal intervals
(\citeproc{ref-gascoignesmith2023}{Gascoigne and Smith 2023}). This
ambiguity has practical consequences: a panel refreshed at a four-wave
stride has, under a tenure-coverage condition stated in the theorem, an
identified set of dimension at least four, and the Japanese panel we use
for illustration is one. Under a connectivity condition on the panel's
observed increments the affine and periodic directions are all there is
in the tenure projection, and Lemma~\ref{lem-trapezoid} gives a
sufficient follow-up period for the last refreshment cohort under which
a staggered panel satisfies it, exact for three cohorts.
Corollary~\ref{cor-curvature} characterizes the identified
functionals---curvature at the design's stride, and ordinary curvature
when the stride is one and the connectivity condition holds.
Section~\ref{sec-interrupted} then removes the assumption that
participation is uninterrupted, which rotation designs violate by
construction. Under continuous recruitment the unidentified direction is
an affine function of the \emph{elapsed calendar time} at which each
interview arrives, not of the interview count, so that a panel whose
interviews are unequally spaced can separate a path that accumulates at
a constant rate per interview from one that accumulates at a constant
rate per month---a distinction no equally spaced design can make. The
Current Population Survey's four-eight-four rotation is such a design,
and Corollary~\ref{cor-cps} exhibits the identified contrast; it is
identified from the joint cohort-by-period array, and we show why the
same weights applied to published month-in-sample averages do not
deliver it. Section~\ref{sec-recovery} gives identifying restrictions (a
plateau, entry-wave negative controls, bounded cohort drift), each
stated together with its support conditions, what it removes, and what
is and is not testable about it: negative controls remove the affine
direction but not the periodic ones, and zero identified curvature is
consistent with a linear path, so saturation cannot be established from
cell means alone.

The second half turns from what cannot be known to where it goes.
Tenure-dependent conditioning raises a measurement concern for panel
analyses that assume stationary measurement error.
Section~\ref{sec-downstream} derives the consequences for the
fixed-effects, difference-in-differences, and event-study estimators
panel users actually run and finds them exactly complementary to
Theorem~\ref{thm-linear}. The \emph{absorption theorem}
(Theorem~\ref{thm-absorb}) shows that two-way fixed-effects estimators
annihilate precisely the unidentified component of the conditioning
path---every kernel direction's tenure component is a sum of a time
function and a unit function, on any support---so that the bias of a
fixed-effects coefficient depends on conditioning only through its
identified part, is invariant to the normalization every conditioning
study must choose, and is itself an identified functional
(Theorem~\ref{thm-correction}). A feasible corrected estimator is the
two-way regression with tenure indicators added, subject to a
residual-rank condition that fails when exposure is scheduled by tenure.
The \emph{spurious pre-trends theorem} (Theorem~\ref{thm-pretrend})
locates the boundary of that protection: when event time is aligned with
tenure, fully dynamic event-study coefficients absorb the non-affine
part of the path exactly, so that pre-period coefficients are nonzero in
expectation without any anticipation or violation of parallel trends in
the outcome, an artifact that outcome-side diagnostics
(\citeproc{ref-roth2022}{Roth 2022};
\citeproc{ref-rambachanroth2023}{Rambachan and Roth 2023}) may detect
but cannot attribute, that cannot be corrected within the aligned
sample, and that an external estimate of the conditioning path can
remove. When external information yields bounds rather than a point
estimate of the conditioning path, Proposition~\ref{prp-sensitivity}
provides a sensitivity analysis: a bound on the path's identified
curvature implies an explicit bound on the measurement-induced shift of
each event-study coefficient. This bound grows quadratically with the
distance from the reference categories. No normalization of the
conditioning level is needed because the shift is invariant to the
unidentified affine direction. An orthogonality characterization
(Proposition~\ref{prp-survive}) states in one line which linear panel
functionals survive which classes of drift.

Relation to Feng, Hu, and Sun (\citeproc{ref-fenghusun2022}{2022}).
Their result achieves \emph{point} identification of CPS
misclassification probabilities by restricting response dynamics to
first-order dependence on the previous report, an exclusion restriction
linking a later interview to the latent status of an earlier one, rank
and eigenvalue conditions, and a cross-rotation-group comparability
condition under which the estimated error probabilities are transported
to the remaining groups. The model class is different from ours---latent
states and misclassification matrices rather than an additive mean
path---and we do not claim that their restrictions select a point of our
identified set. What the two settings share is the dependency between
tenure, period, and cohort; our entry-wave negative controls check one
component of their cross-rotation-group comparability condition---that
entering rotation groups are comparable on time-invariant facts---and
not the full latent-state and response-history transport that their
result requires.

\section{Setup and Notation}\label{sec-setup}

Individuals \(i\) belong to entry cohort
\(e_i \in \mathcal{E} \subset \mathbb{Z}\) and are eligible for
interviews at calendar times \(t \in \mathcal{T} \subset \mathbb{Z}\).
Tenure at time \(t\) is \[
s_{it} \;=\; t - e_i + 1 \;\ge\; 1 .
\] Let \(Y_{ijt}(s)\) denote the potential response of individual \(i\)
to item \(j\) at time \(t\) when the measurement history is coarsened to
tenure \(s\) (Assumption M1 below). Observed responses are
\(Y_{ijt} = Y_{ijt}(s_{it})\) for respondents. We suppress \(j\) where
possible. Define cell means over the target population \[
\mu(e, t) \;=\; \mathbb{E}\!\left[ Y_{it}(t - e + 1) \,\middle|\, e_i = e \right],
\qquad (e, t) \in \mathcal{S},
\] where \(\mathcal{S} \subset \mathcal{E} \times \mathcal{T}\) is the
\emph{observed support}: the finite set of (cohort, period) cells the
panel schedules. The support is a design choice. The leading case is the
staggered trapezoid \(\{(e,t): e \le t \le T\}\), every cohort followed
from entry to a common final period \(T\), but nothing in
Section~\ref{sec-impossibility} requires it; results that do need more
than a finite support say so. Write
\(\mathcal{S}_{\mathrm{obs}} = \{t - e + 1 : (e,t) \in \mathcal{S}\}\)
for the set of observed tenures and \(\mathcal{T}_{\mathrm{obs}}\) for
the set of observed periods. The \emph{cohort--period incidence graph}
of \(\mathcal{S}\) is the bipartite graph with a vertex for each cohort
in \(\mathcal{E}\) and each period in \(\mathcal{T}_{\mathrm{obs}}\) and
an edge for each scheduled cell; it is connected for every staggered
trapezoid, because the first cohort is observed in every period, and it
is disconnected when two groups of cohorts are never observed in a
common period.

We maintain throughout this section (and relax later):

\begin{itemize}
\tightlist
\item
  \textbf{M1 (coarsening / consistency).} Potential responses depend on
  the measurement history only through tenure:
  \(Y_{ijt}(h) = Y_{ijt}(|h|+1)\), and all cohorts are measured with the
  same instrument and mode at a given \(t\). Participation is
  uninterrupted, so that calendar tenure \(t - e + 1\) equals the number
  of interviews received. In interrupted rotation designs---the CPS
  4--8--4 pattern, in which a household is interviewed for four
  consecutive months, rested for eight, and returned for four more---the
  two differ, and the count of interviews, not the calendar spacing, is
  the argument of \(\tau\); Section~\ref{sec-interrupted} treats that
  case, and the difference turns out to carry identifying power rather
  than to be a technicality.
\item
  \textbf{M2 (sampling equivalence).} Cohorts are probability samples
  from the same target population, so that cohort differences in \(\mu\)
  reflect entry-cohort effects and conditioning only.
\item
  \textbf{M3 (no attrition).} Every scheduled cell is realized: at each
  \((e,t) \in \mathcal{S}\) the full entry sample of cohort \(e\)
  responds. A planned end of follow-up---a cell the design never
  schedules---is a feature of \(\mathcal{S}\), not attrition; attrition
  is the failure of a scheduled cell to be realized in full, and is
  relaxed in Section~\ref{sec-equivalence}.
\item
  \textbf{M4 (additive cell means).} On the observed support,
  \begin{equation}\phantomsection\label{eq-model}{
  \mu(e, t) \;=\; \alpha(t) \;+\; g(e) \;+\; \tau(s), \qquad s = t - e + 1,
  }\end{equation} with the normalizations \(\tau(1) = 0\) (no
  conditioning at first interview, by definition) and \(g(e_0) = 0\) for
  a reference cohort \(e_0\).
\end{itemize}

M4 is an assumption about the response process, not a consequence of
M1--M3: it says that period, cohort, and conditioning enter cell means
additively, with no cohort-by-tenure or period-by-tenure interaction. It
is the standard decomposition of the rotation-group-bias and
panel-conditioning literatures, and everything below is a statement
about the additive model; with interactions, the indeterminacy is
item-wise or cohort-wise and larger. The object of interest is the
conditioning path \(\tau\) restricted to the observed tenures,
\(\tau: \mathcal{S}_{\mathrm{obs}} \to \mathbb{R}\); values of \(\tau\)
at tenures the panel never observes are outside the scope of any
identification statement.

\section{The Impossibility Theorem}\label{sec-impossibility}

Let the entry cohorts be \(\mathcal{E} = \{e_1 < \dots < e_K\}\),
\(K \ge 2\), with reference cohort \(e_0 = e_1\), and define the
design's \emph{stride} \[
d \;:=\; \gcd\{\, e_k - e_j : e_k, e_j \in \mathcal{E} \,\} .
\] Consecutive entry cohorts give \(d = 1\); a panel refreshed every
four waves, or with refreshments four and twelve waves after the first
entry, gives \(d = 4\). A function \(\rho: \mathbb{Z} \to \mathbb{R}\)
is \emph{\(d\)-periodic} if \(\rho(s + d) = \rho(s)\) for all \(s\);
when \(d = 1\) the only \(d\)-periodic function with \(\rho(1) = 0\) is
\(\rho \equiv 0\).

\emph{The design matrix and its kernel.} Stack the cell equations
Equation~\ref{eq-model} over \((e,t) \in \mathcal{S}\) as
\(\mu = X_{\mathcal{S}}\,\theta\), where \(\theta = (\alpha, g, \tau)\)
collects the free parameters---\(\alpha(t)\) for
\(t \in \mathcal{T}_{\mathrm{obs}}\), \(g(e)\) for \(e \ne e_0\),
\(\tau(s)\) for
\(s \in \mathcal{S}_{\mathrm{obs}} \setminus \{1\}\)---and
\(X_{\mathcal{S}}\) is the
\(|\mathcal{S}| \times (\text{number of parameters})\) incidence matrix
with a one in the columns of \(\alpha(t)\), \(g(e)\), and
\(\tau(t-e+1)\) of each fielded cell. Two parameter vectors generate the
same cell means if and only if their difference lies in the kernel
\(\mathcal{K} := \ker X_{\mathcal{S}}\). Write
\(h = (h_\alpha, h_g, h_\tau)\) for an element of \(\mathcal{K}\) and
\(\mathcal{K}_\tau := \{h_\tau : h \in \mathcal{K}\}\) for the set of
its tenure components. Everything the panel can and cannot reveal about
the conditioning path is a statement about \(\mathcal{K}_\tau\).

\emph{Increments and connectivity.} An \emph{observed increment} is a
tenure \(u\) such that some cohort is observed at tenures \(u\) and
\(u+1\) in consecutive periods: \((e, t), (e, t+1) \in \mathcal{S}\)
with \(u = t - e + 1\). Let \(\mathcal{U}\) be the set of observed
increments. The \emph{increment graph} has vertex set \(\mathcal{U}\)
and an edge between \(u\) and \(u'\) whenever two cohorts \(e \ne e'\)
are both observed at some \(t\) and \(t+1\) with tenures
\(u = t - e + 1\) and \(u' = t - e' + 1\). Every edge joins increments
whose difference \(e' - e\) is a multiple of \(d\), so the connected
components of the increment graph are contained in residue classes
modulo \(d\). We say that \(\mathcal{S}\) satisfies condition
\textbf{C\(_d\)} (\emph{tenure-connectivity at stride \(d\)}) if (i) it
is \emph{increment-complete}: for every observed tenure \(s\), every
\(u \in \{1, \dots, s-1\}\) is an observed increment, so that each
observed tenure is reached from tenure one through observed increments;
and (ii) the connected components of the increment graph are exactly the
residue classes of \(\mathcal{U}\) modulo \(d\). The staggered trapezoid
is increment-complete because its first cohort is observed in every
period; (ii) is the substantive condition, and Lemma~\ref{lem-trapezoid}
below says when a trapezoid satisfies it. The vertices of the graph are
increments, not tenure levels: the largest observed tenure is never an
increment, and a graph drawn on tenure levels would leave it isolated.

\begin{theorem}[Non-identification of the conditioning
path]\protect\hypertarget{thm-linear}{}\label{thm-linear}

Under M1--M4 and the normalizations \(\tau(1) = 0\), \(g(e_0) = 0\), on
any finite observed support \(\mathcal{S}\):

\begin{enumerate}
\def\labelenumi{(\alph{enumi})}
\item
  \emph{(Identified set.)} The identified set for the conditioning path
  is the affine family \[
  \Theta_\tau \;=\; \tau + \mathcal{K}_\tau \;=\; \{\, \tau + h_\tau : h \in \ker X_{\mathcal{S}} \,\},
  \] and a linear functional
  \(\lambda'\tau = \sum_{s} \lambda_s\,\tau(s)\) is point identified if
  and only if \(\lambda' h_\tau = 0\) for every
  \(h_\tau \in \mathcal{K}_\tau\).
\item
  \emph{(Directions that are never identified.)} For every
  \(m \in \mathbb{R}\) and every \(d\)-periodic \(\rho\) with
  \(\rho(1) = 0\), the vector \[
  h_\tau(s) = m\,(s-1) + \rho(s), \qquad
  h_g(e) = m\,(e - e_0), \qquad
  h_\alpha(t) = -\,m\,(t - e_0) - \rho(t - e_0 + 1)
  \] lies in \(\ker X_{\mathcal{S}}\). These vectors form a linear space
  \(\mathcal{K}_0 \subseteq \mathcal{K}\)---one affine direction and
  \(d - 1\) periodic directions---whose dimension is \(d\) whenever
  \(\mathcal{S}\) contains \(d\) consecutive observed periods, as every
  staggered trapezoid with \(T \ge e_1 + d - 1\) does. Their tenure
  components span a space of dimension \(d\) under a
  \emph{tenure-coverage} condition: that the functions \(s - 1\) and
  \(\mathbf{1}\{s \equiv r\}\), \(r \not\equiv 1\), are linearly
  independent on the set of observed tenures, which holds whenever some
  cohort is observed at \(d + 1\) consecutive tenures. On such supports
  \(\dim \Theta_\tau \ge d\). On sparser supports the projection can
  lose rank even though \(\mathcal{K}_0\) does not: with
  \(\mathcal{E} = \{1, 3\}\) and
  \(\mathcal{S} = \{(1,1), (1,2), (3,3)\}\), \(d = 2\) and
  \(\mathcal{K}_0\) is two-dimensional, but only \(\tau(2)\) is free and
  \(\dim \Theta_\tau = 1\). On every support, no functional with
  \(\sum_s \lambda_s (s-1) \ne 0\) is identified; and when \(d > 1\)
  neither is any functional whose weights fail to sum to zero within
  some residue class modulo \(d\) other than that of \(s = 1\).
\item
  \emph{(Sharpness under C\(_d\).)} If \(\mathcal{S}\) satisfies
  C\(_d\), then \(\mathcal{K}_\tau\) is exactly the set of tenure
  components of \(\mathcal{K}_0\): \[
  \Theta_\tau \;=\; \bigl\{\, s \mapsto \tau(s) + m\,(s-1) + \rho(s) \;:\; m \in \mathbb{R},\ \rho \text{ $d$-periodic},\ \rho(1) = 0 \,\bigr\},
  \] an affine family of dimension \(d\) under the tenure-coverage
  condition of (b), and of dimension equal to the rank of those tenure
  components otherwise; and \(\lambda'\tau\) is identified if and only
  if \(\sum_s \lambda_s (s-1) = 0\) and
  \(\sum_{s \equiv r\,(\mathrm{mod}\ d)} \lambda_s = 0\) for every
  residue \(r \not\equiv 1\). In particular, if \(d = 1\) and C\(_1\)
  holds, \(\Theta_\tau = \{\tau + m(s-1)\}\): the level and linear trend
  of conditioning are not identified and nothing else is lost.
\item
  \emph{(Design matrix; full kernel versus tenure projection.)}
  \(\dim \Theta_\tau = \dim \mathcal{K}_\tau\). Let \(c\) be the number
  of connected components of the cohort--period incidence graph of
  \(\mathcal{S}\). Then \[
  \dim \mathcal{K} \;=\; \dim \mathcal{K}_\tau + (c - 1), \qquad \mathrm{rank}\,X_{\mathcal{S}} \;=\; (\text{number of parameters}) - \dim \Theta_\tau - (c - 1),
  \] so the full kernel and its tenure projection have the same
  dimension exactly when the incidence graph is connected (\(c = 1\)),
  as it is for every staggered trapezoid; the \(c - 1\) extra directions
  move \(\alpha\) and \(g\) against each other on the components not
  containing \(e_0\) and move \(\tau\) not at all. Under the
  tenure-coverage condition of (b), \(\dim \Theta_\tau \ge d\), so the
  rank is at most the number of parameters minus \(d\) minus \((c-1)\),
  with equality under C\(_d\). A design's identified space---its
  dimension and a basis of \(\mathcal{K}_\tau\)---can therefore be
  computed, before any data are collected, from \(X_{\mathcal{S}}\)
  alone; C\(_d\) is a sufficient condition for it to have the closed
  form in (c), not a prerequisite for computing it. C\(_d\)
  characterizes the tenure projection only: it does not imply that the
  incidence graph is connected, and results that need the full
  kernel---Proposition~\ref{prp-nc} and Proposition~\ref{prp-bounds}
  below---state connectedness as a separate condition.
\end{enumerate}

\end{theorem}

\begin{proof}
\emph{(a)} Two parameter vectors \(\theta, \theta'\) generate the same
cell means iff \(X_{\mathcal{S}}(\theta' - \theta) = 0\), i.e.,
\(\theta' - \theta \in \mathcal{K}\); restricting to the
\(\tau\)-coordinates gives \(\Theta_\tau = \tau + \mathcal{K}_\tau\). A
linear functional is constant on an affine family iff it annihilates the
family's direction space.

\emph{(b)} At any \((e,t) \in \mathcal{S}\) with \(s = t - e + 1\), \[
h_\alpha(t) + h_g(e) + h_\tau(s)
= m\bigl[-(t-e_0) + (e-e_0) + (s-1)\bigr] + \bigl[\rho(t - e + 1) - \rho(t - e_0 + 1)\bigr] = 0,
\] because \(-(t-e_0) + (e-e_0) + (t-e+1-1) = 0\) and because
\(e \equiv e_0 \pmod d\) for every \(e \in \mathcal{E}\) (\(d\) divides
\(e - e_0\)), so that \(\rho(t-e+1) = \rho(t-e_0+1)\). The
normalizations are respected: \(h_\tau(1) = \rho(1) = 0\) and
\(h_g(e_0) = 0\). The map \((m, \rho) \mapsto h\) is linear; it is
injective when \(d\) consecutive periods are observed, because \(h_g\)
at any \(e \ne e_0\) determines \(m\) and then
\(h_\alpha(t) = -m(t-e_0) - \rho(t-e_0+1)\) over \(d\) consecutive
periods determines the \(d\)-periodic \(\rho\); and the space of
\(d\)-periodic \(\rho\) with \(\rho(1) = 0\) has dimension \(d - 1\).
The two functional statements follow from (a) by taking \(\rho = 0\),
and \(m = 0\) with \(\rho\) the indicator of one residue class
\(r \not\equiv 1\).

\emph{(c)} Let \(h = (h_\alpha, h_g, h_\tau) \in \mathcal{K}\), so
\(h_\alpha(t) + h_g(e) + h_\tau(t-e+1) = 0\) on \(\mathcal{S}\),
\(h_\tau(1) = 0\), \(h_g(e_0) = 0\). Whenever \((e,t)\) and \((e,t+1)\)
are observed, differencing gives
\(\Delta h_\alpha(t) = -\Delta h_\tau(u)\) with \(u = t - e + 1\) and
\(\Delta f(x) := f(x+1) - f(x)\). The left side does not depend on
\(e\), so whenever two cohorts \(e \ne e'\) are observed at \(t\) and
\(t+1\) with increments \(u, u'\),
\(\Delta h_\tau(u) = \Delta h_\tau(u')\): the first difference of
\(h_\tau\) is constant along every edge of the increment graph, hence
constant on each of its components, hence, by C\(_d\)(ii), constant on
each residue class of \(\mathcal{U}\) modulo \(d\). Write
\(\Delta h_\tau(u) = c_{[u]}\) with \([u]\) the residue of \(u\), let
\(m := d^{-1}\sum_{r} c_r\) be the mean over one period (with
\(c_r := 0\) for any residue class not represented in \(\mathcal{U}\)),
and set \(\rho(s) := \sum_{u=1}^{s-1} c_{[u]} - m(s-1)\) for every
integer \(s \ge 1\), \(\rho(1) = 0\). Then
\(\Delta\rho(s) = c_{[s]} - m\) is \(d\)-periodic with zero mean over a
period, so \(\rho(s+d) - \rho(s) = 0\): \(\rho\) is \(d\)-periodic. By
increment-completeness, for every observed tenure \(s\) the increments
\(1, \dots, s-1\) are all observed, so
\(h_\tau(s) = h_\tau(1) + \sum_{u=1}^{s-1}\Delta h_\tau(u) = m(s-1) + \rho(s)\).
Back-substituting,
\(h_\alpha(t) + h_g(e) = -m(t-e) - \rho(t-e+1) = -m(t - e_0) + m(e - e_0) - \rho(t - e_0 + 1)\),
whose right side is a function of \(t\) plus a function of \(e\); hence
\(h_\alpha(t) = -m(t-e_0) - \rho(t-e_0+1) + \kappa\) and
\(h_g(e) = m(e-e_0) - \kappa\) on each connected component of the
cohort--period incidence graph, and \(h_g(e_0) = 0\) forces
\(\kappa = 0\) on the component containing \(e_0\). Thus \(h_\tau\) is
the tenure component of an element of \(\mathcal{K}_0\), which gives
\(\mathcal{K}_\tau = \{m(s-1) + \rho(s)\}\) and, with (a), the two
displayed statements; the functional characterization is the case of (a)
in which \(\mathcal{K}_\tau\) is spanned by \(s - 1\) and the
residue-class indicators \(r \not\equiv 1\). The \(d = 1\) clause is the
special case \(\rho \equiv 0\).

\emph{(d)} \(\dim \Theta_\tau = \dim \mathcal{K}_\tau\) by (a). The
kernel of the projection \(h \mapsto h_\tau\) restricted to
\(\mathcal{K}\) consists of the \(h\) with \(h_\tau = 0\),
i.e.~\(h_\alpha(t) + h_g(e) = 0\) on \(\mathcal{S}\) with
\(h_g(e_0) = 0\). Along any edge \((e,t)\) of the incidence graph
\(h_\alpha(t) = -h_g(e)\), so \(h_\alpha\) is constant on the periods of
each component and \(h_g\) is the negative of that constant on its
cohorts; the constant is free on each component except the one
containing \(e_0\), where \(h_g(e_0) = 0\) forces it to be zero. Hence
that kernel has dimension \(c - 1\), and
\(\dim \mathcal{K} = \dim \mathcal{K}_\tau + (c - 1)\) by rank--nullity
applied to the projection; the rank statement is rank--nullity applied
to \(X_{\mathcal{S}}\). The inequality is (b) and the equality under
C\(_d\) is (c). \(\square\)
\end{proof}

The theorem distinguishes the general kernel characterization, the
directions that are always unidentified, and the conditions under which
their tenure components exhaust the projected kernel. Part (a) is the
general statement: on a finite support, the identified set is the
translate of the kernel's tenure projection, and identification of a
linear functional is annihilation of that projection---a computation,
valid for any scheduled support including ones with gaps or unequal
follow-up. Part (b) says which directions are in the kernel of
\emph{every} design with the given entry cohorts: the affine direction,
and---when the cohorts share a stride---the periodic ones. Part (c) says
when those are \emph{all} there is in the tenure projection. The
distinction matters because C\(_d\) can fail on ordinary designs: with
cohorts entering at \(1\), \(4\), and \(8\) and observation through
period \(8\), the stride is \(d = 1\), yet the 14-cell design has 17
parameters and rank 14, so the identified set is three-dimensional, not
one-dimensional---the last cohort is observed only at entry, its
increments merge nothing, and the increment graph has the three
components \(\{1,4,7\}\), \(\{2,5\}\), \(\{3,6\}\). Part (a) covers such
a design exactly; part (c) does not apply to it. Part (d) separates the
tenure projection from the full kernel, and the two can differ even when
C\(_d\) holds: with cohorts entering at \(1\), \(4\), and \(6\), each
observed only at its entry period and the next
(\(\mathcal{S} = \{(e,e), (e,e+1)\}\)), the stride is \(d = 1\), the
only observed increment is \(u = 1\), and C\(_1\) holds, so
\(\mathcal{K}_\tau = \{m(s-1)\}\) is one-dimensional; but no two cohorts
share a period, the incidence graph has three components, and the full
kernel is three-dimensional (six cells, nine parameters, rank six). The
two extra directions move \(\alpha\) and \(g\) on the later cohorts
without touching \(\tau\). That distinction is harmless for what a panel
reveals about \(\tau\) and decisive for the recovery results of
Section~\ref{sec-recovery}, which combine the kernel with knowledge of
\(g\). The Japanese panel of Section~\ref{sec-jlps}, with cohorts
entering at waves \(1\), \(5\), and \(13\) and observed through wave
\(19\), has \(d = 4\), satisfies C\(_4\), is connected, and has a design
matrix with 41 cells, 39 free parameters, and rank 35: its identified
set is four-dimensional. The support \(e \in \{1,3\}\),
\(t \in \{3,4,5\}\) with the perturbation
\(h_\tau(s) = \mathbf{1}\{s \text{ even}\}\),
\(h_\alpha(t) = -\mathbf{1}\{t \text{ even}\}\), \(h_g \equiv 0\) is the
smallest example of a periodic direction: it preserves every cell mean
and both normalizations, is not affine, and has
\(\Delta^2 h_\tau(2) = -2\). Figure~\ref{fig-support} shows where the
periodic directions of the Japanese panel come from.

\begin{figure}

\centering{

\pandocbounded{\includegraphics[keepaspectratio]{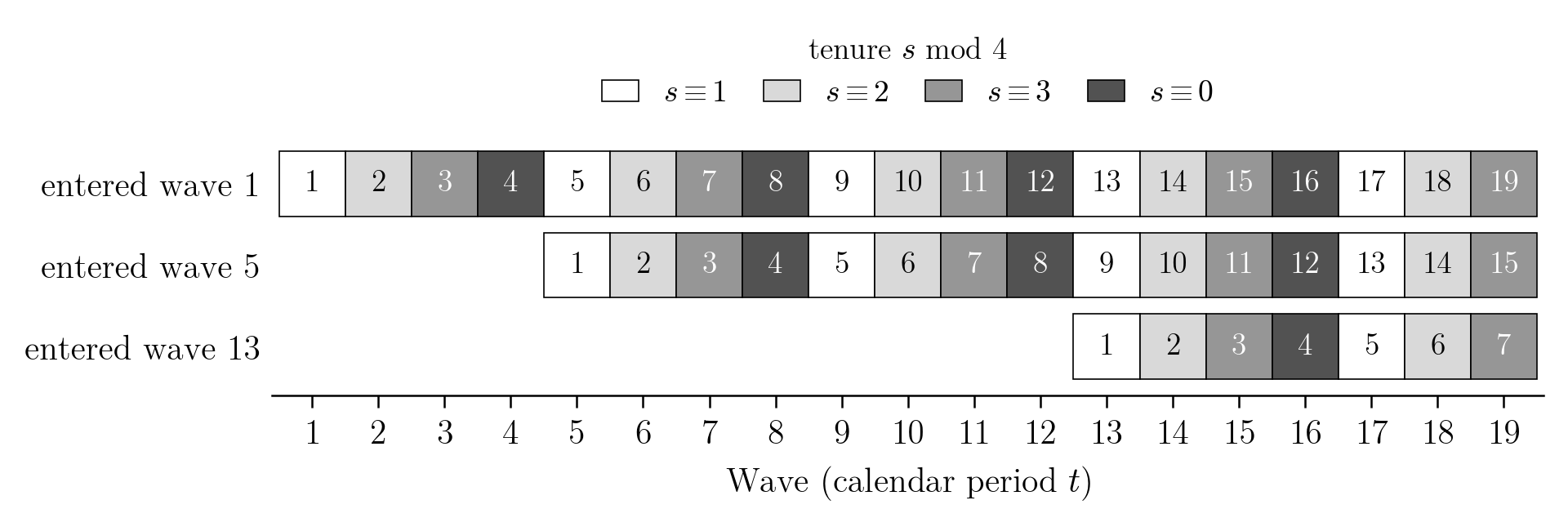}}

}

\caption{\label{fig-support}Observed support of the Japanese panel of
Section~\ref{sec-jlps}: cohorts entering at waves 1, 5, and 13, observed
through wave 19 (41 cells). Each cell shows the tenure \(s = t - e + 1\)
and is shaded by \(s \bmod 4\). Because every cohort enters at a wave
congruent to 1 modulo the stride \(d = 4\), the shade of a cell depends
only on its period, so each column has a single shade. A 4-periodic
function of tenure that vanishes at \(s = 1\) is therefore also a
function of the period and is absorbed by \(\alpha(t)\): these are the
three periodic directions of Theorem~\ref{thm-linear}(b), which together
with the affine direction make up the four-dimensional kernel of this
design (39 free parameters, rank 35).}

\end{figure}%

\emph{Antecedents.} Neither kind of direction is new to the
identification literature, and we position the theorem accordingly. The
affine direction is the exact translation into the measurement domain of
the age--period--cohort linear dependency
(\citeproc{ref-masetal1973}{Mason et al. 1973};
\citeproc{ref-fosse2019}{Fosse and Winship 2019a}) and of the
underidentification of fully dynamic event-study specifications without
never-treated units (\citeproc{ref-borusyak2024}{Borusyak, Jaravel, and
Spiess 2024}; \citeproc{ref-sunabraham2021}{Sun and Abraham 2021}), with
tenure in the role of age or event time. The periodic directions are the
measurement-domain form of the additional non-identifiability of APC
models with \emph{unequal intervals}. When age is recorded in intervals
\(M\) times as wide as periods, so that cohort is \(c = p - Ma\),
Gascoigne and Smith (\citeproc{ref-gascoignesmith2023}{2023, sec. 4.1},
eqs. (4)--(5)) show that an \(M\)-periodic function \(v_M\) can be added
to the period effect and subtracted from the cohort effect without
changing the linear predictor, and that the period and cohort
\emph{curvatures}---identified in the equal-interval case---are then
unidentified, a phenomenon they call the curvature identifiability
problem. Under the correspondence \(p = t\), \(Ma = e - e_0\),
\(c = s - 1\), and \(M = d\), their \(v_M\) is our \(\rho\) up to sign:
entry cohorts that share a stride are age groups recorded at a coarser
interval than periods, and our Theorem~\ref{thm-linear}(b) is their
equation (4) with tenure in the role of cohort. The problem is older
than that paper. Holford (\citeproc{ref-holford2006}{2006, sec. 3})
decomposes the age and period scales of an unequal-interval APC table
into macro-intervals of the least common multiple and micro-intervals
within them, and shows that the micro-effects carry the additional
non-identifiability, which he describes as cyclical at the least common
multiple and warns can be mistaken for a real cycle
(\citeproc{ref-holford2006}{Holford 2006, 984, 991}); he states that the
additional complexity of unequal intervals ``is well known'' and cites
Fienberg and Mason (\citeproc{ref-fienbergmason1979}{1979}) and Holford
(\citeproc{ref-holford1983}{1983}) for it
(\citeproc{ref-holford2006}{Holford 2006, 991}). Those two sources
support the attribution: Fienberg and Mason
(\citeproc{ref-fienbergmason1979}{1979, 37--40}) show that additional
linear dependencies arise when the intervals are unequal but evenly
spaced, and Holford (\citeproc{ref-holford1983}{1983, sec. 3},
pp.~318--320) derives the additional aliasing for balanced integer-ratio
intervals and illustrates the arbitrary saw-tooth pattern it permits
(\citeproc{ref-holford1983}{Holford 1983, sec. 3.3}, pp.~321--322).
Neither states the support-specific characterization or the
increment-graph condition of the present theorem. Smith and Wakefield
(\citeproc{ref-smithwakefield2016}{2016, sec. 4.2}, p.~598) write the
same ambiguity as an explicit indicator-function transformation: for any
reals \(v_1, \dots, v_M\),
\(\delta + \alpha_i + \beta_j + \gamma_k = \delta + \alpha_i + (\beta_j + v_m \mathbf{1}\{j \equiv m \ (\mathrm{mod}\ M)\}) + (\gamma_k - v_m \mathbf{1}\{k \equiv m \ (\mathrm{mod}\ M)\})\),
which is our Theorem~\ref{thm-linear}(b) with \(M = d\) and \(v\) in the
role of \(\rho\), and they too attribute the phenomenon to Holford
(\citeproc{ref-holford2006}{2006}). Gascoigne and Smith
(\citeproc{ref-gascoignesmith2023}{2023}) report a modulo-function
representation in Riebler and Held
(\citeproc{ref-rieblerheld2010}{2010}), which we cite on their authority
alone, not having consulted that article. Borusyak, Jaravel, and Spiess
(\citeproc{ref-borusyak2024}{2024, sec. 3.2}, fn. 15) note in passing
that additional collinearity arises in event studies when treatment is
staggered at periodic intervals. Two further points fix what is and is
not new here. The criterion that a linear functional is estimable in a
rank-deficient design exactly when it lies in the row space of the
design---our annihilation of \(\mathcal{K}_\tau\) in
Theorem~\ref{thm-linear}(a)---is classical; see Searle and Gruber
(\citeproc{ref-searlegruber2017}{2017, sec. 5.4}, ``Estimable
Functions''), where equation (37) states it as \(q' = t'X\) and Theorem
2 gives the equivalent test \(q'H = q'\) with \(H = G(X'X)\) and \(G\) a
generalized inverse of \(X'X\). Holford
(\citeproc{ref-holford2006}{2006, 985}) applies that test to
unequal-interval APC models, citing the first (1971) edition of the same
book for it. It is the foundation we use, not a result we claim. What
the present theorem adds is the support-specific characterization: which
vectors are in the kernel of every staggered design (b), the
increment-graph condition under which the affine-plus-periodic family
exhausts the tenure projection (c) with the trapezoid follow-up rule and
exact three-cohort dimension of Lemma~\ref{lem-trapezoid}, the
interrupted-schedule condition of Theorem~\ref{thm-dose}, and the
connection to regression in Section~\ref{sec-downstream}, which says
what the unidentified component does to fixed-effects and event-study
estimators under the stated rank conditions. We make no claim to have
discovered affine or periodic non-identification.

\begin{lemma}[Staggered
trapezoids]\protect\hypertarget{lem-trapezoid}{}\label{lem-trapezoid}

Let \(\mathcal{S} = \{(e,t): e \in \mathcal{E},\ e \le t \le T\}\) with
\(e_1 < \dots < e_K\), \(K \ge 2\), stride \(d\), and write
\(\Delta_2 := e_2 - e_1\) for the spacing between the first two cohorts,
\(w := T - e_K\) for the number of periods, beyond its entry period,
during which the last cohort is observed, and
\(\nu(\mathcal{S}) := \dim \Theta_\tau\).

\begin{enumerate}
\def\labelenumi{(\alph{enumi})}
\setcounter{enumi}{14}
\item
  \emph{(Component count.)} \(\nu(\mathcal{S})\) equals the number of
  connected components of the increment graph on
  \(\mathcal{U} = \{1, \dots, T - e_1\}\). Hence
  \(\nu(\mathcal{S}) \ge d\), with equality if and only if
  \(\mathcal{S}\) satisfies C\(_d\).
\item
  \emph{(Sufficient condition.)} If \(w \ge \Delta_2 - d\), then
  \(\mathcal{S}\) satisfies C\(_d\) and \(\nu(\mathcal{S}) = d\). In
  particular, when the first two cohorts are \(d\) apart, every
  \(T \ge e_K\) suffices.
\end{enumerate}

\begin{enumerate}
\def\labelenumi{(\roman{enumi})}
\setcounter{enumi}{1}
\item
  \emph{(Three cohorts.)} If \(K = 3\), then for every \(T \ge e_3\), \[
  \nu(\mathcal{S}) \;=\; \max\{\, d,\; \Delta_2 - w \,\},
  \] so the condition in (i) is necessary as well as sufficient: a third
  cohort observed for \(w\) periods removes exactly \(w\) dimensions
  from the two-cohort indeterminacy \(\Delta_2\), down to the floor
  \(d\).
\item
  \emph{(Two cohorts.)} If \(K = 2\),
  \(\nu(\mathcal{S}) = \Delta_2 = d\) for every \(T \ge e_2\).
\item
  \emph{(Upper bound for any number of cohorts.)} For every \(K \ge 3\)
  and every \(T \ge e_K\), \[
  \nu(\mathcal{S}) \;\le\; \max\{\, d,\; \Delta_2 - w \,\},
  \] so the three-cohort formula is an upper bound in general: extra
  cohorts can only shrink the identified set, and the follow-up of the
  \emph{last} cohort is what the bound is stated in.
\end{enumerate}

\end{lemma}

\begin{proof}
\emph{(o)} Let \(h \in \mathcal{K}\). The trapezoid is
increment-complete, and the first cohort is observed at every
\(t \in [e_1, T]\), so \(h_\tau\) is determined on the whole observed
tenure range by \(h_\tau(1) = 0\) and the first differences
\(\Delta h_\tau(u)\), \(u \in \mathcal{U}\), and by the proof of
Theorem~\ref{thm-linear}(c) these differences are constant on the
components of the increment graph. Conversely, let \(\Delta h_\tau\) be
any function on \(\mathcal{U}\) that is constant on components, define
\(h_\tau\) by summation, and set \(f(t,e) := -h_\tau(t - e + 1)\). The
system \(h_\alpha(t) + h_g(e) = f(t,e)\) on \(\mathcal{S}\) is solvable:
put \(h_g(e) := f(T, e) - f(T, e_1)\) and \(h_\alpha(t) := f(t, e_1)\);
the residual
\(f(t,e) - h_\alpha(t) - h_g(e) = f(t,e) - f(t,e_1) - f(T,e) + f(T,e_1)\)
is a rectangle contrast between cohorts \(e\) and \(e_1\) over periods
\(t\) and \(T\), both cohorts being observed throughout \([t, T]\), so
it is the telescoping sum over \(t' \in [t, T-1]\) of the
consecutive-period contrasts
\(\Delta f(t', e) - \Delta f(t', e_1) = \Delta h_\tau(t' - e_1 + 1) - \Delta h_\tau(t' - e + 1)\),
each of which vanishes because \(t' - e + 1\) and \(t' - e_1 + 1\) are
joined by an edge. Hence every assignment of one free value per
component extends to an element of \(\mathcal{K}\), uniquely once
\(h_g(e_0) = 0\) is imposed, and \(\nu(\mathcal{S})\) is the number of
components. Components lie inside residue classes modulo \(d\), so there
are at least \(d\) of them when every class is represented in
\(\mathcal{U}\) (which holds since
\(\mathcal{U} \supseteq \{1, \dots, \Delta_2\}\) and
\(d \le \Delta_2\)), with equality iff the components are the residue
classes, which is C\(_d\)(ii).

\emph{Seeds.} The pair \((e_1, e_k)\) contributes the edges
\(u \sim u + \Delta_k\), \(\Delta_k := e_k - e_1\), for every
\(u \in [1, T - e_k]\); edges between two later cohorts \(e_j < e_k\) at
the same \(t\) join \(t - e_k + 1\) and \(t - e_j + 1\), both of which
are joined to \(t - e_1 + 1\), so they add no components. The pair
\((e_1, e_2)\) joins \(u\) to \(u + \Delta_2\) for all
\(u \le T - e_2\), i.e., whenever both lie in \(\mathcal{U}\); hence
every \(u \in \mathcal{U}\) is in the component of its \emph{seed}
\(x \in [1, \Delta_2]\), \(x \equiv u \pmod{\Delta_2}\), and
\(\nu(\mathcal{S})\) is the number of components of the graph induced on
the seeds. With \(K = 2\) the seeds are pairwise disconnected and
\(\nu = \Delta_2 = d\): (iii). For \(K \ge 3\), the pair \((e_1, e_k)\)
joins the seed of \(u\) to the seed of \(u + \Delta_k\) for each
\(u \le T - e_k\); on the seed set
\(\mathbb{Z}_{\Delta_2} = \{1, \dots, \Delta_2\}\) this is the merge
\(x \sim x + \Delta_k \pmod{\Delta_2}\), available for every seed
\(x \le T - e_k\), hence for every \(x \le w\) (since
\(T - e_k \ge T - e_K = w\)).

\emph{(i).} The subgroup of \(\mathbb{Z}_{\Delta_2}\) generated by
\(\{\Delta_k \bmod \Delta_2 : k \ge 2\}\) is \(d\mathbb{Z}_{\Delta_2}\),
because \(\gcd(\Delta_2, \dots, \Delta_K) = d\); its cosets are the
residue classes modulo \(d\), each of size \(\Delta_2 / d\). Consider
one class \(C\) and the directed graph on \(C\) with edges
\(x \to x + \Delta_k\) for every \(k\) and every \(x \in C\) with
\(x \le w\); vertices \(x > w\) have no out-edges. If
\(w \ge \Delta_2 - d\), the interval \((w, \Delta_2]\) has length at
most \(d\) and therefore contains at most one element of \(C\): at most
one vertex of \(C\) lacks out-edges. Let \(A\) be a weakly connected
component of this graph. If \(A\) contains no vertex without out-edges,
then \(A\) is closed under \(x \mapsto x + \Delta_k\) for every \(k\),
hence under addition by the subgroup they generate, hence \(A = C\). So
every component contains the (at most one) vertex without out-edges, and
there is exactly one component: all seeds in \(C\) are merged. This
holds for every class, so the seeds collapse to \(d\) classes, the
components of the increment graph are the residue classes, C\(_d\)
holds, and \(\nu = d\) by (o). If \(\Delta_2 = d\) the condition
\(w \ge 0\) is vacuous.

\emph{(ii).} With \(K = 3\) the only merges are
\(x \sim x + \Delta_3 \pmod{\Delta_2}\) for \(x \in [1, w]\): at most
\(w\) of them, so \(\nu \ge \Delta_2 - w\), and \(\nu \ge d\) always by
(o). If \(w \ge \Delta_2 - d\), (i) gives
\(\nu = d = \max\{d, \Delta_2 - w\}\). If \(w < \Delta_2 - d\), the
\(w\) merges are independent: a cycle among them would require every
source in some orbit of \(x \mapsto x + \Delta_3\) on
\(\mathbb{Z}_{\Delta_2}\)---an orbit is a full residue class modulo
\(d\) within \([1, \Delta_2]\), whose largest element is at least
\(\Delta_2 - d + 1 > w\)---to lie in \([1, w]\), which is impossible.
Hence \(\nu = \Delta_2 - w = \max\{d, \Delta_2 - w\}\).

\emph{(iv).} Every generator \(\Delta_k\), \(k \ge 3\), is available at
every seed \(x \le T - e_k\), hence at every seed \(x \le w\), since
\(T - e_k \ge T - e_K = w\). Fix a residue class \(C\) modulo \(d\)
within \(\{1, \dots, \Delta_2\}\), let \(A := C \cap [1, w]\) be its
seeds with out-edges and \(B := C \setminus A\) those without. If a
connected component \(P\) of the seed graph restricted to \(C\) contains
no vertex of \(B\), then every \(x \in P\) has all its out-edges, so
\(P\) is closed under \(x \mapsto x + \Delta_k\) for every \(k\), hence
under the subgroup these generate---which is \(d\mathbb{Z}_{\Delta_2}\),
acting transitively on \(C\)---so \(P = C\) and \(B = \emptyset\).
Therefore, when \(B \ne \emptyset\), every component meets \(B\) and
distinct components meet disjoint parts of it: the number of components
of \(C\) is at most \(|B| = |C| - |A|\), and it is \(1\) when
\(B = \emptyset\). Summing over the \(d\) classes,
\(\nu \le \sum_C \max\{1,\ |C| - |A_C|\}\). The classes interleave, so
the \(|A_C|\) differ by at most one. If some class has \(B = \emptyset\)
then every class has \(|C| - |A_C| \le 1\), so \(\Delta_2 - w \le d\)
and the sum is at most \(d\); otherwise every term is \(|C| - |A_C|\)
and the sum is exactly \(\Delta_2 - w\). Either way
\(\nu \le \max\{d, \Delta_2 - w\}\). \(\square\)
\end{proof}

Three remarks. First, (o) turns C\(_d\) into a computation for any
trapezoid---count the components of the increment graph---and (i) into a
design rule: following the last refreshment cohort for at least
\(\Delta_2 - d\) periods beyond its entry is \emph{sufficient} for
C\(_d\); it is necessary only in the three-cohort case, where (ii) makes
the count exact. In the Japanese panel of Section~\ref{sec-jlps},
\(\Delta_2 = 4 = d\) and (i) holds for every \(T\); had the second
cohort entered at wave 6 rather than 5 (\(\Delta_2 = 5\), \(d = 1\) with
the third at 13), the last cohort would have needed \(w \ge 4\) periods
of follow-up before the identified set shrank to the affine line, and
with \(T = 16\) (\(w = 3\)) the identified set is two-dimensional (31
cells, 33 parameters, rank 31). The design with entries at \(1, 4, 8\)
and \(T = 8\) discussed after Theorem~\ref{thm-linear} has
\(w = 0 < \Delta_2 - d = 2\) and \(\nu = 3\), as (ii) gives. Second,
(iv) says that the three-cohort formula bounds every design: a fourth or
fifth cohort can only merge more seeds, never fewer. The bound is not
tight---in the sweep reported in the replication package (all entry sets
\(\mathcal{E} \subset \{1, \dots, 13\}\) with \(3 \le K \le 5\) and
\(T \le e_K + 13\); 10,934 designs) it is attained in \(9{,}746\) of
them and strict in \(1{,}188\), while (o) holds in every case---so the
component count, or the rank of \(X_{\mathcal{S}}\), remains the
operational check. Stating the bound in terms of the \emph{last}
cohort's follow-up is not an artifact: with entries at
\(1, 5, 9, 11, 12\) and \(T = 12\) the third cohort is followed for
three periods and the bound computed from it would be \(1\), while
\(\nu = 3\), because the spacing \(e_3 - e_1 = 8\) is a multiple of
\(\Delta_2 = 4\) and merges nothing. Third, the replication package
provides an R function, \texttt{design\_rank()}, that builds
\(X_{\mathcal{S}}\) for any support (trapezoid or not), reports its
rank, nullity, stride, and increment-graph components, checks the
condition of (i), returns a basis of \(\mathcal{K}_\tau\), and tests
whether a given linear functional of \(\tau\) is identified.

\begin{remark}
\textbf{Normalizations in practice.} Because levels of \(\tau\) are
never identified, every published level of panel conditioning or
rotation-group bias rests on a normalization, and the practices of the
literature can be read as such: the rotation-group index of the Current
Population Survey, defined in the survey's own documentation as ``the
ratio of the {[}estimate{]} based on a particular month-in-sample group
to the average estimate from all eight month-in-sample groups combined,
multiplied by 100'' (quoted in
\citeproc{ref-halpernmanners2012}{Halpern-Manners and Warren 2012, fig.
1}; the same index is defined and tabulated by
\citeproc{ref-bailar1975}{Bailar 1975}), which sets the average of the
eight month-in-sample effects to a constant; Krueger, Mas, and Niu
(\citeproc{ref-krueger2017}{2017})'s indexation of rotation-group bias
to the full-sample mean; the reference-group restriction in the
structural time-series model of van den Brakel and Krieg
(\citeproc{ref-brakelkrieg2015}{2015}), whose rotation-group bias varies
over time and is measured relative to one group; the ``average
multiplicative bias by month in sample \ldots{} relative to average
second-stage estimates'' tabulated by McIllece
(\citeproc{ref-mcillece2022}{2022}, Table 2), which normalizes instead
by the CPS second-stage estimate, a survey estimate computed from all
eight rotation groups, assumed approximately unbiased and calibrated to
external population controls (\citeproc{ref-mcillece2022}{McIllece 2022,
secs. 2.1, 3.1--3.2}); and the omission of two reference categories in
fully dynamic event studies (\citeproc{ref-borusyak2024}{Borusyak,
Jaravel, and Spiess 2024}). We stress that these are analogies between
\emph{normalizations}, not equivalences between \emph{models}: the
models of van den Brakel and Krieg
(\citeproc{ref-brakelkrieg2015}{2015}) and McIllece
(\citeproc{ref-mcillece2022}{2022}) are not the time-invariant additive
model M4, their estimands differ from \(\tau\), and a single restriction
of the kind each imposes removes one dimension of \(\Theta_\tau\), which
is all of it only when \(d = 1\). The observation that disagreements
among published level estimates partly reflect differing conventions is,
in our model, a corollary of Theorem~\ref{thm-linear}; it is not a
demonstration that those disagreements are exactly movements along
\(\Theta_\tau\), which would require the sources to share M4. A
different route is to change the model class rather than to normalize:
Feng, Hu, and Sun (\citeproc{ref-fenghusun2022}{2022}) point-identify
misclassification probabilities in a latent-state model of labor-force
status with history-dependent errors by restricting response dynamics to
first-order dependence, a model in which the object identified is not a
mean path on tenure; we do not claim that their restrictions select a
point of \(\Theta_\tau\), only that the two approaches respond to the
same dependency between tenure, period, and cohort.
\end{remark}

\begin{corollary}[What is identified: curvature at the design's
stride]\protect\hypertarget{cor-curvature}{}\label{cor-curvature}

On any support, a linear functional \(\lambda'\tau\) is point identified
if and only if \(\lambda\) annihilates \(\mathcal{K}_\tau\)
(Theorem~\ref{thm-linear}(a)). Under C\(_d\) this is the pair of
conditions \[
\sum_s \lambda_s\,(s-1) = 0 \qquad\text{and}\qquad \sum_{s \equiv r \ (\mathrm{mod}\ d)} \lambda_s = 0 \ \text{ for every residue } r \not\equiv 1 ,
\] and in particular the \emph{lag-\(d\) second differences} \[
\Delta_d^2\tau(s) \;:=\; \tau(s + d) - 2\,\tau(s) + \tau(s - d)
\] are identified for every \(s\) with \(s - d\) and \(s + d\) in the
support (second differences are \emph{centered} at \(s\) throughout the
paper, and \(\Delta^2 := \Delta_1^2\)); when \(d = 1\) these are the
ordinary second differences, and all curvature of \(\tau\) is
identified. When \(d > 1\), ordinary second differences
\(\Delta^2\tau(s)\) are \emph{not} identified, and neither are the
increments \(\tau(s+d) - \tau(s)\) between refreshment tenures. The
lag-\(d\) second differences are the natural examples, not the whole
identified space: that space has dimension equal to the number of free
conditioning coordinates minus \(\dim \Theta_\tau\), and contains
cross-class contrasts as well. In the design of Section~\ref{sec-jlps}
(\(d = 4\), tenures \(1\) to \(19\)) it is \(18 - 4 = 14\)-dimensional;
the eleven lag-4 second differences span a proper subspace, and
\(\Delta^2\tau(s) - \Delta^2\tau(s+4)\) for \(s = 2, 3, 4\) complete a
basis (Appendix A.2 uses the case \(s = 2\)). The displayed pair of
conditions is the complete criterion under C\(_d\). On a support that
violates C\(_d\) the identified functionals are fewer still, and must be
read off \(\mathcal{K}_\tau\). Zero identified curvature is consistent
with an affine path (Theorem~\ref{thm-linear}(b)), so identified
curvature can refute nonlinear shapes of conditioning but cannot by
itself establish that conditioning has stopped; see
Proposition~\ref{prp-plateau}.

\end{corollary}

\begin{proof}
The first sentence is Theorem~\ref{thm-linear}(a). Under C\(_d\),
\(\mathcal{K}_\tau\) is spanned by \(s - 1\) and the residue-class
indicators \(\mathbf{1}\{s \equiv r\}\), \(r \not\equiv 1\)
(Theorem~\ref{thm-linear}(c)), and annihilating that basis is the
displayed pair. The weights of \(\Delta_d^2\tau(s)\) are \((1, -2, 1)\)
at the centered positions \(s-d, s, s+d\), which lie in one residue
class and sum to zero there, and
\(\sum \lambda_s (s-1) = (s-d-1) - 2(s-1) + (s+d-1) = 0\). For \(d > 1\)
the weights of \(\Delta^2\tau(s)\), at \(s-1, s, s+1\), fall into more
than one residue class---two classes when \(d = 2\) (the class of
\(s \pm 1\), with sum \(2\), and that of \(s\), with sum \(-2\)) and
three when \(d \ge 3\)---and at least one class other than that of
tenure one has a nonzero sum; those of \(\tau(s+d) - \tau(s)\) have
\(\sum\lambda_s(s-1) = d \ne 0\). \(\square\)
\end{proof}

A stride-one design satisfying C\(_1\) identifies changes in the slope
of the conditioning path at unit resolution, so the stylized fact that
``conditioning happens in the first few waves'' has a normalization-free
content there---that the curvature is concentrated in the first few
waves. That conditioning \emph{stops}, however, is not identified from
curvature alone, because zero curvature is compatible with a nonzero
linear trend; it requires an additional restriction such as the plateau
of Proposition~\ref{prp-plateau}. Adjacent entry cohorts with enough
follow-up are the leading case, by Lemma~\ref{lem-trapezoid}(i), but
stride one is what matters, not adjacency; in a panel refreshed at a
stride \(d > 1\), the shape of conditioning is identified only on the
coarser grid of that stride, and every published \emph{level} is a
normalization in every design. This is the sense in which a refreshment
schedule is an identification decision.

\subsection{Interrupted participation: rotation
designs}\label{sec-interrupted}

M1 assumes uninterrupted participation, so that the number of interviews
a unit has received coincides with its calendar tenure. Rotation designs
violate it by construction. The U.S. Current Population Survey
interviews a household for four consecutive months, rests it for eight,
and returns it for four more---the 4--8--4 pattern whose rotation-group
differences are the subject of the literature this paper began with
(\citeproc{ref-bailar1975}{Bailar 1975}; \citeproc{ref-solon1986}{Solon
1986}; \citeproc{ref-krueger2017}{Krueger, Mas, and Niu 2017};
\citeproc{ref-halpernmanners2012}{Halpern-Manners and Warren 2012})---so
that its eighth interview arrives in its sixteenth month. The relaxation
is worth making because the two clocks it separates are exactly the two
sides of Theorem~\ref{thm-linear}: conditioning is a function of the
\emph{dose}, the number of interviews received, while the linear
dependency that makes the path unidentified is a statement about
\emph{calendar} time. When the two coincide, the dependency is
unbreakable. When they do not, it is.

Let the entry cohorts be \(\mathcal{E}\) with stride \(d\) as before,
and let every cohort follow a common \emph{interview pattern}
\(J = \{\,j_1 = 0 < j_2 < \dots < j_K\,\}\) of offsets from entry:
cohort \(e\) is interviewed at the calendar times \(e + j_k\) that lie
in the field period, and the interview at \(e + j_k\) is its \(k\)-th,
so its dose is \(k\). Write \(j(k) := j_k\) for the elapsed calendar
time at which the \(k\)-th interview arrives. The cell means are
\begin{equation}\phantomsection\label{eq-dose}{
\mu\bigl(e,\, e + j_k\bigr) \;=\; \alpha(e + j_k) \;+\; g(e) \;+\; \tau(k), \qquad \tau(1) = 0,\ g(e_0) = 0,
}\end{equation} which is M4 with the dose in place of the tenure;
uninterrupted participation is the case \(J = \{0, 1, \dots, K-1\}\),
where \(j(k) = k - 1\) and Equation~\ref{eq-dose} is
Equation~\ref{eq-model}. We assume:

\begin{itemize}
\tightlist
\item
  \textbf{(P\('\)) Continuous recruitment.} The entry cohorts form an
  arithmetic progression
  \(\mathcal{E} = \{e_0, e_0 + d, \dots, e_0 + Ld\}\), every scheduled
  cell \((e, e + j_k)\) with \(e \in \mathcal{E}\), \(j_k \in J\) lies
  in the field period, the stride is itself a spacing of the pattern
  (\(d = j' - j\) for some \(j, j' \in J\)), and the entry window covers
  the pattern's largest within-class gap, \[
  L\,d \;\ge\; \Gamma(J, d) \;:=\; \max\bigl\{\, j' - j \;:\; j < j' \text{ consecutive in } J \cap (j + d\mathbb{Z}) \,\bigr\}.
  \] For the CPS pattern with monthly entry, \(\Gamma = 12 - 3 = 9\), so
  (P\('\)) asks for ten consecutive monthly cohorts.
\end{itemize}

An earlier version of this section stated a weaker condition---that
every spacing in \(d\mathbb{Z} \cap [1, j_K]\) be \emph{realized} by
some pair of cohorts observed at a common period---and claimed sharpness
under it. That claim is false: with cohorts \(\{1, 2, 4\}\), offsets
\(J = \{0, 1, 3\}\) and \(T = 7\), every spacing \(1, 2, 3\) is realized
and the incidence graph is connected, yet the design has ten columns,
rank eight, and a \emph{two}-dimensional conditioning kernel containing
\(h_\tau = (0, 0, 1)\), which is not a multiple of \(j = (0, 1, 3)\).
The gap in the cohort set is what breaks the argument: realized spacings
constrain the cohort effects only along the pairs that realize them, and
with no cohort at \(3\) nothing forces \(h_g(4) - h_g(2)\) to equal
twice \(h_g(2) - h_g(1)\). (P\('\)) removes the gaps. It is sufficient
and not necessary; the exact criterion for any support is the projected
kernel of Theorem~\ref{thm-linear}(a), which the replication package
computes for arbitrary patterns and cohort sets.

\begin{theorem}[Unidentified directions under a common interrupted
interview schedule]\protect\hypertarget{thm-dose}{}\label{thm-dose}

Under M2--M3, Equation~\ref{eq-dose}, and a common interview pattern
\(J\):

\begin{enumerate}
\def\labelenumi{(\alph{enumi})}
\item
  \emph{(Directions that are never identified.)} For every
  \(m \in \mathbb{R}\) and every \(d\)-periodic \(\rho\) with
  \(\rho(0) = 0\), the vector \(h_\tau(k) = m\,j(k) + \rho(j(k))\),
  \(h_g(e) = m\,(e - e_0)\),
  \(h_\alpha(t) = -m\,(t - e_0) - \rho(t - e_0)\) lies in the kernel of
  the cell design.
\item
  \emph{(Sharpness under continuous recruitment.)} Under (P\('\)) their
  tenure components exhaust the tenure projection of the kernel: \[
  \Theta_\tau \;=\; \bigl\{\, k \mapsto \tau(k) + m\,j(k) + \rho(j(k)) \;:\; m \in \mathbb{R},\ \rho \text{ $d$-periodic},\ \rho(0) = 0 \,\bigr\},
  \] whose dimension is the number of residue classes modulo \(d\)
  represented in \(J\)---one when \(d = 1\). Without (P\('\)) the
  identified set can be strictly larger; its dimension is then the rank
  of the conditioning projection of \(\ker X_{\mathcal{S}}\), computable
  for any support.
\item
  \emph{(Identified functionals.)} Under (P\('\)), \(\lambda'\tau\) is
  identified if and only if \(\sum_k \lambda_k\, j(k) = 0\) and, for
  \(d > 1\), \(\sum_{k:\, j(k) \equiv r} \lambda_k = 0\) for every
  residue \(r \not\equiv 0\) represented in \(J\). When \(d = 1\),
  writing \(T(j) := \tau(k(j))\) for the path as a function of elapsed
  calendar time, the identified functionals are exactly the linear
  combinations of second-order divided differences of \(T\) on
  \(J\)---the curvature of conditioning against the \emph{calendar}, not
  against the dose. For \(d > 1\) the identified functionals are those
  with \(\sum_k \lambda_k j(k) = 0\) and zero weight on every class
  \(r \not\equiv 0\); these include the second-order divided differences
  taken within a residue class, but are not spanned by them---with
  \(J = \{0,1,2,3\}\) and \(d = 2\) no class has three offsets, and the
  one-dimensional identified space is spanned by the cross-class
  contrast \([T(2) - T(0)] - [T(3) - T(1)]\), a difference of
  within-class per-period rates.
\item
  \emph{(Constant-rate paths.)} Under (P\('\)), a nonzero path that is
  linear in the dose, \(\tau(k) = c\,(k-1)\) with \(c \ne 0\), belongs
  to the identified set of the null path---and is therefore
  indistinguishable from no conditioning at all---if and only if
  \(k - 1 = m\,j(k) + \rho(j(k))\) for some such \((m, \rho)\); with
  \(d = 1\) this says precisely that the interviews are \emph{equally
  spaced}. (The zero path, \(c = 0\), is the null path itself and
  belongs to its identified set on every schedule.) The statement
  concerns constant-rate paths only. Under one common schedule \(k\) and
  \(j(k)\) are in one-to-one correspondence on the observed support, so
  an \emph{unrestricted} function of the interview count can always be
  rewritten as a function of elapsed time and conversely; what an
  unequally spaced design separates is a constant increment per
  interview from a constant increment per unit of calendar time, not the
  two mechanisms in general.
\end{enumerate}

\end{theorem}

\begin{proof}
\emph{(a)} At the cell \((e, t) = (e, e + j_k)\) we have
\(j(k) = t - e\), so \(m\,j(k) = m(t - e_0) - m(e - e_0)\) cancels
against \(h_\alpha\) and \(h_g\); and
\(\rho(j(k)) = \rho(t - e) = \rho(t - e_0)\) because \(d\) divides
\(e - e_0\), which cancels against the second term of \(h_\alpha\). The
normalizations hold since \(\rho(0) = 0\).

\emph{(b)} Let \(h\) be in the kernel and put \(H(j) := h_\tau(k(j))\)
for \(j \in J\), so \(H(0) = h_\tau(1) = 0\). If cohorts \(e\) and
\(e - \delta\) are both observed at a period \(t\), their offsets are
\(j = t - e\) and \(j + \delta\), both in \(J\), and subtracting the two
cell equations gives \begin{equation}\phantomsection\label{eq-pair}{
h_g(e) - h_g(e - \delta) \;=\; H(j + \delta) - H(j) .
}\end{equation} The right-hand side is a function of the offsets alone.
Two facts follow, and (P\('\)) supplies the cohorts and cells that each
requires. \emph{First, \(h_g\) is affine on \(\mathcal{E}\).} Because
\(d\) is a spacing of the pattern, there are \(j, j + d \in J\); for
every \(e \in \mathcal{E}\) other than \(e_0\), the cohorts \(e\) and
\(e - d\) are both in the sample and both observed at \(t = e + j\)
(offsets \(j\) and \(j + d\)), so Equation~\ref{eq-pair} gives
\(h_g(e) - h_g(e - d) = H(j + d) - H(j) =: G\), the same number for
every \(e\). Hence \(h_g(e_0 + \ell d) = \ell G\), that is
\(h_g(e) = m(e - e_0)\) with \(m := G/d\). \emph{Second, \(H\) is affine
on each residue class.} Let \(j < j'\) be consecutive members of \(J\)
in the same class modulo \(d\), so that
\(\delta := j' - j \in d\mathbb{Z}\) and
\(\delta \le \Gamma(J, d) \le Ld\). Then \(e_0\) and \(e_0 + \delta\)
are both in the sample, both observed at \(t = e_0 + j'\) (offsets
\(j'\) and \(j\)), and Equation~\ref{eq-pair} with the first fact gives
\(H(j') - H(j) = h_g(e_0 + \delta) - h_g(e_0) = m\,\delta = m(j' - j)\).
Chaining along each class, \(H(j) - mj\) is constant on the class;
writing \(\rho\) for that constant as a function of the
class---\(\rho(0) = 0\) by \(H(0) = 0\)---gives \(H(j) = mj + \rho(j)\),
which is (a). The parameters are free and the representation is
injective: if \(mj + \rho(j) = 0\) on \(J\), the pair \(j, j + d \in J\)
lies in one class and gives \(md = 0\), so \(m = 0\) and then every
\(\rho(r) = 0\); hence the dimension is the number of represented
classes, the class of \(0\) contributing \(m\) and each other class its
constant. On residues not represented in \(J\) the periodic extension of
\(\rho\) is arbitrary, but such residues never occur as \(t - e_0\) at
an observed cell, so the kernel vector does not depend on that choice.
Back-substituting,
\(h_\alpha(t) = -h_g(e) - H(t - e) = -m(t - e_0) - \rho(t - e_0)\) on
every cell, consistently across cohorts, which fixes \(h_\alpha\); the
incidence graph is connected under (P\('\)), so nothing else is free.

The earlier argument failed at the first fact. It read
Equation~\ref{eq-pair} as defining an increment \(G(\delta)\) that
depends on the spacing alone and is additive in \(\delta\); but
Equation~\ref{eq-pair} constrains \(h_g\) only along the pairs that are
actually observed together, and additivity across spacings needs the
intermediate cohorts to exist. In the counterexample above the pairs
\((1,2)\), \((2,4)\) and \((1,4)\) realize the three spacings, but the
relation \(h_g(4) - h_g(1) = [h_g(4) - h_g(2)] + [h_g(2) - h_g(1)]\) is
an identity, not a constraint, and nothing ties \(h_g(4) - h_g(2)\) to
\(2[h_g(2) - h_g(1)]\).

\emph{(c)} is Theorem~\ref{thm-linear}(a) with the kernel of (b): the
functional is constant on \(\Theta_\tau\) iff it annihilates
\(j(\cdot)\) and each class indicator other than that of \(0\) (the
class of \(0\) carries no free constant, since \(\rho(0) = 0\)). For
\(d = 1\) this is the single condition \(\sum_k \lambda_k j(k) = 0\);
because \(\tau(1) = 0\), the weight on \(k = 1\) is immaterial and can
be chosen to make the weights sum to zero, after which the functionals
annihilating \(\{1, j\}\) on the finite set \(J\) are exactly the linear
combinations of second-order divided differences of \(T\). For \(d > 1\)
the class conditions are additional and the space is as described in the
statement.

\emph{(d)} is (b) read as a membership question: \(c(k-1)\) lies in the
identified set of the null path iff \(c(k-1) = m\,j(k) + \rho(j(k))\)
for some \((m,\rho)\), and for \(c \ne 0\) dividing by \(c\) gives the
displayed condition; for \(d = 1\), \(k - 1 = m\,j(k)\) for all \(k\)
holds iff \(j(k)\) is a constant multiple of \(k-1\), i.e.~iff the
offsets are equally spaced. \(\square\)
\end{proof}

\begin{corollary}[What a 4--8--4 rotation
identifies]\protect\hypertarget{cor-cps}{}\label{cor-cps}

Take the CPS pattern \(J = \{0,1,2,3,12,13,14,15\}\) with monthly entry
(\(d = 1\)), an entry window of at least ten consecutive cohorts, and a
field period covering every scheduled interview, so that (P\('\)) holds
(\(\Gamma = 9\)). Then \(\Theta_\tau\) is one-dimensional,
\(\{\tau + m\,j(\cdot)\}\), and:

\begin{enumerate}
\def\labelenumi{(\roman{enumi})}
\item
  no level \(\tau(k)\), \(k \ge 2\), and no increment
  \(\tau(k+1) - \tau(k)\) is identified;
\item
  every difference of \emph{per-month} rates is identified: for any two
  interview intervals, \[
  \frac{\tau(k+1) - \tau(k)}{j(k+1) - j(k)} \;-\; \frac{\tau(l+1) - \tau(l)}{j(l+1) - j(l)}
  \] is a point-identified functional. In particular \[
  \mathcal{D} \;:=\; \bigl[\tau(5) - \tau(4)\bigr] \;-\; 9\,\bigl[\tau(2) - \tau(1)\bigr]
  \] is identified. It is identified \emph{from the cohort-by-period
  array of cell means}, in which the cohort effects \(g\) are estimated
  jointly with \(\alpha\) and \(\tau\); it is \textbf{not} in general
  the same weights applied to month-in-sample averages taken across
  periods, because those averages retain the cohort effects. If
  \(\tau \equiv 0\) and \(g(e) = a(e^2 - e_0^2)\), the contrast of
  month-in-sample means at any period equals \(126a\) while
  \(\mathcal{D} = 0\), and the discrepancy is constant across periods,
  so averaging does not remove it. Only the affine part of \(g\) cancels
  under the weights of \(\mathcal{D}\);
\item
  \(\mathcal{D} = 0\) for every path that accumulates at a constant rate
  per calendar month, and \(\mathcal{D} = -8c\) for a path that
  accumulates at a constant rate \(c\) per interview. Under (P\('\)),
  and given estimates from the joint array, the 4--8--4 design can
  therefore distinguish these two \emph{constant-rate} hypotheses, which
  no equally spaced design can; it does not distinguish an unrestricted
  function of interviews from an unrestricted function of time, since on
  a common schedule the two are reparametrizations of each other
  (Theorem~\ref{thm-dose}(d)).
\end{enumerate}

\end{corollary}

\begin{proof}
\leavevmode

\begin{enumerate}
\def\labelenumi{(\roman{enumi})}
\tightlist
\item
  and (ii) are Theorem~\ref{thm-dose}(c): the weights of a level or of a
  single increment do not annihilate \(j(\cdot)\), while those of a
  difference of per-month rates do, since each rate contributes
  \(\{j(k+1) - j(k)\}/\{j(k+1) - j(k)\} = 1\). For \(\mathcal{D}\),
  \(\sum_k \lambda_k j(k) = (12 - 3) - 9(1 - 0) = 0\). For the second
  claim of (ii), the month-in-sample-\(k\) mean at period \(t\) is
  \(\alpha(t) + g(t - j_k) + \tau(k)\); applying the weights of
  \(\mathcal{D}\) gives
  \(\mathcal{D} + [g(t-12) - g(t-3)] - 9[g(t-1) - g(t)]\), and with
  \(g(e) = a(e^2 - e_0^2)\) the bracketed cohort term is
  \(a(-18t + 135) - a(-18t + 9) = 126a\), the constant \(ae_0^2\)
  cancelling because the weights sum to zero. (iii) Substitute
  \(\tau(k) = a + b\,j(k)\), which gives \(b(12-3) - 9b(1-0) = 0\); and
  \(\tau(k) = c(k-1)\), which gives \(c(4-3) - 9c(1-0) = -8c\).
  \(\square\)
\end{enumerate}

\end{proof}

Figure~\ref{fig-cps} plots the two constant-rate paths of part (iii).

\begin{figure}

\centering{

\pandocbounded{\includegraphics[keepaspectratio]{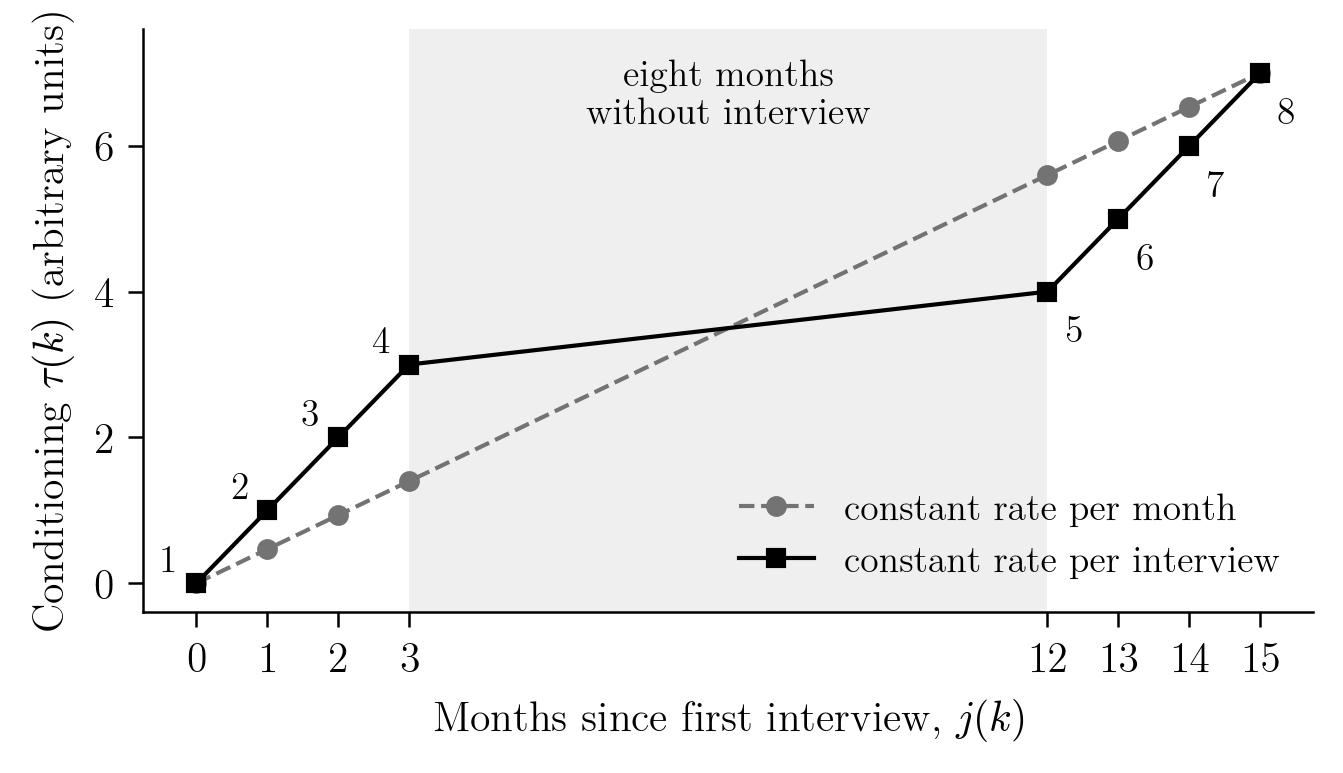}}

}

\caption{\label{fig-cps}The two constant-rate paths of
Corollary~\ref{cor-cps}(iii) under the CPS 4--8--4 pattern, plotted
against the months since first interview \(j(k)\); labels give the
interview number \(k\), and the shaded band spans the eight months
without an interview between the fourth and fifth interviews. A path
that accumulates at a constant rate per month (dashed) is linear in
\(j(k)\), so it lies in the direction \(m\,j(\cdot)\) that no common
interview schedule identifies (Theorem~\ref{thm-dose}(a)) and cannot be
told apart from no conditioning. A path that accumulates at a constant
rate \(c\) per interview (solid) rises by the same amount at every
interview, including the one after the rest; its \(\mathcal{D} = -8c\)
is nonzero, so under (P\('\)), with the cohort effects estimated from
the joint cohort-by-period array, it can be. Both paths are scaled to
the same value at the eighth interview; units are arbitrary.}

\end{figure}%

Three comments. First, everything downstream is unchanged:
Theorem~\ref{thm-absorb} is stated for an arbitrary element of
\(\ker X_{\mathcal{S}}\), so two-way fixed effects annihilate the
unidentified component of a dose design exactly as they do that of a
tenure design, and the corrected estimator of
Theorem~\ref{thm-correction} is the same regression with \emph{dose}
indicators in place of tenure indicators. What changes is which
component is unidentified---\(m\,j(k)\) rather than \(m(k-1)\)---and
therefore which functionals the correction can use. Second, the
mechanism is the mirror image of the one in the Antecedents paragraph
above: unequal spacing between \emph{entry cohorts} creates periodic
directions that no schedule can see, while unequal spacing between
\emph{a unit's own interviews} can remove from the kernel a
constant-rate path that every equally spaced schedule must live
with---it does so whenever \(d = 1\), and for \(d > 1\) exactly when
\(k - 1\) is not of the form \(m\,j(k) + \rho(j(k))\), which an
unequally spaced pattern can still satisfy (with \(J = \{0,1,3,4\}\) and
\(d = 3\), \(k - 1 = \tfrac{2}{3} j(k) + \rho(j(k))\) with
\(\rho \equiv \tfrac13\) on the class of \(1\)). Identification in this
model is entirely a question about which functions of elapsed calendar
time can be written as a function of the period plus a function of the
cohort. Third, (P\('\)) is a real condition and not a formality: with
the CPS pattern and between two and nine entry cohorts the sample cannot
bridge the nine-month gap between offsets \(3\) and \(12\), the fourth
and fifth interviews, and the identified set is two-dimensional rather
than one (with a single cohort nothing separates the doses and it is
seven-dimensional); and with gaps in the cohort set the identified set
can be larger still, as the counterexample above shows. Both are
verified by exact rank computation in the replication package, which
also reproduces every row of the table below.

\needspace{0.65\textheight}
\begin{longtable}[]{@{}
  >{\raggedright\arraybackslash}p{(\linewidth - 6\tabcolsep) * \real{0.3704}}
  >{\centering\arraybackslash}p{(\linewidth - 6\tabcolsep) * \real{0.0617}}
  >{\raggedright\arraybackslash}p{(\linewidth - 6\tabcolsep) * \real{0.3704}}
  >{\centering\arraybackslash}p{(\linewidth - 6\tabcolsep) * \real{0.1975}}@{}}
\caption{Examples of what a panel's schedule decides before any data are
collected; the complete criterion is Corollary~\ref{cor-curvature} for
uninterrupted designs and Theorem~\ref{thm-dose}(c) for interrupted ones
under (P\('\)), and the projected kernel of Theorem~\ref{thm-linear}(a)
in every case. The dimension is that of the identified set of the
conditioning path. The identified functionals form a space of dimension
(number of free conditioning coordinates) minus \(\dim \Theta_\tau\),
which for the JLPS design is \(18 - 4 = 14\), of which the eleven lag-4
second differences are a proper subset; the third column lists examples,
not the whole space. The last column asks whether a path accumulating at
a constant rate per interview can be distinguished from no conditioning
at all. Every row is verified by exact rank computation in the
replication package.}\label{tbl-designs}\tabularnewline
\toprule\noalign{}
\begin{minipage}[b]{\linewidth}\raggedright
Schedule: entry stride; interview pattern
\end{minipage} & \begin{minipage}[b]{\linewidth}\centering
\(\dim \Theta_\tau\)
\end{minipage} & \begin{minipage}[b]{\linewidth}\raggedright
Examples of identified functionals of \(\tau\)
\end{minipage} & \begin{minipage}[b]{\linewidth}\centering
Constant-rate dose path separable from none?
\end{minipage} \\
\midrule\noalign{}
\endfirsthead
\toprule\noalign{}
\begin{minipage}[b]{\linewidth}\raggedright
Schedule: entry stride; interview pattern
\end{minipage} & \begin{minipage}[b]{\linewidth}\centering
\(\dim \Theta_\tau\)
\end{minipage} & \begin{minipage}[b]{\linewidth}\raggedright
Examples of identified functionals of \(\tau\)
\end{minipage} & \begin{minipage}[b]{\linewidth}\centering
Constant-rate dose path separable from none?
\end{minipage} \\
\midrule\noalign{}
\endhead
\bottomrule\noalign{}
\endlastfoot
adjacent entry; uninterrupted; long follow-up & 1 & all second
differences \(\Delta^2\tau\) & no \\
entry every \(d\) waves; uninterrupted; long follow-up & \(d\) &
lag-\(d\) second differences \(\Delta_d^2\tau\) & no \\
JLPS: entries at waves 1, 5, 13; uninterrupted; \(T = 19\) & 4 &
\(\Delta_4^2\tau\), and 3 further contrasts & no \\
entries at 1, 4, 8; uninterrupted; \(T = 8\) (no follow-up) & 3 & only
within increment-graph components & no \\
single refreshment at wave \(e_2\); uninterrupted & \(e_2 - 1\) &
\(\Delta_{e_2-1}^2\tau\) & no \\
monthly entry, \(\ge 10\) cohorts; CPS 4--8--4 pattern & 1 & differences
of per-month rates (Corollary~\ref{cor-cps}) & \textbf{yes} \\
monthly entry, 2--9 cohorts; CPS 4--8--4 pattern & 2 & a strict subset
of the above & no \\
cohorts 1, 2, 4; offsets 0, 1, 3; \(T = 7\) & 2 & none: no nonzero
functional of the two free coordinates & no \\
adjacent entry; 2 on, 2 off, 2 on & 1 & differences of per-month rates &
\textbf{yes} \\
adjacent entry; every other period (1 on, 1 off) & 1 & second
differences in the dose & no \\
entry every 2 periods; CPS 4--8--4 pattern & 2 & the above, plus a
parity restriction & \textbf{yes} \\
\end{longtable}

Table 1 summarizes the design implications. The spacing of entry
cohorts, the timing of each unit's interviews, and the length of
follow-up jointly determine which functionals of conditioning a panel
can ever identify: the first sets the stride and with it the minimum
dimension of what cannot be known, the second decides whether a
constant-rate dose path is separable from the null, and the third
decides whether C\(_d\) holds, so that the dimension is exactly the
stride (Lemma~\ref{lem-trapezoid}(i)). None of these decisions is
usually taken on identification grounds; their identification
implications are cheap to compute before fielding from the scheduled
support. A scheduled interruption can change which restrictions are
testable, subject to the support conditions of (P\('\)): a rest period
that a rotation design has for cost reasons is what lets a constant
per-interview increment be told apart from a constant per-month
one---provided the cohort effects are estimated jointly rather than
averaged away.

\section{Attrition and Refreshment Samples: What
Changes}\label{sec-equivalence}

Theorem~\ref{thm-linear} assumes M3, no attrition. Relaxing it changes
the problem in a way that this paper does not solve but must locate.
With nonignorable attrition, the stayers' cell means are contaminated by
selection on the latent outcome, and a refreshment sample---a fresh
cross-section drawn at a later wave---identifies the population marginal
that the retained sample no longer represents. The companion paper
(\citeproc{ref-okubo2026refresh}{Okubo 2026}, Theorem 1) characterizes
exactly what such a sample identifies about a conditioning map when
attrition is unrestricted: the retention-scaled distribution of the
stayers' implied latent outcomes must be dominated, set by set, by the
refreshment distribution, and every map satisfying this constraint is
rationalized by an explicitly constructed admissible selection process.
Three consequences of that characterization matter here. First, for
observables generated by outcome-independent selection, the identified
set for a location shift under \emph{unrestricted} candidate selection
mechanisms depends on the tails of the refreshment distribution, and
collapses to a point whenever the refreshment density has tails thinner
than exponential (\citeproc{ref-okubo2026refresh}{Okubo 2026}, Corollary
3); if independence is instead maintained for every candidate mechanism,
the shift is identified by distributional alignment without any tail
condition. Outcome independence of the data-generating selection alone
does not imply point identification in the unrestricted class: with a
Laplace refreshment density \(f(y) = e^{-|y|}/2\), a true shift of zero
and constant retention \(p = 1/2\), the identified set is the interval
\([-\log 2, \log 2]\) (\citeproc{ref-okubo2026refresh}{Okubo 2026},
Corollary 3(ii)), every shift \(\varepsilon\) in it being rationalized
by the outcome-dependent candidate mechanism
\(\pi_\varepsilon(y) = \tfrac12 f(y+\varepsilon)/f(y)\), which lies in
\([0,1]\) exactly when \(|\varepsilon| \le \log 2\) and preserves the
reported distribution. What rules out confounding between selection and
conditioning is therefore a tail condition on the refreshment
distribution or a restriction on the candidate class, not the
independence of the actual selection by itself. Second, under
outcome-dependent attrition, the identified \emph{set} for a location
shift is governed by the distributional domination condition over the
entire outcome support, including the tails; it is an interval when the
refreshment distribution has a unimodal, upper semicontinuous density
(\citeproc{ref-okubo2026refresh}{Okubo 2026}, Theorem 3), and it need
not be connected otherwise---with a refreshment distribution placing
mass \(.5\) on each of two points, retention \(.4\), and stayers all
reporting the lower point, the feasible latent locations are the two
points and nothing between them. Third, within a single cohort and with
unrestricted cross-wave selection, the panel's own longitudinal
observations do not narrow that set
(\citeproc{ref-okubo2026refresh}{Okubo 2026}, Theorem 2); separation is
restored only by additional information---negative-control items
measured once at entry, or restrictions linking selection across waves.
This paper does not use any Gaussian construction of selection
processes; the companion paper's completion argument is the treatment we
rely on.

\section{Recovery: Identifying Restrictions}\label{sec-recovery}

Each result below adds a restriction that removes part or all of
\(\mathcal{K}_\tau\), and each is stated together with its support
conditions and with what is, and is not, testable about it. Throughout
this section C\(_d\) is assumed, so that \(\mathcal{K}_\tau\) is the
affine-plus-periodic family of Theorem~\ref{thm-linear}(c); on a support
violating C\(_d\) the same restrictions are checked against the computed
\(\mathcal{K}_\tau\). Two of the three results need more than the tenure
projection. Proposition~\ref{prp-nc} and Proposition~\ref{prp-bounds}
combine the cell means with information about the cohort effects \(g\),
and information about \(g\) constrains \(\tau\) only through the
\emph{full} kernel: they therefore assume in addition

\begin{itemize}
\tightlist
\item
  \textbf{(CG) Connected incidence graph.} The cohort--period incidence
  graph of \(\mathcal{S}\) is connected---every staggered trapezoid
  satisfies it.
\end{itemize}

Under C\(_d\) and (CG), \(\mathcal{K} = \mathcal{K}_0\) exactly: every
kernel vector is of the form displayed in Theorem~\ref{thm-linear}(b),
with \(h_g(e) = m(e - e_0)\) (this is the end of the proof of
Theorem~\ref{thm-linear}(c), with the component constant \(\kappa\)
forced to zero everywhere). Without (CG) the conclusions of
Proposition~\ref{prp-nc} and Proposition~\ref{prp-bounds} can fail even
when C\(_d\) holds: on the six-cell support of
Section~\ref{sec-impossibility} (cohorts \(1, 4, 6\), each observed at
entry and once more) the kernel vector \(h_\tau(2) = a\),
\(h_\alpha(e+1) = -a\), \(h_g \equiv 0\) leaves \(g\) untouched, so
exact knowledge of \(g\)---or a drift bound of \(M = 0\)---leaves
\(\tau(2)\) free. (CG) is sufficient, not necessary, for either result.
The step in Proposition~\ref{prp-nc} from \(h_g = 0\) to \(m = 0\) needs
only that some component of the incidence graph contain two cohorts:
with cohorts \(1, 2, 5\) on the same two-cell schedule the graph has two
components, yet every kernel vector moves \(g\) and knowledge of \(g\)
pins \(\tau(2)\). The interval of Proposition~\ref{prp-bounds} is
guaranteed to be the identified set under (CG) and can fail to be
without it, but need not: on that same three-cohort support with every
cell mean zero and the zero representative, writing \(a = \tau(2)\) and
\(b = g(5)\), the first two cohorts force \(g(2) = a\) and the drift
constraints read \(|a| \le M\) and \(|(b-a)/3| \le M\), so every
\(a \in [-M, M]\) is feasible (take \(b = a\)) and the identified set
for \(\tau(2)\) is exactly the displayed interval \([-M, M]\). What (CG)
rules out is the mechanism of failure: without it a drift constraint
between cohorts in different components is absorbed by the free
component constant rather than imposed on \(m\), and the implied drift
between such cohorts can be moved by that constant without any
compensating movement of \(\hat\tau\), so the displayed interval depends
on the representative and can be strictly narrower than the identified
set, or empty---on the six-cell support it is \([-M, M]\) for the zero
representative while \(\tau(2)\) is unbounded for every \(M\), and on
the three-cohort support a representative with \(\hat g(5) = 2.7M\)
displays \([-M, 0.1M]\) while the identified set is still \([-M, M]\).
We state (CG) because it is the condition every staggered trapezoid
satisfies and the one under which both results hold as displayed, with
the displayed interval independent of the representative.
Proposition~\ref{prp-plateau} uses only \(\mathcal{K}_\tau\) and needs
no such condition. The replication package's
\texttt{check\_recovery\_support.py} verifies the six-cell example, the
three identified sets for \(\tau(2)\) just quoted (by exact
Fourier--Motzkin elimination over the kernel), and the identity
\(\dim\mathcal{K} = \dim\mathcal{K}_\tau + (c-1)\) of
Theorem~\ref{thm-linear}(d) by exact arithmetic on 9,183 supports.

\begin{proposition}[Plateau (saturation) as an identifying
restriction]\protect\hypertarget{prp-plateau}{}\label{prp-plateau}

Suppose \(\tau(s) = \bar{\tau}\) for all \(s \ge S^*\), and suppose the
observed tenure range contains at least \(d + 1\) consecutive tenures
\(s \ge S^*\). Then \(\Theta_\tau\) collapses to a point: the
restriction forces \(m = 0\) and \(\rho \equiv 0\) in
Theorem~\ref{thm-linear}(c), so \(\tau\) is point identified on the
observed tenure set. The restriction has testable \emph{observable
implications}: every identified functional \(\lambda'\tau\) of
Corollary~\ref{cor-curvature} that is supported on tenures \(s \ge S^*\)
\emph{and} has zero-sum weights, \(\sum_s \lambda_s = 0\), must
vanish---in particular \(\Delta_d^2\tau(s) = 0\) whenever \(s - d\),
\(s\), and \(s + d\) are observed and \(s - d \ge S^*\) (all three
points on the plateau; equivalently \(s \ge S^* + d\)) and, with three
or more cohorts, the constancy over time of the cohort gaps observed
within the plateau window. Identified functionals whose weights do not
sum to zero, such as \(2\tau(5) - \tau(9)\) on a stride-four support,
are not implied to vanish---on a plateau they equal
\(\bar\tau\sum_s\lambda_s\). The restriction itself is not testable: the
affine path \(\tau(s) = m(s-1)\), \(m \ne 0\), satisfies every
observable implication and does not saturate. A non-rejection therefore
supports the plateau only against nonlinear departures; the zero-slope
part of the restriction must be argued on substantive grounds or
supplied by an external anchor (Proposition~\ref{prp-nc}, or a
refreshment design of the kind treated in the companion paper). The
restriction is the measurement-domain form of the oldest identifying
device in the age--period--cohort literature: Mason et al.
(\citeproc{ref-masetal1973}{1973}) break the linear dependency by
assuming that two categories within one dimension have identical
effects---``under the assumption that any two ages, periods or cohorts
have identical effect parameters, differences of the form \ldots{} are
now estimable''---and a plateau is that assumption imposed on a block of
adjacent tenures rather than on a pair.

\end{proposition}

\begin{proof}
On the plateau range \(h_\tau(s) = m(s-1) + \rho(s)\) must be constant.
The range contains some \(s\) and \(s + d\), and \(\rho(s+d) = \rho(s)\)
gives \(md = 0\), so \(m = 0\); then \(\rho\) is constant on \(d + 1\)
consecutive tenures, hence on every residue class, hence everywhere, and
it equals \(\rho(1) = 0\). For the observable implications, an
identified functional supported on the plateau takes the value
\(\bar\tau\sum_s\lambda_s\) there, which is zero iff the weights sum to
zero. For untestability, \(\tau(s) = m(s-1)\) lies in \(\Theta_\tau\) of
the null path and has zero identified curvature everywhere. \(\square\)
\end{proof}

\begin{proposition}[Entry-wave negative
controls]\protect\hypertarget{prp-nc}{}\label{prp-nc}

Let \(N\) be a set of items measuring time-invariant facts, each
reported exactly once per respondent at entry (\(s=1\)). By construction
\(\tau_N \equiv 0\), so entry contrasts across cohorts identify
\(g_N(e) - g_N(e')\) plus sampling and selection differences. Suppose
that (i) recall of the fact is stable across the calendar times at which
different cohorts enter, (ii) the instrument and its administration are
the same for every cohort, (iii) the entry contrasts are computed on the
full entry cohorts, before any attrition (M3), or after a correction for
selection of the kind the companion paper supplies, and (iv) cohort
effects transport to the target item \(j\) with a known or bounded
loading, \(g_j(e) = \Gamma_j\, g_N(e)\). Then, under C\(_d\) and (CG),
the negative controls identify (or bound) the cohort effects \(g_j\),
and knowledge of \(g_j\) removes the affine direction of
\(\mathcal{K}_\tau\): the coefficient \(m\) in
Theorem~\ref{thm-linear}(c) is pinned, because under (CG) every kernel
vector has \(h_g(e) = m(e - e_0)\), which must vanish. It does
\emph{not} remove the periodic directions, which leave \(g\) unchanged.
Consequently: if \(d = 1\), \(\tau_j\) is point identified (or bounded
by the range of \(\Gamma_j\)); if \(d > 1\), the negative controls
identify \(\tau_j(s)\) at the tenures \(s \equiv 1 \pmod d\) and the
identified functionals of Corollary~\ref{cor-curvature} combined with
such levels, while \(\tau_j(s)\) at other tenures remains identified
only up to the \(d - 1\) periodic directions, whatever the loading.
Point identification of the full path with \(d > 1\) requires, in
addition, a restriction that removes the periodic directions---the
plateau of Proposition~\ref{prp-plateau} spanning \(d + 1\) consecutive
tenures, or an external anchor of the periodic component. Conditions
(i)--(iv) are assumptions about the items and the sample, not
consequences of measuring the same respondents: an immutable fact does
not imply a stable recall, and a common battery does not imply a common
loading. We therefore recommend reporting the sensitivity of the
recovered path to \(\Gamma_j\) over a substantive range and, where no
range can be defended, treating the battery as a diagnostic of cohort
comparability rather than as an identifying instrument. Absent any
restriction linking \(g_j\) across cohorts, Theorem~\ref{thm-linear}
applies item by item.

\end{proposition}

\begin{proof}
Under (i)--(iv), \(g_j\) is known (or bounded) on \(\mathcal{E}\). Two
triples generating the same cell means differ by \(h \in \mathcal{K}\),
and under C\(_d\) and (CG), \(\mathcal{K} = \mathcal{K}_0\), so
\(h_g(e) = m(e - e_0)\); if both triples have the known \(g_j\), then
\(m(e - e_0) = 0\) at some \(e \ne e_0\), so \(m = 0\), and the
remaining freedom is \(h_\tau = \rho\), \(d\)-periodic with
\(\rho(1) = 0\), which vanishes at \(s \equiv 1 \pmod d\) and is
identically zero iff \(d = 1\). Without (CG) the first step fails, as
the six-cell example above shows. On the stride-four support of
Section~\ref{sec-jlps}, for instance,
\(\rho(s) = a\,\mathbf{1}\{s \equiv 2 \pmod 4\}\) with
\(h_\alpha(t) = -a\,\mathbf{1}\{t \equiv 2 \pmod 4\}\) and
\(h_g \equiv 0\) preserves every cell mean, every entry-wave
negative-control contrast, and the exact value of \(g\), while moving
\(\tau(2)\) by an arbitrary \(a\). \(\square\)
\end{proof}

\begin{proposition}[Partial identification under bounded cohort
drift]\protect\hypertarget{prp-bounds}{}\label{prp-bounds}

Suppose cohort effects drift at a bounded per-period rate: for all
consecutive entry cohorts \(e_k < e_{k+1}\), \[
\left| \frac{g(e_{k+1}) - g(e_k)}{e_{k+1} - e_k} \right| \;\le\; M .
\] Assume C\(_d\) and (CG). The periodic directions of
Theorem~\ref{thm-linear} do not move \(g\), so this restriction
constrains only the affine coefficient \(m\). Fix any representative
\((\hat\alpha, \hat g, \hat\tau) \in \Theta\) and let
\(\hat d_k = \{\hat g(e_{k+1}) - \hat g(e_k)\}/(e_{k+1}-e_k)\) be its
implied drifts. Under (CG) the true triple has
\(g = \hat g + m(e - e_0)\) for some \(m\) (every kernel vector is in
\(\mathcal{K}_0\)), so \(|\hat d_k + m| \le M\) for every \(k\), giving
\(m \in [\,-M - \min_k \hat d_k,\; M - \max_k \hat d_k\,]\).
Consequently, for every tenure \(s \equiv 1 \pmod d\)---the tenures at
which the periodic component vanishes---the bounds \[
\tau(s) - \tau(1) \;\in\; \left[\, \hat\tau(s) + m_{lo}(s-1),\; \hat\tau(s) + m_{hi}(s-1) \,\right]
\] are sharp, with endpoints linear in \(M\); the same holds for any
identified functional of Corollary~\ref{cor-curvature} combined with
such a level. At tenures \(s \not\equiv 1 \pmod d\) the periodic
component remains free and the identified set for \(\tau(s) - \tau(1)\)
is unbounded for every \(M\). The interval is nonempty iff
\(\max_k \hat d_k - \min_k \hat d_k \le 2M\), so with three or more
cohorts the bound is refutable; \(M = 0\) delivers point identification
at the covered tenures together with an overidentification test
(\(\hat d_1 = \hat d_2\)); \(M \to \infty\) recovers \(\Theta_\tau\).
Reporting the breakdown value \(M^*\) at which conclusions reverse is
the recommended default output. In a panel refreshed at stride \(d\) the
covered tenures are \(1, 1 + d, 1 + 2d, \dots\)---in the Japanese panel
of Section~\ref{sec-jlps}, tenures 5, 9, 13, and 17. These are
tenures---arguments of \(\tau\), common to every cohort, and equal to
the wave number only for the original cohort---not calendar waves: that
panel's refreshment samples entered at waves 5 and 13, and the bound
applies at the listed tenures because the periodic component vanishes
there.

\emph{Why first differences and not second.} The affine direction
\(g(e) + m(e - e_0)\) is linear in \(e\), so any restriction on second
differences of \(g\)---smoothness in the Rambachan--Roth sense---is
invariant along \(\Theta_\tau\) and yields no information about \(m\):
under \(|\Delta^2 g| \le M\) the identified set for \(\tau(s)-\tau(1)\)
remains the whole real line for every \(M\). This is the
measurement-domain image of the APC fact that it is only the
\emph{slopes} of an APC model that are unidentified, so that constraints
invariant to the linear component cannot recover it; the bounds must
instead restrict the linear component, whether by a direct bound on a
slope or, as Fosse and Winship (\citeproc{ref-fosse2019bounds}{2019b})
also show, indirectly through restrictions on the effect paths that
imply one.

\end{proposition}

\section{Where the Unidentified Component Goes: Downstream
Propagation}\label{sec-downstream}

Theorem~\ref{thm-linear} says what a panel cannot reveal about
conditioning. This section asks how nonidentification of the
conditioning path affects fixed-effects, difference-in-differences, and
event-study estimators computed on the same panel, and finds that the
answer is exactly complementary: the component of \(\tau\) that cannot
be identified is the component that the workhorse panel estimators
discard, and the component they do not discard is identified.
Identification failure upstream and robustness downstream are the same
piece of linear algebra.

\subsection{Setting}\label{sec-downstream-setting}

Keep the tenure structure of Section~\ref{sec-setup}: individual \(i\)
enters at \(e_i\) and is observed at calendar times \(t\) with tenure
\(s_{it} = t - e_i + 1\). An analyst is interested in an outcome
\(Y^{\circ}_{it}\)---the response the individual would give at tenure
one, \(Y^{\circ}_{it} := Y_{it}(1)\), the unconditioned response---and
possibly in the effect of a regressor \(D_{it}\) on it. The analyst
observes the report \begin{equation}\phantomsection\label{eq-report}{
Y_{it} \;=\; Y^{\circ}_{it} \;+\; \tau(s_{it}) \;+\; \nu_{it},
}\end{equation} where \(\tau\) is the conditioning path of
Equation~\ref{eq-model} and \(\nu_{it}\) collects individual departures
from the mean path. We maintain:

\begin{itemize}
\tightlist
\item
  \textbf{D1 (mean-additive drift).}
  \(\mathbb{E}[\nu_{it} \mid \{D_{i't'}, s_{i't'}\}_{i',t'}] = 0\): the
  conditioning effect may be heterogeneous, but its mean given the
  design is the common path \(\tau(s)\). Heterogeneity that is itself
  correlated with \(D\) is excluded by D1; it is a separate, standard
  endogeneity problem and is not the subject here. Every statement below
  is a statement about a common additive tenure effect satisfying D1.
\item
  \textbf{D2 (uninterrupted participation).} Units are interviewed at
  every wave from entry until they leave, so that calendar tenure
  \(t - e_i + 1\) equals the number of interviews received. The baseline
  presentation assumes uninterrupted participation; the absorption
  argument of Theorem~\ref{thm-absorb} extends to the dose-indexed
  specification of Section~\ref{sec-interrupted}, with dose in place of
  tenure, as noted after Corollary~\ref{cor-cps}, while the results that
  use the tenure clock explicitly (Theorem~\ref{thm-pretrend},
  Proposition~\ref{prp-sensitivity}) are stated for the uninterrupted
  case only.
\end{itemize}

Let \(W\) be the \(n \times (N + T)\) matrix of unit and calendar-time
indicators for the observed rows (any unbalanced support),
\(M = I - W(W'W)^{-}W'\) the residual-maker that projects off unit and
time effects, and for any \(n\)-vector \(v\) write \(v(s)\) for the
vector with entries \(v(s_{it})\). Throughout, ``TWFE'' means the OLS
coefficient on \(D\) in a regression with unit and time effects,
\(\hat\beta = (D'MD)^{-1} D'MY\), on whatever support is observed. The
two-way \emph{within} transformation on a balanced panel is the special
case \(M D = D_{it} - \bar D_{i\cdot} - \bar D_{\cdot t} + \bar D\);
nothing below requires balance.

\subsection{The absorption theorem}\label{sec-absorption}

\begin{theorem}[Absorption of the unidentified component by two-way
fixed effects]\protect\hypertarget{thm-absorb}{}\label{thm-absorb}

Under D1, let \(\hat\beta\) be the TWFE coefficient computed on the
reports \(Y\) and \(\hat\beta^{\circ}\) the (infeasible) coefficient
computed on the unconditioned outcomes \(Y^{\circ}\), both on the same
support with the same regressor. Then, conditional on the design
\((D, s)\), \begin{equation}\phantomsection\label{eq-absorb}{
\mathbb{E}\big[\hat\beta - \hat\beta^{\circ}\big] \;=\; (D'MD)^{-1} D'M\,\tau(s) \;=\; \frac{\sum_{it} \ddot D_{it}\, \tau(s_{it})}{\sum_{it} \ddot D_{it}^2},
\qquad \ddot D := MD,
}\end{equation} and the following hold.

\begin{enumerate}
\def\labelenumi{(\roman{enumi})}
\item
  \emph{Absorption.} Let \(\mathcal{S}\) be the cohort--period support
  of the estimation sample and \(Z_\tau\) the matrix of tenure
  indicators \(\mathbf{1}\{s_{it} = s\}\). For every
  \(h = (h_\alpha, h_g, h_\tau) \in \ker X_{\mathcal{S}}\),
  \(M Z_\tau h_\tau = 0\). In particular \(M\{a + b\,s + \rho(s)\} = 0\)
  for every \(a, b \in \mathbb{R}\) and every \(d\)-periodic \(\rho\),
  \(d\) the stride of the entry cohorts in the sample. Hence the bias in
  Equation~\ref{eq-absorb} is unchanged when \(\tau\) is replaced by any
  other member of the identified set \(\Theta_\tau\) of
  Theorem~\ref{thm-linear}(a) computed on the same support: it is
  invariant along the \emph{entire} identified set, on any support,
  whether or not C\(_d\) holds.
\item
  \emph{Exactness for unidentified drift.} If
  \(\tau \in \mathcal{K}_\tau\)---the identified set of the null path on
  the sample's support; under C\(_d\), an affine function of \(s\) plus
  a \(d\)-periodic function---then
  \(\mathbb{E}[\hat\beta - \hat\beta^{\circ}] = 0\).
\item
  \emph{Only the identified part matters.} Writing
  \(\tau = \tau_{\parallel} + \tau_{\perp}\) for any decomposition into
  a member of \(\mathcal{K}_\tau\) and a remainder, the bias equals
  \((D'MD)^{-1}D'M\,\tau_{\perp}(s)\); in particular it is zero if and
  only if \(\ddot D \perp \tau_{\perp}(s)\).
\end{enumerate}

\end{theorem}

\begin{proof}
By the Frisch--Waugh--Lovell theorem, \(\hat\beta = (D'MD)^{-1}D'MY\)
and \(\hat\beta^{\circ} = (D'MD)^{-1}D'MY^{\circ}\) on the same support,
so by Equation~\ref{eq-report},
\(\hat\beta - \hat\beta^{\circ} = (D'MD)^{-1}D'M\{\tau(s) + \nu\}\), and
D1 removes the second term in expectation. Since \(M\) is symmetric and
idempotent,
\(D'M\tau(s) = (MD)'\tau(s) = \sum_{it}\ddot D_{it}\tau(s_{it})\). For
(i), let \(h \in \ker X_{\mathcal{S}}\). The cell equation
\(h_\alpha(t) + h_g(e) + h_\tau(t - e + 1) = 0\) holds for every fielded
cell, so for every observed row,
\((Z_\tau h_\tau)_{it} = h_\tau(s_{it}) = -h_\alpha(t) - h_g(e_i)\): a
function of calendar time plus a function of the unit (through its entry
cohort). Both lie in \(\mathrm{col}(W)\), so
\(Z_\tau h_\tau \in \mathrm{col}(W)\) and \(M Z_\tau h_\tau = 0\). The
affine and periodic directions are the elements of
\(\mathcal{K}_0 \subseteq \ker X_{\mathcal{S}}\) of
Theorem~\ref{thm-linear}(b), which gives the special case; directly,
\(s_{it} = t + (1 - e_i)\) and \(\rho(s_{it}) = \rho(t - e_0 + 1)\) when
all cohorts are congruent modulo \(d\). Replacing \(\tau\) by another
member of \(\Theta_\tau\) adds \(D'M Z_\tau h_\tau = 0\) to the
numerator. (ii) is the case \(\tau \in \mathcal{K}_\tau\). (iii) follows
from (i) by linearity:
\(D'M\tau(s) = D'M\tau_{\parallel}(s) + D'M\tau_{\perp}(s) = D'M\tau_{\perp}(s)\),
and the bias is zero iff \(\ddot D'\tau_{\perp}(s) = 0\). \(\square\)
\end{proof}

Three remarks. First, Theorem~\ref{thm-absorb} is an exact finite-sample
statement conditional on the design, not a large-\(N\) approximation,
and it holds on any unbalanced support: attrition, rotation, and
refreshment do not disturb the argument because they change the rows of
\(W\), not the fact that \(s_{it}\) is a sum of a unit term and a time
term. Second, the theorem does not say that TWFE is immune to
conditioning; it says that TWFE is immune to \emph{every unidentified
component} of the path, and that an identified component affects a
particular coefficient only when it projects onto the residualized
regressor \(\ddot D\)---so immunity to the identified part is
coefficient-specific, not general. The periodic directions, which are
not affine, do not bias fixed-effects estimates either: on any sample
whose cohorts share a stride they are functions of calendar time, and on
any sample at all every kernel direction's tenure component is a sum of
a time function and a unit function, which is the one-line content of
part (i). Third, the same argument covers any estimator that partials
out unit and time effects before using \(D\)---the two-way Mundlak, the
interactive-fixed-effects estimators with unit and time effects included
additively, and first differences with time effects, since
\(\Delta\tau(s_{it})\) is constant for affine \(\tau\) and absorbed by
the intercept.

\begin{corollary}[Protected and exposed
designs]\protect\hypertarget{cor-protect}{}\label{cor-protect}

~

\begin{enumerate}
\def\labelenumi{(\alph{enumi})}
\item
  \emph{Single-cohort panels.} If all units in the estimation sample
  share one entry cohort, then \(s_{it}\) is a function of \(t\) alone,
  \(M\tau(s) = 0\) for \textbf{every} path \(\tau\), and every TWFE,
  difference-in-differences, and first-difference coefficient computed
  within the cohort has zero measurement-induced bias relative to its
  unconditioned counterpart, under tenure drift of arbitrary shape. This
  is a statement about the incremental bias due to conditioning, not
  about the causal validity of the unconditioned estimand.
\item
  \emph{Tenure-orthogonal treatment.} More generally, the bias vanishes
  for every \(\tau\) if and only if the double-demeaned regressor sums
  to zero within every tenure level,
  \(\sum_{it:\,s_{it}=s}\ddot D_{it} = 0\) for all \(s\) in the support.
  When treatment timing is assigned independently of entry cohort given
  calendar time, \(\mathbb{E}[D_{it}]\) is a function of \(t\) alone, so
  \(\mathbb{E}[\ddot D] = M\,\mathbb{E}[D] = 0\) and the numerator of
  Equation~\ref{eq-absorb} has mean zero over the assignment
  distribution.
\item
  \emph{Level contrasts.} Estimands that compare levels across units of
  different tenure at a common period---pooled cross-section
  comparisons, cohort contrasts, panel-versus-fresh-sample benchmarks,
  and the ``experienced versus new respondent'' contrasts of the
  conditioning literature itself---identify the \emph{sum}
  \(g(e) - g(e') + \tau(s) - \tau(s')\), which is invariant along
  \(\Theta_\tau\); but its decomposition into a cohort difference and a
  conditioning difference is not identified, and the conditioning
  component \(\tau(s) - \tau(s')\) that such contrasts are meant to
  measure inherits the non-identification of Theorem~\ref{thm-linear}
  unless a normalization or recovery condition of
  Section~\ref{sec-recovery} is invoked.
\end{enumerate}

\end{corollary}

\begin{proof}
\leavevmode

\begin{enumerate}
\def\labelenumi{(\alph{enumi})}
\tightlist
\item
  With a common \(e\), \(\tau(s_{it}) = \tau(t - e + 1)\) depends on
  \((i,t)\) only through \(t\), so \(\tau(s) \in \mathrm{col}(W)\) and
  \(M\tau(s) = 0\). (b) The numerator of Equation~\ref{eq-absorb} is
  \(\sum_{it}\ddot D_{it}\tau(s_{it}) = \sum_s \tau(s)\sum_{it:\,s_{it}=s}\ddot D_{it}\),
  which is zero for every \(\tau\) if the stated condition holds, and
  which is nonzero for \(\tau = \mathbf{1}\{s = s_0\}\) whenever the
  condition fails at \(s_0\). For the second claim, \(\ddot D = MD\) is
  linear in \(D\), so \(\mathbb{E}[\ddot D] = M\,\mathbb{E}[D]\), and a
  function of \(t\) alone lies in \(\mathrm{col}(W)\). (c) is
  Equation~\ref{eq-model} directly: a level contrast at a common \(t\)
  between cohorts \(e, e'\) with tenures \(s = t - e + 1\),
  \(s' = t - e' + 1\) equals \(g(e) - g(e') + \tau(s) - \tau(s')\).
  Along \(\mathcal{K}_0\) the conditioning component moves by
  \(m(s - s') + \rho(s) - \rho(s')\) and the cohort component by
  \(m(e - e')\); since \(s - s' = e' - e\) and
  \(\rho(s) = \rho(s') = \rho(t - e_0 + 1)\), the two movements cancel
  and the sum is unchanged, while each component moves. \(\square\)
\end{enumerate}

\end{proof}

Part (a) is the reassurance the applied literature has been implicitly
relying on: a panel that recruited everyone at once cannot have its
within-unit estimates contaminated by a common additive tenure effect
satisfying D1, whatever its shape. Part (c) is the warning: the moment a
design mixes tenures within a period---every rotating panel, every panel
with refreshment samples, every pooled analysis of an old and a new
cohort---the observed level contrast is identified but its conditioning
component is not, so its reading as a conditioning effect is
normalization-dependent; the TWFE protection of
Theorem~\ref{thm-absorb}, by contrast, extends to the whole unidentified
component, and what it does not cover is the identified curvature.

\subsection{The bias is identified even though the level is
not}\label{sec-correction}

Theorem~\ref{thm-linear} and Theorem~\ref{thm-absorb} combine into a
constructive statement that we regard as the practical center of the
paper.

\begin{theorem}[Identified correction of two-way fixed-effects
estimates]\protect\hypertarget{thm-correction}{}\label{thm-correction}

Let \(\lambda := MD\,(D'MD)^{-1}\) be the implicit weights of the TWFE
estimator, so that the bias in Equation~\ref{eq-absorb} is the linear
functional
\(b(\tau) = \sum_{it}\lambda_{it}\,\tau(s_{it}) = \sum_s \Lambda_s\,\tau(s)\)
with \(\Lambda_s = \sum_{it:\, s_{it}=s}\lambda_{it}\). Then
\(\Lambda \perp \mathcal{K}_\tau\) for the kernel of the estimation
sample's own cell design---under C\(_d\), the pair of conditions of
Corollary~\ref{cor-curvature}, \(\sum_s \Lambda_s (s-1) = 0\) and
\(\sum_{s \equiv r}\Lambda_s = 0\) for every residue class \(r\) modulo
the stride of the sample. Consequently \(b(\tau)\) is point identified
whenever \(\tau\) is identified up to \(\Theta_\tau\) on the estimation
sample (M1--M4 on that sample, or an external estimate of the shape of
\(\tau\) on a design whose kernel is no larger). A feasible corrected
estimator is the coefficient on \(D\) in the joint regression
\begin{equation}\phantomsection\label{eq-joint}{
Y \;=\; W\gamma \;+\; \beta D \;+\; Z_\tau\theta \;+\; \text{error},
}\end{equation} where \(Z_\tau\) is the matrix of tenure indicators
\(\mathbf{1}\{s_{it} = s\}\) and the columns of \(Z_\tau\) that are
collinear with \(W\)---exactly the directions of \(\Theta_\tau\)---are
dropped. Suppose that, on the true outcome,
\(\mathbb{E}[Y^{\circ} \mid D, W, s] = W\gamma + \beta D\), and that

\begin{itemize}
\tightlist
\item
  \textbf{(R) Residual rank.} After partialling out \(W\), the regressor
  \(D\) is not collinear with the retained tenure indicators:
  \([\,MD,\; MZ_\tau\,]\) has full column rank.
\end{itemize}

Then, under D1--D2, the joint OLS coefficient \(\hat\beta^{c}\)
satisfies \(\mathbb{E}[\hat\beta^{c}] = \beta\) exactly, conditional on
the design. Equivalently, writing \(\hat\theta^{c}\) for the tenure
coefficients fitted \emph{jointly} with \(\hat\beta^{c}\) in
Equation~\ref{eq-joint}, \[
\hat\beta^{c} \;=\; \hat\beta \;-\; (D'MD)^{-1}D'M Z_\tau\,\hat\theta^{c},
\] and the right-hand side is the same for every representative
\(\hat\theta^{c} + h_\tau\), \(h \in \ker X_{\mathcal{S}}\), of the
joint fit, provided the retained tenure columns span
\(\mathrm{col}(Z_\tau)\) modulo \(\mathrm{col}(W)\) (the display alone
does not check that they do); the correction subtracts the joint-OLS
estimate of \(b(\tau)\), which does not depend on which representative
the software reports. It is an identity for the joint fit, not for an
arbitrary preliminary estimate of the path substituted into the display.
No normalization of the conditioning level is required. Condition (R)
fails, and \(\beta\) is not identified, when treatment status is a
deterministic function of tenure on the sample (for instance
\(D_{it} = \mathbf{1}\{s_{it} \ge s^*\}\)), or more generally when
\(MD\) lies in the span of the retained tenure indicators. In that case
the coefficient on \(D\) is not determined by the data: fits with any
value of \(\beta\) reproduce the same fitted values. Which column a
regression routine then reports as aliased depends on the order of the
columns, and a routine that keeps \(D\) and drops a tenure indicator
reports a coefficient for \(\beta\) that depends on an arbitrary
normalization rather than on an identified feature of the data; (R) must
therefore be checked directly, by computing the rank of
\([MD, MZ_\tau]\) (or of \([W_0, D, Z_\tau]\) against
\([W_0, Z_\tau]\)), before the coefficient is interpreted.

\end{theorem}

\begin{proof}
For any \(h \in \ker X_{\mathcal{S}}\),
\(\sum_s \Lambda_s h_\tau(s) = \lambda' Z_\tau h_\tau = (D'MD)^{-1} D'M Z_\tau h_\tau = 0\)
by Theorem~\ref{thm-absorb}(i), so \(\Lambda \perp \mathcal{K}_\tau\)
and, by Theorem~\ref{thm-linear}(a), \(b\) is constant on
\(\Theta_\tau\); under C\(_d\) the same follows from \(\Lambda'1 = 0\),
\(\Lambda'(s-1) = 0\) (because \(1, s \in \mathrm{col}(W)\)) and
\(\sum_{s \equiv r}\Lambda_s = \lambda'\mathbf{1}\{s \equiv r\} = 0\)
(the residue-class indicator is a function of \(t\) on the sample). For
the feasible estimator, the reported outcome satisfies
\(Y = Y^{\circ} + Z_\tau\tau + \nu\) exactly, since
\(\tau(s_{it}) = (Z_\tau\tau)_{it}\); hence
\(Y = W\gamma + \beta D + Z_\tau\tau + \{Y^{\circ} - W\gamma - \beta D\} + \nu\),
and the bracketed term and \(\nu\) have conditional mean zero by the
outcome model and D1. Dropping the columns of \(Z_\tau\) collinear with
\(W\) changes \(\theta\) to another representative of \(\Theta_\tau\)
and leaves the fitted values, hence the coefficient on \(D\), unchanged.
Under (R) the remaining design \([W, D, Z_\tau]\) has full column rank
up to \(\mathrm{col}(W)\), so OLS is unbiased for \(\beta\). For the
equivalent form, the normal equation for \(D\) in the joint fit, after
partialling out \(W\), is
\(D'M(Y - D\hat\beta^{c} - Z_\tau\hat\theta^{c}) = 0\), which rearranges
to the display; adding \(h_\tau\) to \(\hat\theta^{c}\) changes the
subtracted term by \((D'MD)^{-1}D'MZ_\tau h_\tau = 0\) by
Theorem~\ref{thm-absorb}(i). If \(MD \in \mathrm{col}(MZ_\tau)\) the
coefficient on \(D\) is not determined;
\texttt{check\_recovery\_support.py} exhibits, for
\(D = \mathbf{1}\{s \ge 4\}\) on the staggered support of
Section~\ref{sec-jlps}, two exact fits of the same outcome with
\(\beta = 0\) and \(\beta = 1\). \(\square\)
\end{proof}

The result inverts the usual reading of an impossibility theorem. The
level of conditioning cannot be known; but the quantity the applied
analyst needs---the conditioning bias in a fixed-effects
coefficient---does not depend on the level, and is recoverable from the
panel's own staggered structure or from an external estimate of the
conditioning \emph{shape} (a refreshment design, or the shape library
assembled in companion work). Three cautions accompany it. The
correction removes the implied measurement component when the outcome
model and D1 hold; its uncertainty includes the estimation of that
component, and the joint regression's variance estimator must respect
the dependence in the data---we use standard errors clustered by
unit---whereas a two-step procedure that first calibrates \(\tau\) from
the same outcome's cell means and then subtracts it neither propagates
that uncertainty nor, if the cell-mean model omits \(D\), keeps the
genuine effect out of the estimated tenure component. Condition (R) is a
real restriction: designs in which exposure is scheduled by tenure
cannot separate the effect of exposure from conditioning at all. And the
correction is a \emph{within-sample} construction: it uses the
estimation sample's own staggered structure to estimate the identified
part of \(\tau\), and it is not available when that structure is absent,
as in the exactly aligned event studies of Section~\ref{sec-pretrend},
where an \emph{external} estimate of \(\tau\) must be supplied instead.
Simulation 4 (Appendix C) verifies the identity, reports the
finite-sample behavior of Equation~\ref{eq-joint} under five designs
including one that violates (R)---detected there by an explicit rank
computation, not by which coefficient the software declines to
report---and shows the practical caveat that with three entry cohorts
and stride four the identified curvature rests on few degrees of
freedom, so the correction is exact in expectation and imprecise in
practice; the plateau restriction of Proposition~\ref{prp-plateau} or an
external shape estimate are the routes to precision.

\subsection{Spurious pre-trends in tenure-aligned event
studies}\label{sec-pretrend}

The protection of Theorem~\ref{thm-absorb} has a precise boundary, and
event studies sit on it. Let \(E_i\) be the calendar time of an event
for unit \(i\) and \(k_{it} = t - E_i\) event time. The fully dynamic
event-study specification regresses the outcome on unit effects, time
effects, and indicators \(X_k = \mathbf{1}\{k_{it} = k\}\) for event
times \(k\) in a window \(\mathcal{K}\), omitting a set of reference
categories. It is well known (\citeproc{ref-borusyak2024}{Borusyak,
Jaravel, and Spiess 2024}; \citeproc{ref-sunabraham2021}{Sun and Abraham
2021}) that when every unit is eventually treated a complete set of
event-time indicators is collinear with unit and time effects in two
directions---a constant, and a linear trend in \(k\), since
\(k_{it} = t - E_i\) is the sum of a time term and a unit term---so that
two reference categories must be omitted (or an untreated group
included) for the coefficients to be defined. The second of those
directions is the same one as in Theorem~\ref{thm-linear}, and it is
exactly what makes tenure drift masquerade as dynamics.

\begin{theorem}[Mechanical pre-trends under exact tenure--event
alignment]\protect\hypertarget{thm-pretrend}{}\label{thm-pretrend}

Suppose the design aligns tenure with event time:
\(s_{it} = k_{it} + \delta\) for a constant \(\delta\) on the estimation
sample (equivalently \(E_i = e_i + \delta - 1\): every unit's event
falls a fixed number of waves after its entry, as in rotation panels,
event-recruited samples, and surveys fielded on entry schedules). Let
the analyst estimate the fully dynamic specification with unit and time
effects and indicators \(X\) for
\(k \in \mathcal{K}\setminus\{k_0, k_1\}\), \(k_0 \neq k_1\) two
reference categories, on an estimation sample whose rows all have
\(k_{it} \in \mathcal{K}\), and assume that the event-time indicators
are not collinear with the fixed effects: \(MX\) has full column rank,
where \(M\) projects off unit and time effects (equivalently,
\([W_0, X]\) has full column rank for any basis \(W_0\) of
\(\mathrm{col}(W)\); the full set of unit and time indicators is itself
rank-deficient by one and is not the object to check). This rank
condition should be checked directly; overlap of event times and
calendar periods alone does not guarantee it. For example, with event
dates \(E \in \{0,2,4\}\), periods \(t \in \{4,\dots,7\}\), and window
\(k \in \{2,\dots,5\}\), every period contains several event times and
every event time is observed in several periods, yet the eight columns
of \([W_0, X]\) have rank seven, because all event dates are even and
\(k \bmod 2\) is then a function of \(t\)---the stride phenomenon of
Theorem~\ref{thm-linear} again. Let \(\ell\) be the affine function of
\(k\) with \(\ell(k_0) = \tau(k_0 + \delta)\) and
\(\ell(k_1) = \tau(k_1 + \delta)\). Then, conditional on the design, the
coefficients on the reports and on the unconditioned outcomes satisfy,
for every \(k \in \mathcal{K}\setminus\{k_0,k_1\}\),
\begin{equation}\phantomsection\label{eq-pretrend}{
\mathbb{E}\big[\hat\beta_k - \hat\beta^{\circ}_k\big] \;=\; \tau(k + \delta) \;-\; \ell(k),
}\end{equation} exactly. In particular, if the unconditioned outcome is
such that \(\mathbb{E}[\hat\beta^{\circ}_k] = 0\) for \(k < 0\)---an
additional premise, satisfied for instance under no anticipation,
parallel trends, dynamic effects that are homogeneous across event
cohorts, \emph{and} reference categories at which the effect is zero
(both \(k_0\) and \(k_1\) before the event, say); it is not implied by
no anticipation and parallel trends alone when effects are heterogeneous
across cohorts (\citeproc{ref-sunabraham2021}{Sun and Abraham 2021,
secs. 3.4--3.5}), and not by homogeneity either when an omitted category
carries a nonzero effect, since the affine interpolant through a nonzero
post-event reference is in general nonzero at the pre-periods other than
the omitted one and so induces nonzero reported pre-event coefficients
(\citeproc{ref-sunabraham2021}{Sun and Abraham 2021, sec. 3.7}, eqs.
(24)--(25), p.~184); in the example in the proof below it is
\(-\tfrac12\) at \(k = -2\) and zero at the omitted \(k = -1\)---then
the \emph{expected} reported pre-period coefficients equal
\(\tau(k+\delta) - \ell(k)\): they are nonzero in expectation whenever
\(\tau\) is not affine on the pre-period tenure range, with no
anticipation, no selection, and no violation of parallel trends in the
outcome.

\end{theorem}

\begin{proof}
By Equation~\ref{eq-report} the reported outcome is
\(Y = Y^{\circ} + v + \nu\) with
\(v_{it} = \tau(s_{it}) = \tau(k_{it} + \delta)\). Decompose
\(\tau(k + \delta) = \ell(k) + r(k)\) with \(r(k_0) = r(k_1) = 0\). The
affine part satisfies \(\ell(k_{it}) = \ell_0 + \ell_1(t - E_i)\), the
sum of a function of \(t\) and a function of \(i\), hence
\(\ell(k) \in \mathrm{col}(W)\) and \(M\ell(k) = 0\). The remainder
satisfies
\(r(k_{it}) = \sum_{k \in \mathcal{K}\setminus\{k_0,k_1\}} r(k)\,X_{k,it}\)
because \(r\) vanishes at the omitted categories, so \(Mv = MX\gamma\)
with \(\gamma = (r(k))_k\). By Frisch--Waugh--Lovell the event-time
coefficients are \((X'MX)^{-1}X'MY\), and since \(MX\) has full column
rank, \((X'MX)^{-1}X'Mv = \gamma\) exactly; D1 removes \(\nu\) in
expectation. For the parenthetical premise, the same argument applied to
a homogeneous effect path \(\beta^{\circ}(k)\) with
\(\mathbb{E}[Y^{\circ}_{it}] = a_i + b_t + \beta^{\circ}(k_{it})\) gives
\(\mathbb{E}[\hat\beta^{\circ}_k] = \beta^{\circ}(k) - \ell^{\circ}(k)\)
with \(\ell^{\circ}\) the affine interpolant of \(\beta^{\circ}\) at the
references; this vanishes at every pre-period \(k\) when
\(\beta^{\circ}\) is zero before the event and at both references, and
it can fail when a reference carries a nonzero effect. With event dates
\(3, \dots, 7\), window \(k \in [-2, 2]\),
\(\beta^{\circ}(k) = \mathbf{1}\{k \ge 0\}\) and references
\(k_0 = -1\), \(k_1 = +1\), the design has full rank \(16\) and
\(\mathbb{E}[\hat\beta^{\circ}_{-2}] = 0 - (-2+1)/2 = +\tfrac12\)
(\texttt{check\_recovery\_support.py}). \(\square\)
\end{proof}

The theorem has two consequences that distinguish this artifact from the
failures the pre-trends literature is built to catch. Outcome-side
diagnostics may detect the resulting discrepancy without identifying its
measurement origin: the unconditioned outcome satisfies parallel trends
by construction, so a pretest (\citeproc{ref-roth2022}{Roth 2022}) may
detect the discrepancy without identifying whether it arises from
measurement drift, and it need not detect it at all, since power depends
on the size of \(\tau(k+\delta) - \ell(k)\) relative to sampling error.
Sensitivity analyses that place restrictions on violations of parallel
trends (\citeproc{ref-rambachanroth2023}{Rambachan and Roth 2023})
remain applicable if the measurement-induced deviation belongs to the
restriction set chosen; what Equation~\ref{eq-pretrend} adds is the
\emph{shape} of the deviation, \(\tau(k + \delta) - \ell(k)\), which is
neither linear nor smooth in general and which changes with the choice
of reference categories (the term \(\ell\) moves with \(k_0, k_1\)), so
that a restriction set chosen with outcome-side violations in mind may
or may not contain it. Detecting a pre-trend is not identifying its
cause; removing a measurement drift does not remove other confounding.
The artifact is in principle correctable from measurement-side
information, because Equation~\ref{eq-pretrend} is a deterministic
transform of the conditioning path; but under exact alignment the
correction cannot be made within the sample. Each event-time indicator
is then a tenure indicator, \(X_k = Z_{k + \delta}\), so the joint
regression of Theorem~\ref{thm-correction} with tenure indicators added
cannot separate event dynamics from conditioning: the residual-rank
condition (R) fails by construction, and the sample's own cell means
identify only the curvature of the \emph{sum} of the event path and
\(\tau\). What is required is an \emph{external} estimate of \(\tau\) on
the relevant tenure range---from a sample in which event time and tenure
are not aligned, from a refreshment design of the kind treated in the
companion paper, or from a shape library---whose identified part,
\(\tau(k+\delta) - \ell(k)\) for the analyst's reference categories, is
subtracted from the estimated coefficients. Simulation 4 reproduces
Equation~\ref{eq-pretrend} to machine precision (Appendix C) and
exhibits the rank-deficient design above; when \(MX\) is rank-deficient
the decomposition of the drift between time effects and event-time
coefficients is not unique and Equation~\ref{eq-pretrend} does not apply
as stated.

\begin{proposition}[Partial
alignment]\protect\hypertarget{prp-partial}{}\label{prp-partial}

If \(s_{it} = k_{it} + \delta_i\) with \(\delta_i\) varying across
units, the drift is no longer a function of event time alone and the
shift in the event-study coefficients is the linear projection \[
\mathbb{E}\big[\hat\beta - \hat\beta^{\circ}\big] \;=\; (X'M X)^{-1} X'M\,\tau(s) \;=\; (X'MX)^{-1} X'M Z_\tau\,\tau ,
\] which reduces to Equation~\ref{eq-pretrend} when
\(\delta_i \equiv \delta\). For a given path the shift is zero if and
only if \(X'MZ_\tau\tau = 0\), and it is zero for every path if and only
if \(X'MZ_\tau = 0\): the residualized event-time indicators must be
orthogonal to every tenure indicator. Balance of the tenure composition
of each event-time category across calendar periods does \emph{not}
suffice: with eight event dates \(E = 6, \dots, 13\), two units per date
with \(\delta_i \in \{6, 7\}\), window \(k \in [-5, 6]\) and references
\(k_0 = -1\), \(k_1 = -2\), every event-time category has half its
observations at tenure \(k + 6\) and half at \(k + 7\) in every period,
\([W_0, X]\) has full column rank 44, and yet the path
\(\tau(s) = 0.3\{1 - e^{-0.6(s-1)}\}\) shifts the pre-period
coefficients at \(k = -5, -4, -3\) by \(-0.142\), \(-0.054\), and
\(-0.014\). When \(\delta_i\) varies, the joint regression of
Theorem~\ref{thm-correction} can be feasible, since \(X\) and \(Z_\tau\)
are no longer identical, and (R) is again the condition to check.

\end{proposition}

\begin{proof}
Frisch--Waugh--Lovell with \(Z = [W, X]\) gives the projection formula,
and \(\tau(s) = Z_\tau\tau\). The zero conditions read off the formula:
\((X'MX)^{-1}\) is nonsingular, so the shift vanishes iff
\(X'MZ_\tau\tau = 0\), and for every \(\tau\) iff every column of
\(X'MZ_\tau\) vanishes. The example is verified in Appendix C; its point
is that \(M\) residualizes on unit as well as time effects, so a tenure
composition that is balanced within event time across \emph{periods} is
not thereby balanced after unit effects are removed. \(\square\)
\end{proof}

\begin{proposition}[Measurement-side sensitivity bounds for event
studies]\protect\hypertarget{prp-sensitivity}{}\label{prp-sensitivity}

Assume the conditions of Theorem~\ref{thm-pretrend} (exact alignment,
\(s = k + \delta\)), with adjacent reference categories \(k_0\) and
\(k_1 = k_0 - 1\), and suppose an external source---a staggered panel
with stride one, a refreshment design, or substantive knowledge---bounds
the curvature of the conditioning path on the tenures the window visits:
\[
\bigl|\Delta^2\tau(s)\bigr| \;\le\; C \qquad \text{for all } s \in [\,\min\mathcal{K} + \delta + 1,\ \max\mathcal{K} + \delta - 1\,].
\] Then the measurement-induced shift of every event-study coefficient
obeys \begin{equation}\phantomsection\label{eq-sensitivity}{
\bigl|\,\mathbb{E}[\hat\beta_k - \hat\beta^{\circ}_k]\,\bigr| \;=\; \bigl|\tau(k+\delta) - \ell(k)\bigr| \;\le\; C\,\frac{m_k\,(m_k+1)}{2},
\qquad
m_k \;=\; \begin{cases} k - k_0, & k > k_0,\\ k_1 - k, & k < k_1,\end{cases}
}\end{equation} and the bound is attained when \(\Delta^2\tau\) is
constant at \(+C\), or constant at \(-C\), on the tenures between the
references and \(k + \delta\): the extrapolation weights on the one side
of adjacent references all have the same sign, so constant curvature
saturates the inequality and curvature that changes sign generally does
not. Two features make this operational. The shift is invariant to the
affine direction of \(\Theta_\tau\): it is unchanged when \(\tau\) is
replaced by \(\tau + a + b(s-1)\), because \(\ell\) is the affine
interpolant and absorbs the perturbation. Hence \emph{no normalization
of the level of conditioning is required} to apply
Equation~\ref{eq-sensitivity}---only a bound on identified curvature. If
the calibrating design has stride one, satisfies C\(_1\), and observes
the tenures \(k + \delta\), \(k_0 + \delta\), and \(k_1 + \delta\), so
that its tenure-projected kernel is the affine line
(Theorem~\ref{thm-linear}(c)) and the shift is a linear functional on
its observed tenures, the shift is itself an identified functional of
\(\tau\) on that design (its weights sum to zero and annihilate
\(s - 1\)) and Equation~\ref{eq-sensitivity} can be replaced by a point
correction; stride one alone does not suffice, as the design with
entries at \(1, 4, 8\) and no follow-up shows. If the calibrating design
has stride \(d > 1\), its periodic directions are not affine in \(k\),
the shift is identified only up to them, and the bound must be widened
by the corresponding periodic range.

\end{proposition}

\begin{proof}
Write \(f(x) := \tau(x + \delta)\) and \(L[f]\) for the affine function
through \((k_0, f(k_0))\) and \((k_1, f(k_1))\), so that the shift is
\(f(k) - L[f](k)\) by Equation~\ref{eq-pretrend}. The map
\(f \mapsto f(k) - L[f](k)\) annihilates affine functions, hence factors
through second differences. With adjacent nodes \(k_1 = k_0 - 1\), write
\(D(x) := f(x+1) - f(x)\) for the first difference, so that the centered
second difference is \(\Delta^2 f(x) = D(x) - D(x-1)\). Telescoping the
first differences gives, for \(k > k_0\), \[
f(k) - L[f](k) \;=\; \sum_{x = k_0}^{k-1}\bigl\{D(x) - D(k_1)\bigr\} \;=\; \sum_{x=k_0}^{k-1}\ \sum_{y=k_0}^{x} \Delta^2 f(y),
\] so the weight on \(\Delta^2 f(y)\) is the number of
\(x \in [k_0, k-1]\) with \(x \ge y\), namely \(k - y\) for
\(y = k_0, \dots, k-1\); these are \(m_k, m_k - 1, \dots, 1\), all
positive, and sum to \(m_k(m_k+1)/2\). The case \(k < k_1\) is
symmetric, with weights \(1, \dots, m_k\) on \(\Delta^2 f(y)\) for
\(y = k+1, \dots, k_1\). Since the weights on either side share a sign,
\(|f(k) - L[f](k)| \le C \sum(\text{weights})\), with equality when
\(\Delta^2 f \equiv +C\) or \(\equiv -C\) on the relevant range.
Invariance under \(f \mapsto f + a + b x\) holds because
\(L[a + bx] = a + bx\). \(\square\)
\end{proof}

The bound quantifies the largest pre-event shift compatible with the
maintained curvature restriction. In the configuration of Simulation
4---references at \(k_0 = -1\), \(k_1 = -2\), window
\(k \in [-5, 6]\)---the constants \(m_k(m_k+1)/2\) are \(6, 3, 1\) at
\(k = -5, -4, -3\) and \(1, 10, 28\) at \(k = 0, 3, 6\), so a path whose
identified curvature is bounded by \(C = .06\) in outcome units can move
the pre-period coefficients by at most \(.36\), \(.18\), and \(.06\),
and the post-period coefficient at \(k = 6\) by at most \(1.68\). The
saturating path of Simulation 4 has maximal centered curvature \(.061\),
slightly above \(.06\), for which the first bound is \(.37\); its actual
shifts are \(.18\), \(.07\), and \(.02\). The bound is conservative for
that path by a factor of two to three because its curvature decays
geometrically away from entry rather than staying at its maximum; the
factor is specific to the path, not a property of the bound. Two
practical readings follow, and both concern expectations rather than
realized estimates: Equation~\ref{eq-sensitivity} bounds the conditional
expected measurement-induced shift, not any single sample coefficient,
and it is read under the premise of Theorem~\ref{thm-pretrend} that the
expected unconditioned pre-event coefficient is zero. Under that
premise, a confidence interval for the expected reported coefficient
that excludes the entire range of admissible measurement-induced shifts,
\([-C\,m_k(m_k+1)/2,\; C\,m_k(m_k+1)/2]\), is evidence against drift of
the maintained magnitude as the sole explanation of a pre-trend; a point
estimate outside that range is not by itself sufficient, since sampling
error alone can place it there when the expected shift respects the
bound, and if \(C\) is itself estimated its uncertainty enters as well.
Whether the comparison concerns one prespecified coefficient or a
jointly assessed set should be stated in advance, the latter calling for
simultaneous rather than pointwise intervals. A coefficient whose
interval lies inside the range cannot be separated from measurement
drift without an external estimate of the path. And because the bound
grows quadratically with distance from the reference categories, distant
leads and lags permit larger measurement-induced shifts under the same
curvature restriction; that is a reason to report the bound alongside
the coefficients when choosing the event window, not a statement that
the realized artifact is largest there.

\subsection{Which functionals survive}\label{sec-survive}

\begin{proposition}[Orthogonality
characterization]\protect\hypertarget{prp-survive}{}\label{prp-survive}

A linear panel functional
\(\theta = \sum_{it}\lambda_{it}\,\mathbb{E}[Y_{it}]\) is unbiased for
its unconditioned counterpart
\(\theta^{\circ} = \sum_{it}\lambda_{it}\,\mathbb{E}[Y^{\circ}_{it}]\)
under every drift \(\tau\) in a class \(\mathcal{M}\) if and only if
\(\sum_{it}\lambda_{it}\,\tau(s_{it}) = 0\) for all
\(\tau \in \mathcal{M}\). For \(\mathcal{M}\) the affine paths this is
the pair of moment conditions \(\sum_{it}\lambda_{it} = 0\) and
\(\sum_{it}\lambda_{it}\,s_{it} = 0\), satisfied by every estimator that
partials out unit and time effects; for \(\mathcal{M}\) unrestricted it
is tenure-composition balance, \(\sum_{it:\,s_{it}=s}\lambda_{it} = 0\)
for every \(s\) in the support.

\end{proposition}

\begin{proof}
\(\theta - \theta^{\circ} = \sum_{it}\lambda_{it}\tau(s_{it})\) by
Equation~\ref{eq-report} and D1; the ``if and only if'' is the
definition of unbiasedness for all \(\tau \in \mathcal{M}\). For affine
\(\mathcal{M}\) the condition is linear in \((a,b)\) and reduces to its
two coefficients; for unrestricted \(\mathcal{M}\) take
\(\tau = \mathbf{1}\{s = s_0\}\) for each \(s_0\). The TWFE weights
satisfy the affine pair by the proof of Theorem~\ref{thm-absorb}.
\(\square\)
\end{proof}

The second condition is the bridge to the design problem treated in
companion work: a published aggregate that is balanced on time in sample
is a linear functional satisfying the composition condition, so it is
drift-free, and \emph{the same functional} carries no information about
the drift. This is a statement about one functional at a time---a robust
aggregate estimator does not imply that the survey's
rotation-group-level data are uninformative about conditioning, only
that the aggregate is; the companion paper's balance--information
theorem makes the distinction precise.

\section{Monte Carlo Design (ADEMP)}\label{sec-mc}

\emph{Aims}: noiseless verification of the identities of
Theorem~\ref{thm-linear} and Corollary~\ref{cor-curvature} on the
JLPS-calibrated support (stride \(d = 4\)); size and power of the
plateau observable-implications test; coverage of the drift-bound
intervals of Proposition~\ref{prp-bounds} at the covered tenures;
dependence of level estimates on the normalization.
\emph{Data-generating mechanisms}: 19-wave staggered panel calibrated to
JLPS (three cohorts, entry at \(t = 1, 5, 13\)); conditioning paths
(null, saturating, linear-drift, non-monotone); cohort effects (null,
drifting, shock at 2011 entry). \emph{Estimands}: the centered lag-4
second difference \(\Delta_4^2\tau(5) = \tau(9) - 2\tau(5) + \tau(1)\)
and the identified contrast \(\Delta^2\tau(2) - \Delta^2\tau(6)\);
\(\bar\tau - \tau(1)\) under the plateau restriction; drift bounds for
\(\tau(13) - \tau(1)\) and \(\tau(17) - \tau(1)\) at
\(M \in \{0.4, 1, 2, 4\} \times\) true maximal drift. \emph{Methods}:
explicit cell-mean contrasts and regression readouts from arbitrary
representatives; the plateau-constrained estimator with its
observable-implications test; drift-bound intervals with the
Imbens--Manski critical value. \emph{Performance measures}: bias, MCSE,
coverage of the true parameter value by the drift-bound intervals,
specification-test rejection. Simulation 4 (Appendix C) verifies the
identities of Section~\ref{sec-downstream} (absorption, invariance, the
identified correction, and the exact pre-trend formula) to machine
precision and reports the finite-sample behaviour of the corrected
estimator. The full protocol, with prespecified versus exploratory
analyses flagged, is Appendix B (ADEMP-PreReg format); results are in
Appendix A.

\section{Illustration: A 19-Wave Panel with Two Refreshment
Samples}\label{sec-jlps}

This section illustrates the design calculus; it reports no new
empirical estimates. The Japanese Life Course Panel Survey (JLPS)
enrolled cohorts in 2007, 2011, and 2019---waves 1, 5, and 13---and has
been observed through wave 19. Its stride is \(d = 4\), its second
cohort entered \(\Delta_2 = 4 = d\) waves after the first, so
Lemma~\ref{lem-trapezoid}(i) holds for every \(T\) and C\(_4\) is
satisfied; its cell design has 41 cells, 39 free parameters, rank 35,
and a four-dimensional kernel, exactly the affine-plus-periodic family
(Figure~\ref{fig-support}). By Corollary~\ref{cor-curvature} the
curvature of conditioning is identified only at lag four, and the drift
bounds of Proposition~\ref{prp-bounds} apply at tenures 5, 9, 13, and
17, where the periodic component vanishes. Those are tenures---arguments
of \(\tau\), common to every cohort, and equal to the wave number only
for the original cohort---not calendar waves: the refreshment samples
themselves entered at waves 5 and 13. The design is thus matched to
lag-four contrasts and to levels on that tenure grid, which are the
objects the companion papers estimate, and poorly matched to the
ordinary second differences an analyst might have expected to read off
it. Two features of the design bear on Section~\ref{sec-recovery}. The
entry negative-control set is unusually rich---23 time-invariant
childhood-circumstance items asked once, at entry, of every cohort---so
the affine direction can be anchored under the conditions of
Proposition~\ref{prp-nc}; but with \(d = 4\) the battery leaves the
three periodic directions untouched, and the path at tenures
\(s \not\equiv 1 \pmod 4\) is recovered only under the plateau
restriction of Proposition~\ref{prp-plateau} or an external anchor. And
the 2019 entrants share no birth-cohort support with the incumbents, so
that episode anchors increments rather than levels---the configuration
treated fully in the companion refreshment-theory paper. The companion
empirical paper implements these recovery conditions on the restricted
microdata under the archive's disclosure rules and reports the resulting
battery contrasts and conditioning estimates; we do not reproduce those
numbers here, and the present paper's claims do not depend on them.

\subsection{A descriptive illustration from published rotation-group
indices}\label{sec-cps}

Corollary~\ref{cor-cps} says that, under (P\('\)), a rotation design
identifies differences of per-month rates of conditioning from the joint
cohort-by-period array of cell means. Statistical agencies do not
publish that array; they publish \emph{month-in-sample indices}, each an
average over many periods of the ratio of the month-in-sample-\(k\)
estimate to a benchmark. This subsection computes, from those indices,
the transformation that Corollary~\ref{cor-cps} would apply to the
array, and reports the result as what it is: a descriptive profile of
the published indices, not an estimate of conditioning. The distinction
is not pedantic. The month-in-sample-\(k\) mean at period \(t\) is
\(\alpha(t) + g(t - j_k) + \tau(k)\), so a contrast of month-in-sample
means carries the cohort effects \(g\) at eight different entry dates;
the weights of \(\mathcal{D}\) annihilate an affine \(g\) but not a
curved one, and a quadratic cohort effect \(g(e) = a(e^2 - e_0^2)\)
contributes \(126a\) to the contrast at every period, which averaging
over periods leaves intact (Corollary~\ref{cor-cps}(ii)). Isolating
conditioning from these indices would require estimating \(g\) jointly,
from cohort-by-period data, together with the sampling covariances of
the published ratios, the benchmark constructions of the two sources
(which differ), and a justification of the additive model on the log
scale. None of that is available here, and the readings below claim none
of it.

Write \(b(k)\) for the published index at month in sample \(k\) and
\(y(k) = \log b(k)\). With the CPS pattern
\(j = (0,1,2,3,12,13,14,15)\), the transformation of
Corollary~\ref{cor-cps} produces the profile of per-month changes \[
r_k \;=\; \frac{y(k+1) - y(k)}{j(k+1) - j(k)}, \qquad k = 1, \dots, 7,
\] and the contrast \(\mathcal{D}\) with the same weights as in
Corollary~\ref{cor-cps}. Two \emph{restricted descriptive patterns} can
be compared with the profile: a constant change per calendar month,
under which all seven rates are equal and \(\mathcal{D} = 0\); and a
constant change per interview, under which the six within-block rates
are equal and the rate across the rest period is one ninth of them.
Comparison is by inspection; the published tables carry no covariances,
so nothing here is a test.

\needspace{0.3\textheight}
\begin{longtable}[]{@{}
  >{\raggedright\arraybackslash}p{(\linewidth - 16\tabcolsep) * \real{0.3600}}
  >{\raggedleft\arraybackslash}p{(\linewidth - 16\tabcolsep) * \real{0.0800}}
  >{\raggedleft\arraybackslash}p{(\linewidth - 16\tabcolsep) * \real{0.0800}}
  >{\raggedleft\arraybackslash}p{(\linewidth - 16\tabcolsep) * \real{0.0800}}
  >{\raggedleft\arraybackslash}p{(\linewidth - 16\tabcolsep) * \real{0.0800}}
  >{\raggedleft\arraybackslash}p{(\linewidth - 16\tabcolsep) * \real{0.0800}}
  >{\raggedleft\arraybackslash}p{(\linewidth - 16\tabcolsep) * \real{0.0800}}
  >{\raggedleft\arraybackslash}p{(\linewidth - 16\tabcolsep) * \real{0.0800}}
  >{\raggedleft\arraybackslash}p{(\linewidth - 16\tabcolsep) * \real{0.0800}}@{}}
\caption{Per-month changes in the logarithms of published CPS
month-in-sample indices, in log points per month, with \(\mathcal{D}\)
formed with the weights of Corollary~\ref{cor-cps}. The profile is a
descriptive transformation of published averages: it retains any
non-affine cohort effect and carries no sampling covariances, and it is
not an estimate of conditioning. \(^{\dagger}\) In 1968--69 the CPS
asked additional questions of persons not in the labor force \emph{only}
at months in sample 1 and 5 (\citeproc{ref-bailar1975}{Bailar 1975}), so
this row's instrument varies with the dose and assumption M1 fails; it
is shown to make the failure visible.}\label{tbl-cps}\tabularnewline
\toprule\noalign{}
\begin{minipage}[b]{\linewidth}\raggedright
Source and item
\end{minipage} & \begin{minipage}[b]{\linewidth}\raggedleft
\(r_1\)
\end{minipage} & \begin{minipage}[b]{\linewidth}\raggedleft
\(r_2\)
\end{minipage} & \begin{minipage}[b]{\linewidth}\raggedleft
\(r_3\)
\end{minipage} & \begin{minipage}[b]{\linewidth}\raggedleft
\(r_4\) (rest)
\end{minipage} & \begin{minipage}[b]{\linewidth}\raggedleft
\(r_5\)
\end{minipage} & \begin{minipage}[b]{\linewidth}\raggedleft
\(r_6\)
\end{minipage} & \begin{minipage}[b]{\linewidth}\raggedleft
\(r_7\)
\end{minipage} & \begin{minipage}[b]{\linewidth}\raggedleft
\(\mathcal{D}\)
\end{minipage} \\
\midrule\noalign{}
\endfirsthead
\toprule\noalign{}
\begin{minipage}[b]{\linewidth}\raggedright
Source and item
\end{minipage} & \begin{minipage}[b]{\linewidth}\raggedleft
\(r_1\)
\end{minipage} & \begin{minipage}[b]{\linewidth}\raggedleft
\(r_2\)
\end{minipage} & \begin{minipage}[b]{\linewidth}\raggedleft
\(r_3\)
\end{minipage} & \begin{minipage}[b]{\linewidth}\raggedleft
\(r_4\) (rest)
\end{minipage} & \begin{minipage}[b]{\linewidth}\raggedleft
\(r_5\)
\end{minipage} & \begin{minipage}[b]{\linewidth}\raggedleft
\(r_6\)
\end{minipage} & \begin{minipage}[b]{\linewidth}\raggedleft
\(r_7\)
\end{minipage} & \begin{minipage}[b]{\linewidth}\raggedleft
\(\mathcal{D}\)
\end{minipage} \\
\midrule\noalign{}
\endhead
\bottomrule\noalign{}
\endlastfoot
McIllece (\citeproc{ref-mcillece2022}{2022}, Tab. 2): unemployed,
2003--22 & \(-.055\) & \(-.038\) & \(-.025\) & \(+.001\) & \(-.046\) &
\(-.017\) & \(-.001\) & \(+.510\) \\
McIllece (\citeproc{ref-mcillece2022}{2022}, Tab. 2): employed, 2003--22
& \(.000\) & \(-.001\) & \(-.005\) & \(+.000\) & \(-.010\) & \(+.001\) &
\(+.003\) & \(+.002\) \\
Bailar (\citeproc{ref-bailar1975}{1975}, Tab. 1): hours 35--40, 1968--69
& \(+.052\) & \(+.022\) & \(+.016\) & \(-.003\) & \(+.032\) & \(+.011\)
& \(+.008\) & \(-.502\) \\
Bailar (\citeproc{ref-bailar1975}{1975}, Tab. 1): hours 35--40, 1970--72
& \(+.048\) & \(+.022\) & \(+.016\) & \(-.003\) & \(+.025\) & \(+.014\)
& \(+.008\) & \(-.457\) \\
Bailar (\citeproc{ref-bailar1975}{1975}, Tab. 1): unemployed, 1968--69
\(^{\dagger}\) & \(-.167\) & \(-.052\) & \(-.038\) & \(+.018\) &
\(-.125\) & \(-.041\) & \(-.017\) & \(+1.671\) \\
\end{longtable}

Three readings, all descriptive. First, the displayed unemployment
profile is inconsistent with a constant per-month pattern: the rates
range over \(.057\) log points per month and \(\mathcal{D} = +.51\)
where that pattern gives zero. Second, it is also inconsistent with a
constant per-interview pattern, which would imply equal within-block
increments: within each four-month block the rate attenuates
(\(-.055, -.038, -.025\); \(-.046, -.017, -.001\)), and the second
block's first rate is nearly as steep as the first block's. This is the
``pattern \ldots{} also apparent within each four-month wave'' that
McIllece (\citeproc{ref-mcillece2022}{2022}) reports; whether it
reflects conditioning, cohort composition, or the benchmark construction
cannot be decided from the indices. Third, the employment index varies
little across months in sample: its rates are an order of magnitude
smaller and \(\mathcal{D} = +.002\). A small profile is not evidence
that conditioning is absent, since the profile retains cohort effects of
either sign; it is evidence only that the published index is nearly
flat.

The last row of the table is a caution, not a finding. Because the
1968--69 CPS asked additional labor-force probes only of households in
their first and fifth months in sample, the instrument varied with the
dose---a direct violation of the clause of M1 that all cohorts face the
same instrument at a given time---and the resulting profile shows it:
the two intervals that step off or onto a probed month in sample,
\(1 \to 2\) and \(4 \to 5\), are precisely the outliers---the steepest
rate in the table and the only positive one. Solon
(\citeproc{ref-solon1986}{1986}, Tab. 1) shows the same artifact after
the probes were moved to months in sample 4 and 8 in 1970: in his
1974--83 averages the only two positive within-block rates are the
intervals \(3 \to 4\) and \(7 \to 8\), which are exactly the two that
end at a probed month. The framework localizes the damage; in a joint
cohort-by-period analysis it would also say which contrasts avoid the
affected months while annihilating \(j(\cdot)\). The descriptive
analogue in Solon (\citeproc{ref-solon1986}{1986})'s averages---the rate
from month three to month five against the rate from month two to month
three, which differ by \(14.9\) thousand persons per month---is again a
transformation of published averages, not an identified contrast.

\section{Discussion}\label{discussion}

\textbf{What the theorems change in practice.} Five practical
implications follow. First, any reported \emph{level} of panel
conditioning or rotation-group bias should be read jointly with its
normalization; comparisons of reported levels across studies that
normalize differently reflect those normalization choices as well as any
genuine difference between surveys, and cannot separate the two without
the identifying restriction each study used. The discussion of
normalization choices after Lemma~\ref{lem-trapezoid} reads the major
published conventions in these terms. Second, the \emph{identified
shape} of conditioning---curvature at the design's stride, and ordinary
curvature when the stride is one and C\(_1\) holds---is the exportable,
normalization-free quantity; meta-analyses and bias-prior libraries
should accumulate identified contrasts rather than level estimates
without their identifying restrictions, and should record the stride of
the design that produced each one. Third, fixed-effects users can
estimate Equation~\ref{eq-joint}, the two-way regression with tenure
indicators, rather than either ignoring conditioning or choosing a
level, after checking the residual-rank condition (R) directly;
event-study users on tenure-aligned designs should treat pre-period
coefficients as suspect until the measurement-side check of
Theorem~\ref{thm-pretrend} has been run. Fourth, rotation designs with
continuous recruitment should analyze month-in-sample differences on the
joint cohort-by-period array, with cohort effects estimated rather than
averaged away, and report them as differences of per-month rates: under
(P\('\)) those differences are the identified functionals
(Theorem~\ref{thm-dose}, Corollary~\ref{cor-cps}), and the rest period
between interview blocks is what lets a constant per-interview increment
be told apart from a constant per-month one. Published month-in-sample
indices, being averages over periods, retain the cohort effects and
support only descriptive comparison (Section~\ref{sec-cps}). Fifth,
before fielding or analyzing a staggered design, compute the kernel of
its cell design matrix (Theorem~\ref{thm-linear}(a), (d)): the
refreshment schedule and the follow-up of its last cohort fix which
functionals of conditioning the panel can ever deliver, a schedule with
a common divisor larger than one forfeits ordinary curvature, and a
cohort followed too briefly forfeits more (Lemma~\ref{lem-trapezoid}).

\textbf{Relation to companion work.} The identification theory of
refreshment designs under attrition---the sharp identified set under a
single refreshment, and conditions for point identification with
refreshment samples at multiple tenures---is developed in companion
work, where the JLPS's two refreshment cohorts (2011, 2019) with a
repeated entry-wave negative-control battery constitute, to our
knowledge, the first dataset satisfying those conditions.
Section~\ref{sec-downstream} characterizes the propagation of
\emph{additive} tenure drift under D1--D2 and the stated rank
conditions; its extension to report-dependent (state-dependent)
measurement error, where the drift depends on the previous report and is
not absorbed by fixed effects, is the subject of a companion paper on
transition estimators. The same companion work develops the design
problem: which rotation, refreshment, and mode-transition schedules make
conditioning identifiable at all, and at what cost---including the
result that rotation designs balanced on time in sample hold the tenure
component of published aggregates constant over time at the price of the
information from which it could have been estimated.

\textbf{Limits.} The model is additive throughout, in two places: cell
means are additive in period, cohort, and conditioning (M4), and the
individual report is the unconditioned outcome plus a common path (D1).
Both are the standard assumptions of the rotation-group-bias and
panel-conditioning literatures, and both are where these results end.
With cohort-by-tenure or period-by-tenure interactions the indeterminacy
is item-wise or cohort-wise and larger; with conditioning that depends
on the previous \emph{report} rather than on the count of
interviews---state-dependent measurement error---the drift is not a
function of the design alone, two-way fixed effects do not absorb it,
and the companion paper on transition estimators takes it up. Nothing
here is a general theory of measurement error, and the results should be
read as statements about an additively separable conditioning path;
Solon (\citeproc{ref-solon1986}{1986}) argued, from the rotation-group
evidence itself, that the additive model of rotation-group bias is not
consistent with the data for unemployment, and proposed a multiplicative
alternative---on published multiplicative indices the natural reading of
our model is that it applies to the logarithm of the index, which is
what we do in Section~\ref{sec-cps}. Our impossibility results concern
the additive cohort-by-period cell-mean model without auxiliary
identifying information; they do not bind designs that bring in
validation records, paradata with exclusion restrictions, or randomized
question exposure, which is precisely where the recovery conditions
point. And the additive cell-mean structure, while standard, is a
choice: with item-by-tenure interactions the indeterminacy is item-wise,
which strengthens rather than weakens the case for shape-based
reporting.

\section*{Related literature}\label{related-literature}
\addcontentsline{toc}{section}{Related literature}

Our results draw on four literatures: the rotation-group-bias tradition
in official statistics (\citeproc{ref-bailar1975}{Bailar 1975};
\citeproc{ref-solon1986}{Solon 1986};
\citeproc{ref-krueger2017}{Krueger, Mas, and Niu 2017};
\citeproc{ref-mcillece2022}{McIllece 2022}), where normalizations are
ubiquitous---and where, by Theorem~\ref{thm-dose}, the interrupted
rotation pattern that tradition works with is what lets a constant
per-interview increment be separated from a constant per-month one,
under (P\('\)) and once cohort effects are estimated jointly; the
panel-conditioning literature in survey methodology
(\citeproc{ref-warren2012}{Warren and Halpern-Manners 2012};
\citeproc{ref-halpernmanners2012}{Halpern-Manners and Warren 2012};
\citeproc{ref-halpernmanners2017}{Halpern-Manners, Warren, and Torche
2017}; \citeproc{ref-kraemer2024}{Kraemer et al. 2024},
\citeproc{ref-kraemer2025}{2025}), which documents effects, and within
which Das, Toepoel, and van Soest (\citeproc{ref-dasetal2011}{2011})
give a formal identification analysis of conditioning and attrition for
a binary item in a two-wave design with a refreshment sample---the
closest antecedent to our Section~\ref{sec-equivalence} and the
companion refreshment paper; the APC identification literature
(\citeproc{ref-masetal1973}{Mason et al. 1973};
\citeproc{ref-fosse2019}{Fosse and Winship 2019a},
\citeproc{ref-fosse2019bounds}{2019b}), whose linear dependency we
transplant to the measurement domain, and whose unequal-interval
extension (\citeproc{ref-gascoignesmith2023}{Gascoigne and Smith 2023})
already contains the periodic ambiguity that a common refreshment stride
produces; and the modern event-study literature
(\citeproc{ref-sunabraham2021}{Sun and Abraham 2021};
\citeproc{ref-borusyak2024}{Borusyak, Jaravel, and Spiess 2024}), whose
fully dynamic underidentification and contamination results have
analogues here under M4 and D1---with tenure in the role of event time
and the survey instrument in the role of the outcome equation; the
closest antecedents to Theorem~\ref{thm-pretrend} are their treatments
of excluded periods, which we use in stating its premise. Feng, Hu, and
Sun (\citeproc{ref-fenghusun2022}{2022}) point-identify CPS
misclassification probabilities in a latent-state model with
history-dependent errors; their object and model class differ from the
additive mean model here, and we make no claim that their restrictions
select a point of \(\Theta_\tau\). What the two share is the dependency
between tenure, period, and cohort that both must confront; the
entry-wave negative controls used here check the comparability of
entering groups on time-invariant facts, one component of the
cross-rotation-group comparability their transport step requires, and
not that condition in full. Bertoli, Jakli, and Pascoe
(\citeproc{ref-bertoli2026}{2026}) formalize conditioning bias in
two-wave event-impact surveys and propose a dual-randomization design;
the present results address the many-wave staggered case under the
additive cell-mean model, a different model class from theirs, in which
the design choices they analyze for two waves recur at every
refreshment.

\section*{Appendix A: Simulation
results}\label{appendix-a-simulation-results}
\addcontentsline{toc}{section}{Appendix A: Simulation results}

All simulations use the JLPS-calibrated support---cohorts entering
\(t = 1, 5, 13\), \(T = 19\), stride \(d = 4\)---so that the identified
set of Theorem~\ref{thm-linear} is four-dimensional and ordinary second
differences are not identified. The specifications below correspond to
the accompanying code; every number quoted in this appendix is read from
a named column of the corresponding output file, and the replication
package's \texttt{check\_manuscript\_values.R} reproduces each of them
from those files.

\subsection*{Simulation 1: noiseless identities, normalizations, and the
plateau
test}\label{simulation-1-noiseless-identities-normalizations-and-the-plateau-test}
\addcontentsline{toc}{subsection}{Simulation 1: noiseless identities,
normalizations, and the plateau test}

Design: \(B = 500\); cell noise \(1/\sqrt{1000}\) (one thousand
respondents per cell with unit outcome variance); scenarios S1 null / S2
saturating, \(\tau(s) = 0.2\{1 - e^{-0.6(s-1)}\}\) / S3 = S2 + linear
drift \(0.01(s-1)\) / S4 = S2 + late nonlinear drift
\(0.004\max(s-6,0)^2\); cohort shock G1 (\(+0.15\)) on the 2011 cohort;
plateau window \(s \ge 6\).

\emph{Noiseless identity tests.} On population cell means, adding the
periodic null direction \(\rho = (0, 0.30, -0.10, 0.20)\) (period four,
\(\rho(1) = 0\)) to \(\tau\) and subtracting \(\rho(t)\) from \(\alpha\)
changes no cell mean (\(5.6 \times 10^{-17}\)). The ordinary second
difference \(\Delta^2\tau(3)\) moves from \(-0.022\) to \(+0.678\); the
centered lag-4 second difference
\(\Delta_4^2\tau(5) = \tau(9) - 2\tau(5) + \tau(1)\) does not move; and
the explicit cell-mean contrast
\([\mu(1,9) - \mu(5,9)] - [\mu(1,5) - \mu(5,5)]\)---first argument the
entry date, as in Equation~\ref{eq-model}; the cohorts entering at
\(t = 1\) and \(t = 5\) compared at \(t = 9\) and at \(t = 5\)---returns
\(\Delta_4^2\tau(5) = \tau(9) - 2\tau(5) + \tau(1) = -0.1654\) to
\(10^{-17}\) before and after the perturbation. (The same two cohorts
compared at \(t = 13\) and \(t = 9\) would return
\(\Delta_4^2\tau(9) = \tau(13) - 2\tau(9) + \tau(5) = -0.0150\), a
different identified functional.) This is Theorem~\ref{thm-linear} and
Corollary~\ref{cor-curvature} at machine precision. Earlier versions
labeled these two contrasts by their first tenure, \(\Delta_4^2\tau(1)\)
and \(\Delta_4^2\tau(5)\); the quantities are unchanged, and the labels
now follow the centered convention of Corollary~\ref{cor-curvature}.

\begin{longtable}[]{@{}
  >{\raggedright\arraybackslash}p{(\linewidth - 14\tabcolsep) * \real{0.1250}}
  >{\raggedright\arraybackslash}p{(\linewidth - 14\tabcolsep) * \real{0.1250}}
  >{\raggedright\arraybackslash}p{(\linewidth - 14\tabcolsep) * \real{0.1250}}
  >{\raggedright\arraybackslash}p{(\linewidth - 14\tabcolsep) * \real{0.1250}}
  >{\raggedright\arraybackslash}p{(\linewidth - 14\tabcolsep) * \real{0.1250}}
  >{\raggedright\arraybackslash}p{(\linewidth - 14\tabcolsep) * \real{0.1250}}
  >{\raggedright\arraybackslash}p{(\linewidth - 14\tabcolsep) * \real{0.1250}}
  >{\raggedright\arraybackslash}p{(\linewidth - 14\tabcolsep) * \real{0.1250}}@{}}
\toprule\noalign{}
\begin{minipage}[b]{\linewidth}\raggedright
Scenario
\end{minipage} & \begin{minipage}[b]{\linewidth}\raggedright
bias level N1
\end{minipage} & \begin{minipage}[b]{\linewidth}\raggedright
bias level N2
\end{minipage} & \begin{minipage}[b]{\linewidth}\raggedright
\(\Delta_4^2\tau(5)\) true
\end{minipage} & \begin{minipage}[b]{\linewidth}\raggedright
bias (contrast)
\end{minipage} & \begin{minipage}[b]{\linewidth}\raggedright
MC sd
\end{minipage} & \begin{minipage}[b]{\linewidth}\raggedright
bias \(\bar\tau\) (plateau)
\end{minipage} & \begin{minipage}[b]{\linewidth}\raggedright
rejection
\end{minipage} \\
\midrule\noalign{}
\endhead
\bottomrule\noalign{}
\endlastfoot
S1 null & 0.000 & 0.000 & 0.000 & 0.001 & .064 & 0.001 & .058 \\
S2 saturating & −0.051 & −0.102 & −0.165 & 0.004 & .061 & −0.006 &
.048 \\
S3 + linear drift & −0.089 & −0.140 & −0.165 & 0.001 & .065 & −0.097 &
.054 \\
S4 + nonlinear & −0.173 & −0.124 & −0.129 & 0.005 & .060 & −0.408 &
1.000 \\
\end{longtable}

Rows are the G0 (no cohort shock) setting; the G1 rows of
\texttt{sim1\_results.csv} differ from them only within Monte Carlo
error. Three theorem-consistent facts: (i) level estimates differ by
normalization; (ii) the estimated bias of the identified lag-4 contrast
is small relative to its Monte Carlo uncertainty in every scenario
(Monte Carlo standard error \(.003\)); (iii) the plateau
\emph{observable-implications} test (constancy over time of the cohort
gaps within the window) has rejection rates near the nominal \(.05\) in
S1 (null), S2 (saturating path) and S3 (S2 plus the pure linear drift
that Proposition~\ref{prp-plateau} declares undetectable)---\(.058\),
\(.048\) and \(.054\) over \(500\) replications---and rejects in every
replication under the nonlinear violation S4 (499 of 500 in the G1
setting), while the plateau-constrained estimator absorbs the drift it
cannot see (bias \(-0.097\) under S3, which adds a linear drift to the
saturating path of S2). {[}sims/sim1\_normalizations.R;
sim1\_results.csv{]}

\subsection*{Simulation 2: identified and unidentified curvature
readouts}\label{simulation-2-identified-and-unidentified-curvature-readouts}
\addcontentsline{toc}{subsection}{Simulation 2: identified and
unidentified curvature readouts}

The target is \(\theta := \Delta^2\tau(2) - \Delta^2\tau(6)\), an
identified functional on this support---its weights sum to zero within
every residue class modulo four and are orthogonal to
\(s - 1\)---contrasted with \(\Delta^2\tau(3)\), which is not identified
(Corollary~\ref{cor-curvature}). For noiseless cell means a saturated
regression returns a representative of \(\Theta_\tau\): unidentified
functionals of it depend on which representative the software chose,
whereas identified functionals equal their true values.

\emph{Noiseless identities.} \(\theta = -0.03702\); the correctly
centred contrast \(-\Delta^2 D(6)\), \(D(t) = \mu(1,t) - \mu(5,t)\)
(first argument the entry date: the cohorts entering at \(t = 1\) and
\(t = 5\)), returns it to \(10^{-16}\); the regression readout
\(\Delta^2\hat\tau(2) - \Delta^2\hat\tau(6)\) from an arbitrary
representative returns it to \(10^{-15}\); the regression readout of
\(\Delta^2\hat\tau(3)\) is \(-0.072\) against a true \(-0.022\), and
moves with the representative.

\begin{longtable}[]{@{}
  >{\raggedright\arraybackslash}p{(\linewidth - 14\tabcolsep) * \real{0.1250}}
  >{\raggedright\arraybackslash}p{(\linewidth - 14\tabcolsep) * \real{0.1250}}
  >{\raggedright\arraybackslash}p{(\linewidth - 14\tabcolsep) * \real{0.1250}}
  >{\raggedright\arraybackslash}p{(\linewidth - 14\tabcolsep) * \real{0.1250}}
  >{\raggedright\arraybackslash}p{(\linewidth - 14\tabcolsep) * \real{0.1250}}
  >{\raggedright\arraybackslash}p{(\linewidth - 14\tabcolsep) * \real{0.1250}}
  >{\raggedright\arraybackslash}p{(\linewidth - 14\tabcolsep) * \real{0.1250}}
  >{\raggedright\arraybackslash}p{(\linewidth - 14\tabcolsep) * \real{0.1250}}@{}}
\toprule\noalign{}
\begin{minipage}[b]{\linewidth}\raggedright
Scenario
\end{minipage} & \begin{minipage}[b]{\linewidth}\raggedright
\(\theta\) true
\end{minipage} & \begin{minipage}[b]{\linewidth}\raggedright
bias (contrast)
\end{minipage} & \begin{minipage}[b]{\linewidth}\raggedright
MC sd
\end{minipage} & \begin{minipage}[b]{\linewidth}\raggedright
bias (regression readout)
\end{minipage} & \begin{minipage}[b]{\linewidth}\raggedright
MC sd
\end{minipage} & \begin{minipage}[b]{\linewidth}\raggedright
\(\Delta^2\tau(3)\) true
\end{minipage} & \begin{minipage}[b]{\linewidth}\raggedright
mean readout
\end{minipage} \\
\midrule\noalign{}
\endhead
\bottomrule\noalign{}
\endlastfoot
S2 & −0.037 & 0.006 & .110 & 0.000 & .080 & −0.022 & −0.075 \\
S3 & −0.037 & 0.003 & .107 & 0.003 & .081 & −0.022 & −0.072 \\
\end{longtable}

The Monte Carlo estimates of both estimators are consistent with
unbiasedness for this identified functional (Monte Carlo standard error
\(.005\)); the regression readout is the more efficient because it uses
every cell. The recommendation: read identified functionals off any
representative, but check identification first.
{[}sims/sim2\_direct\_contrast.R; sim2\_results.csv{]}

\subsection*{\texorpdfstring{Simulation 3: drift-bound intervals
(Proposition~\ref{prp-bounds})}{Simulation 3: drift-bound intervals (Proposition~)}}\label{simulation-3-drift-bound-intervals-prp-bounds}
\addcontentsline{toc}{subsection}{Simulation 3: drift-bound intervals
(Proposition~\ref{prp-bounds})}

Design: same support; saturating \(\tau(s) = 0.25\{1 - e^{-(s-1)/2}\}\);
cohort effects \(g = (0, 0.10, 0.22)\), per-period drifts
\((0.025, 0.015)\), so the true maximal drift is \(M_0 = 0.025\);
\(n = 500\) per cell, \(B = 300\), endpoint standard errors by
parametric bootstrap (\(150\) draws). The interval is
\([\hat\tau_{lo} - c\,\hat\sigma_{lo},\ \hat\tau_{hi} + c\,\hat\sigma_{hi}]\)
with a critical value \(c\) in the form of Imbens and Manski
(\citeproc{ref-imbensmanski2004}{2004}), solving
\(\Phi(c + \hat\Delta/\hat\sigma) - \Phi(-c) = 0.95\), where
\(\hat\Delta\) is the estimated width of the identified set
\emph{truncated at zero} and \(\hat\sigma\) the larger endpoint standard
error (a standard error, unlike the asymptotic standard deviation that
plays this role in Imbens and Manski
(\citeproc{ref-imbensmanski2004}{2004})), so that \(c \to 1.645\) for a
wide set and \(c \to 1.96\) for a point; the endpoints themselves are
the raw estimates and may cross, in which case the estimated set is
empty and the rate is recorded. The target is coverage of the true
parameter value, not of the whole set. The bounds are evaluated at
\(s = 13\) and \(s = 17\), the tenures congruent to 1 modulo 4 at which
Proposition~\ref{prp-bounds} applies. Two further designs probe the
narrow-set case. G3 sets the true drifts to \((+0.025, -0.025)\): at
\(M = M_0\) both constraints bind and the identified set for \(m\) is a
point, but its two endpoints are extrema over \emph{distinct} drifts, so
each endpoint estimator is a smooth function of the cell means---the
population endpoints coincide without any tie inside an extremum. G4
sets the true drifts to \((+0.025, +0.025)\): both extrema are tied, the
upper endpoint \(M - \max_k \hat d_k\) subtracts the maximum of two
equal quantities and is not smooth, and at \(M = M_0\) the identified
set for \(m\) is an interval of width
\(2M_0 - (\max_k d_k - \min_k d_k) = 2M_0 = .050\), with the truth at
its upper endpoint; the corresponding population widths for \(\tau(13)\)
and \(\tau(17)\) are \(12 \times .050 = .600\) and
\(16 \times .050 = .800\) (\texttt{check\_recovery\_support.py} and
\texttt{check\_manuscript\_values.R} assert these).

At the population level the representative's implied drifts under G2 are
\((0.0117, 0.0017)\), each displaced from the truth by the same
constant---the linear part the representative reassigns to the cohort
dimension---while their difference is preserved; on the subgrid
\(s \in \{1, 5, 9, 13, 17\}\) the representative agrees with the truth
up to that linear part to \(10^{-16}\). At \(M = M_0\) the sharp
interval for \(\tau(13)\) is \([-0.23, 0.25]\) and contains the truth
\(0.249\) at its upper endpoint (the true drift equals \(M_0\)); at
\(M = 0.4M_0\) it is \([-0.05, 0.07]\) and excludes it.

\begin{longtable}[]{@{}
  >{\raggedright\arraybackslash}p{(\linewidth - 16\tabcolsep) * \real{0.1111}}
  >{\raggedright\arraybackslash}p{(\linewidth - 16\tabcolsep) * \real{0.1111}}
  >{\raggedright\arraybackslash}p{(\linewidth - 16\tabcolsep) * \real{0.1111}}
  >{\raggedright\arraybackslash}p{(\linewidth - 16\tabcolsep) * \real{0.1111}}
  >{\raggedright\arraybackslash}p{(\linewidth - 16\tabcolsep) * \real{0.1111}}
  >{\raggedright\arraybackslash}p{(\linewidth - 16\tabcolsep) * \real{0.1111}}
  >{\raggedright\arraybackslash}p{(\linewidth - 16\tabcolsep) * \real{0.1111}}
  >{\raggedright\arraybackslash}p{(\linewidth - 16\tabcolsep) * \real{0.1111}}
  >{\raggedright\arraybackslash}p{(\linewidth - 16\tabcolsep) * \real{0.1111}}@{}}
\toprule\noalign{}
\begin{minipage}[b]{\linewidth}\raggedright
DGP
\end{minipage} & \begin{minipage}[b]{\linewidth}\raggedright
\(M / M_0\)
\end{minipage} & \begin{minipage}[b]{\linewidth}\raggedright
coverage, \(\tau(13)\) (IM / fixed \(z\))
\end{minipage} & \begin{minipage}[b]{\linewidth}\raggedright
median width
\end{minipage} & \begin{minipage}[b]{\linewidth}\raggedright
median set width
\end{minipage} & \begin{minipage}[b]{\linewidth}\raggedright
median \(c\)
\end{minipage} & \begin{minipage}[b]{\linewidth}\raggedright
coverage, \(\tau(17)\) (IM / fixed \(z\))
\end{minipage} & \begin{minipage}[b]{\linewidth}\raggedright
median width
\end{minipage} & \begin{minipage}[b]{\linewidth}\raggedright
empty rate
\end{minipage} \\
\midrule\noalign{}
\endhead
\bottomrule\noalign{}
\endlastfoot
G2 & 0.4 (violated) & .127 / .123 & 0.33 & 0.12 & 1.65 & .107 / .107 &
0.43 & .05--.07 \\
G2 & 1.0 & \textbf{.957} / .957 & 0.69 & 0.48 & 1.64 & \textbf{.957} /
.957 & 0.91 & 0 \\
G2 & 2.0 & 1.000 / 1.000 & 1.29 & 1.08 & 1.64 & 1.000 / 1.000 & 1.71 &
0 \\
G2 & 4.0 & 1.000 / 1.000 & 2.49 & 2.28 & 1.64 & 1.000 / 1.000 & 3.31 &
0 \\
G3 (coincident endpoints) & 0.4 (violated) & .020 / .007 & --- &
\(-0.36\) & 1.96 & .023 / .007 & --- & 1.00 \\
G3 (coincident endpoints) & 1.0 & \textbf{.953} / .900 & 0.25 & 0.00 &
1.93 & \textbf{.940} / .910 & 0.33 & .47--.50 \\
G3 (coincident endpoints) & 2.0 & 1.000 / 1.000 & 0.82 & 0.60 & 1.64 &
1.000 / 1.000 & 1.08 & 0 \\
G3 (coincident endpoints) & 4.0 & 1.000 / 1.000 & 2.02 & 1.80 & 1.64 &
1.000 / 1.000 & 2.68 & 0 \\
G4 (tied constraint) & 0.4 (violated) & .017 / .017 & 0.38 & 0.19 & 1.64
& .013 / .013 & 0.51 & 0 \\
G4 (tied constraint) & 1.0 & \textbf{.880} / .880 & 0.74 & 0.55 & 1.64 &
\textbf{.890} / .890 & 0.99 & 0 \\
G4 (tied constraint) & 2.0 & 1.000 / 1.000 & 1.34 & 1.15 & 1.64 & 1.000
/ 1.000 & 1.79 & 0 \\
G4 (tied constraint) & 4.0 & 1.000 / 1.000 & 2.54 & 2.35 & 1.64 & 1.000
/ 1.000 & 3.39 & 0 \\
\end{longtable}

Three readings, all of a finite-sample experiment with \(300\)
replications; the Monte Carlo standard error of a coverage figure is
\(\sqrt{p(1-p)/300}\), about \(.013\) near \(.95\) and about \(.019\)
and \(.018\) at the G4 figures of \(.880\) and \(.890\). First, at the
true \(M\) on the designs without a tie, estimated coverage of the true
value is close to 95\% (G2: \(.957\); G3: \(.953\) and \(.940\) with the
adaptive critical value against \(.900\) and \(.910\) with the fixed
\(1.645\)); above the true \(M\) the interval is conservative, and under
a violated bound coverage collapses, as it should. On G2 the set is wide
relative to the endpoint standard errors, \(c = 1.645\), and the two
intervals coincide. On G3 the estimated set is empty in half the
replications, since sampling error puts the two coincident endpoints in
the wrong order with probability one half, and the adaptive critical
value rises toward \(1.96\); the five-point difference between the two
intervals is the observed difference on this design, not a demonstration
of a theorem. Second, on G4---the design with a genuine tie---the
interval \textbf{undercovers}: \(.880\) and \(.890\) at \(M = M_0\),
with the adaptive and fixed critical values agreeing because the
estimated set is wide. The mechanism is the non-smoothness the tie
creates: the sample maximum of two drifts with equal population values
is biased upward, so the estimated upper endpoint
\(M - \max_e \hat d_e\) is biased downward, and a bootstrap standard
deviation of that endpoint does not correct a bias; the truth sits at
the upper end of the population set and is missed about one time in
nine. Third, none of this is a validity result. The uniform coverage of
Imbens and Manski (\citeproc{ref-imbensmanski2004}{2004}, Lemma 4,
p.~1850) holds under their Assumption 1, whose part (iii) reads: for all
\(\epsilon > 0\) there are \(\nu > 0\), \(K\), and \(N_0\) such that
\(N \ge N_0\) implies
\(\Pr(\sqrt N\,|\hat\Delta - \Delta| > K\Delta^{\nu}) < \epsilon\),
uniformly over the model class. At a point-identified parameter,
\(\Delta = 0\), this requires \(\Pr(\hat\Delta \neq 0) \to 0\). Our raw
width \((s-1)(2M - |\hat d_1 - \hat d_2|)\) is continuously distributed
at every finite \(n\), and its truncation at zero is nonzero with
probability tending to one half, so the condition fails and Lemma 4 does
not apply to the interval as implemented uniformly over a model class
that includes the point-identification boundary represented by G3. That
is a failure of a sufficient condition over that class, not a proof of
undercoverage, and not a statement about a class bounded away from zero
width and from tied constraints, where a different validity argument
remains possible; G4's tie additionally strains the joint
asymptotic-normality premise of the endpoint estimators. (Our
\(\hat\sigma\) is an endpoint standard error, so the ratio
\(\hat\Delta/\hat\sigma\) used in the critical-value equation is their
\(\sqrt N \hat\Delta / \max(\hat\sigma_l, \hat\sigma_u)\) of equation
(7), in which \(\hat\sigma\) denotes an asymptotic standard deviation;
no factor is missing.) What the experiment supports is narrower: the
sharp population bounds of Proposition~\ref{prp-bounds} are exact; the
adaptive critical value improves on the fixed one on the
coincident-endpoint design; and where a constraint is tied the endpoint
estimator's bias, not its variance, is the problem, which a procedure
that resamples the constraint set or corrects the extremum's bias would
need to address. Two practical lessons attach. The intervals are wide at
late tenure: the interval for \(m\) enters \(\tau(s)\) multiplied by
\((s-1)\), so at \(s = 17\) even the correctly specified bound spans
nearly a standard deviation; the level of long-run conditioning is
expensive, and only some curvature contrasts are estimable, under the
support conditions of Corollary~\ref{cor-curvature}. And with three
cohorts the emptiness test has little power against modest violations:
refutation of a drift bound needs more refreshment cohorts or an
external anchor. {[}sims/sim3\_drift\_bounds.R; sim3\_results.csv{]}

\section*{Appendix B: Simulation protocol (ADEMP-PreReg
format)}\label{appendix-b-simulation-protocol-ademp-prereg-format}
\addcontentsline{toc}{section}{Appendix B: Simulation protocol
(ADEMP-PreReg format)}

This appendix states the simulation protocol in the ADEMP structure,
with the prespecification status of every analysis flagged:
\textbf{{[}P{]}} = prespecified before any simulation was run;
\textbf{{[}E→C{]}} = discovered in an exploratory run, then confirmed in
a fresh simulation; \textbf{{[}A{]}} = added or corrected after the
protocol was first fixed, and therefore not prespecified.

\textbf{Aims.} (A1) \textbf{{[}P{]}} Verify numerically that level
estimates of \(\tau\) depend on the normalization while identified
curvature does not (Theorem~\ref{thm-linear},
Corollary~\ref{cor-curvature}). (A1\('\)) \textbf{{[}A{]}} Noiseless
identity tests of the periodic null direction. (A2) \textbf{{[}P{]}}
Size and power of the plateau observable-implications test
(Proposition~\ref{prp-plateau}), including its predicted blindness to
pure linear drift. (A3) \textbf{{[}P{]}} Bias of the plateau-constrained
estimator under correct and violated saturation. (A4) \textbf{{[}P{]}}
Coverage of the drift-bound intervals (Proposition~\ref{prp-bounds}) at
correctly specified, conservative, and violated \(M\), at the covered
tenures. (A5) \textbf{{[}A{]}} Compare readouts of identified and
unidentified functionals from saturated regressions with explicit
contrasts. (A6) \textbf{{[}A{]}} Coverage of the Imbens--Manski interval
against a fixed-\(z\) interval on a design in which the identified set
is a point. (A7) \textbf{{[}A{]}} Coverage on a design in which the
binding constraint is tied.

\textbf{Data-generating mechanisms.} JLPS-calibrated support: cohorts
entering \(t = 1, 5, 13\), \(T = 19\). Sims 1--2: cell noise
\(1/\sqrt{1000}\); paths S1 null, S2
\(\tau(s) = 0.2\{1 - e^{-0.6(s-1)}\}\), S3 = S2 + \(0.01(s-1)\), S4 = S2
+ \(0.004\max(s-6,0)^2\); cohort effects G0 null, G1 shock \(+0.15\) on
the 2011 cohort. Sim 3: outcome sd \(1\), \(n = 500\) per cell; S2\('\):
\(\tau(s) = 0.25\{1 - e^{-(s-1)/2}\}\); G2 drifting
\(g = (0, 0.10, 0.22)\), drifts \((+M_0, +0.6M_0)\) with
\(M_0 = 0.025\); G3 \emph{coincident endpoints},
\(g = (0, 0.10, -0.10)\), drifts \((+M_0, -M_0)\), so that at
\(M = M_0\) the set for \(m\) is a point whose two endpoints are extrema
over distinct drifts; G4 \emph{tied constraint},
\(g = (0, 0.10, 0.30)\), drifts \((+M_0, +M_0)\), so that the maximum
over drifts is a tie and the set for \(m\) at \(M = M_0\) has width
\(2M_0\). Factorial: sim 1 crosses S1--S4 × \{G0, G1\}; sim 2 runs S2,
S3 × G1; sim 3 runs S2\('\) × \{G2, G3, G4\}. \(B = 500\) (sims 1--2),
\(B = 300\) with \(150\) bootstrap draws (sim 3).

\textbf{Estimands.} \(\Delta_4^2\tau(5) = \tau(9) - 2\tau(5) + \tau(1)\)
and \(\theta = \Delta^2\tau(2) - \Delta^2\tau(6)\) (identified);
\(\Delta^2\tau(3)\) (not identified; reported to exhibit its
arbitrariness); \(\bar\tau - \tau(1)\) under the plateau restriction
with window \(s \ge 6\); \(\tau(13) - \tau(1)\) and
\(\tau(17) - \tau(1)\) for the drift bounds; the observable-implications
and emptiness test statistics.

\textbf{Methods.} Saturated cell-mean OLS under normalization N1
(reference level) and N2 (zero mean along the affine direction);
explicit cell-mean contrasts; the plateau-constrained estimator with the
cohort-gap-constancy test; the representative-plus-drift-interval
construction of Proposition~\ref{prp-bounds} with the linear part read
on the subgrid \(s \equiv 1 \pmod 4\), parametric-bootstrap endpoint
SEs, and the interval
\([\hat\tau_{lo} - c\hat\sigma_{lo}, \hat\tau_{hi} + c\hat\sigma_{hi}]\)
with a critical value \(c\) in the form of Imbens and Manski
(\citeproc{ref-imbensmanski2004}{2004}) (coverage of the true parameter
value; \(c\) between \(1.645\) and \(1.96\) according to the estimated
width, truncated at zero, relative to the endpoint standard error; raw
endpoints, which may cross), with the fixed-\(z\) (\(1.645\)) interval
reported alongside. The adaptive critical value is used as a
finite-sample device; no uniform-validity claim is made for it, since
the interval does not satisfy Assumption 1(iii) of the cited article at
a point-identified parameter.

\textbf{Performance measures.} Bias and Monte Carlo SE of each estimand;
exact invariance of identified functionals under null-direction
perturbations (predicted identity, not approximation); rejection rates
of the observable-implications test at nominal .05; coverage of the true
\(\tau(s)\) by the Imbens--Manski interval and by the fixed-\(z\)
interval, median interval width, median estimated set width, and median
critical value at each \(M/M_0 \in \{0.4, 1, 2, 4\}\); emptiness rate.

\textbf{Deviations from protocol.} Recorded here in full. (1)
Proposition~\ref{prp-bounds} was restated in terms of first-difference
drift, which the affine direction moves, before sim 3 was run; an
earlier draft had bounded second differences, which it does not. (2)
Sims 1--2 were rewritten with noiseless identity tests after the
periodic null direction was recognized; an earlier reading of the same
discrepancy as ``boundary-dummy contamination'' of saturated regressions
was a misreading of the non-identification of ordinary second
differences on a stride-four support and is withdrawn. (3) Sim 2's
contrast estimator is centred so that its estimand is \(\theta\) above.
(4) Sim 3 replaced a fixed critical value of \(1.645\) by the
Imbens--Manski critical value and added the design G3, on which the
fixed value gives \(90\%\) coverage. G3's endpoints coincide without a
tie inside an extremum; the tied-constraint design G4 was added, on
which the interval undercovers (\(.88\)--\(.89\)), and the
uniform-validity language was withdrawn. (5) Sim 4's attrition mechanism
was made monotone, consistent with D2; the design of the exactly aligned
event study is unchanged. (6) Sim 1's cell-mean contrast for
\(\tau(9) - 2\tau(5) + \tau(1)\) compares the two cohorts at \(t = 9\)
and \(t = 5\), as the code always did; the appendix prose had named the
wrong periods. (7) The lag-4 second differences are now labeled by their
center (\(\Delta_4^2\tau(5)\), \(\Delta_4^2\tau(9)\)) rather than by
their first tenure, the G4 population width is stated correctly as
\(2M_0\), the coverage target is named as coverage of the true parameter
value, and \texttt{check\_manuscript\_values.R} was rewritten to fail on
missing, duplicated, or non-finite rows and to use a scale-aware
tolerance for floating-point identities; no random draw changed. All
scripts and seed-controlled outputs accompany the replication package,
and \texttt{check\_manuscript\_values.R} asserts every number quoted in
Appendices A and C against named output columns.

\section*{Appendix C: Simulation 4 --- downstream
propagation}\label{appendix-c-simulation-4-downstream-propagation}
\addcontentsline{toc}{section}{Appendix C: Simulation 4 --- downstream
propagation}

\emph{Design.} Staggered panel with cohorts entering \(t = 1, 5, 13\)
(stride \(d = 4\)), \(T = 19\), \(n = 200\) per cohort, unit and time
effects, monotone attrition with an exit hazard increasing in tenure (a
unit that leaves does not return, so that D2 holds and the support is an
unbalanced trapezoid), and a staggered binary regressor \(D\) adopted at
a random calendar time after entry (30\% never adopt) with true
coefficient \(0.5\). Reports add a conditioning path to the
unconditioned outcome: S (saturating,
\(\tau(s) = 0.30\{1 - e^{-0.6(s-1)}\}\)), L (affine, \(0.05(s-1)\)), or
S+L. For the pre-trend block, eight consecutive cohorts
(\(e = 1, \dots, 8\), \(T = 20\), \(n = 60\), no attrition) each
experience an event five waves after entry (\(\delta = 6\)), the
estimation window is \(k \in [-5, 6]\), and the unconditioned outcome
has no event effect and no anticipation; the fully dynamic specification
omits \(k_0 = -1\) and \(k_1 = -2\). Seed 20260916; code
\texttt{sims/sim4\_downstream.R} (v3); every number below is written by
the script to \texttt{sim4\_manuscript\_values.csv} under a named
quantity. The cell-mean design of the main support has 41 cells, 39 free
parameters, rank 35, and nullity 4, as Theorem~\ref{thm-linear}(d)
requires for \(d = 4\) under C\(_4\) on a connected support.

\emph{Deterministic identities.} Conditional on one realization of the
design, Theorem~\ref{thm-absorb} holds exactly: the observed bias of the
TWFE coefficient equals \((D'MD)^{-1}D'M\tau(s)\) to
\(2 \times 10^{-14}\) (S and S+L: bias \(0.0256\); L: bias
\(9 \times 10^{-15}\)). Adding the affine direction \(0.7 + 0.2(s-1)\)
changes the estimate by \(4 \times 10^{-14}\), and adding the
\emph{periodic} direction \(\rho = (0, 0.30, -0.10, 0.20)\) of period
four changes it by \(9 \times 10^{-15}\): the bias is invariant along
the entire four-dimensional identified set, as
Theorem~\ref{thm-absorb}(i) states. In a single-cohort panel the bias is
\(10^{-16}\) for the saturating path (Corollary~\ref{cor-protect}(a)).
Under exact alignment, all ten event-study coefficients shift by
\(\tau(k+\delta) - \ell(k)\) to \(1.5 \times 10^{-14}\)
(Theorem~\ref{thm-pretrend}). The residuals just quoted are those of the
archived run, whose environment the release record documents; they are
machine-precision quantities that vary with the platform's linear
algebra, and the checker asserts each identity at the scale-aware
tolerance \(10^{-10}\), some five orders of magnitude below the
displayed precision of any estimate, rather than requiring the archived
residuals to recur. The measurement-induced \emph{shifts} in the three
pre-period coefficients are \(-0.183\), \(-0.070\), and \(-0.018\) at
\(k = -5, -4, -3\); these are not the estimated coefficients themselves,
which in this realization are \(+0.068\), \(+0.061\), and \(+0.014\) on
the reports and \(+0.252\), \(+0.131\), and \(+0.032\) on the
unconditioned outcome---the latter being sampling error of a regression
with 480 units and no event effect, and the former that error plus the
shift. On the same aligned design the residualized event-time and tenure
indicators are collinear, \(\mathrm{rank}[MX, MZ_\tau] = 10\) against
\(21\) columns, so condition (R) fails and no within-sample correction
exists. The rank-deficient design of Section~\ref{sec-pretrend} (event
dates \(\{0,2,4\}\), periods \(4\)--\(7\), window \(2\)--\(5\), two
references omitted) has eight columns and rank seven. The
partial-alignment design of Proposition~\ref{prp-partial} (event dates
\(6, \dots, 13\), two units per date with \(\delta_i \in \{6,7\}\)) has
a tenure composition of exactly one half at each of \(k + 6\) and
\(k + 7\) in every \((k,t)\) cell, \(\mathrm{rank}[W_0, X] = 44\) of
\(44\), \(\max|X'MZ_\tau| = 6.0\), and shifts of \(-0.1419\),
\(-0.0544\), \(-0.0142\) at \(k = -5, -4, -3\) for the saturating path.

\emph{Identified correction.} On the main design, the joint regression
Equation~\ref{eq-joint} (four tenure indicators aliased with the fixed
effects and dropped) returns \(0.578\) against the infeasible
unconditioned estimate \(0.568\) and the reported-outcome estimate
\(0.594\), with unit-clustered standard error \(0.051\). Under
tenure-determined adoption, \(D = \mathbf{1}\{s \ge 4\}\), the regressor
is collinear with the tenure indicators and condition (R) fails. The
script detects this by computing the residual rank of \(D\) after the
fixed effects and tenure indicators are partialled out and records the
corrected estimator as not identified; it also enters \(D\) last so that
R's \texttt{lm} reports \(D\), rather than a tenure indicator, as
aliased, but the rank computation is the check, since with \(D\) entered
first the same routine would return a coefficient on \(D\) and drop a
tenure column instead.

\emph{Monte Carlo.} \(B = 100\) replications, \(n = 120\) per cohort,
unit-clustered 95\% intervals. Five designs: baseline (S+L path, random
adoption); zero conditioning (\(\tau \equiv 0\)); heterogeneous
treatment effects (\(\beta_i \sim N(0.5, 0.2^2)\)); cohort-dependent
adoption (never-adopter shares \(.5, .3, .1\) by cohort);
tenure-determined adoption.

\begin{longtable}[]{@{}
  >{\raggedright\arraybackslash}p{(\linewidth - 12\tabcolsep) * \real{0.1429}}
  >{\raggedright\arraybackslash}p{(\linewidth - 12\tabcolsep) * \real{0.1429}}
  >{\raggedright\arraybackslash}p{(\linewidth - 12\tabcolsep) * \real{0.1429}}
  >{\raggedright\arraybackslash}p{(\linewidth - 12\tabcolsep) * \real{0.1429}}
  >{\raggedright\arraybackslash}p{(\linewidth - 12\tabcolsep) * \real{0.1429}}
  >{\raggedright\arraybackslash}p{(\linewidth - 12\tabcolsep) * \real{0.1429}}
  >{\raggedright\arraybackslash}p{(\linewidth - 12\tabcolsep) * \real{0.1429}}@{}}
\toprule\noalign{}
\begin{minipage}[b]{\linewidth}\raggedright
Design
\end{minipage} & \begin{minipage}[b]{\linewidth}\raggedright
bias, reported
\end{minipage} & \begin{minipage}[b]{\linewidth}\raggedright
RMSE
\end{minipage} & \begin{minipage}[b]{\linewidth}\raggedright
coverage
\end{minipage} & \begin{minipage}[b]{\linewidth}\raggedright
bias, corrected
\end{minipage} & \begin{minipage}[b]{\linewidth}\raggedright
RMSE
\end{minipage} & \begin{minipage}[b]{\linewidth}\raggedright
coverage
\end{minipage} \\
\midrule\noalign{}
\endhead
\bottomrule\noalign{}
\endlastfoot
baseline & \(0.025\) & \(.067\) & \(.95\) & \(0.002\) & \(.063\) &
\(.98\) \\
zero conditioning & \(0.010\) & \(.063\) & \(.99\) & \(0.012\) &
\(.065\) & \(.98\) \\
heterogeneous effect & \(0.018\) & \(.059\) & \(.99\) & \(-0.006\) &
\(.058\) & \(.99\) \\
cohort-dependent adoption & \(0.042\) & \(.074\) & \(.92\) & \(-0.001\)
& \(.064\) & \(.98\) \\
tenure-determined adoption & \(0.149\) & \(.164\) & \(.46\) & not
identified (R fails) & --- & --- \\
\end{longtable}

Monte Carlo standard errors of the reported biases are \(.006\). The
corrected estimator is unbiased in every design in which (R) holds and
is unaffected by treatment-effect heterogeneity that is independent of
the design, as Theorem~\ref{thm-correction} requires; its RMSE is below
the naive estimator's when conditioning is present (\(.063\) against
\(.067\); \(.064\) against \(.074\)) and marginally above it when there
is none (\(.065\) against \(.063\)), which is the observed price of
estimating a tenure component that is zero. Its intervals are
conservative here (\(.98\)--\(.99\)). Under tenure-determined adoption
the naive estimator is badly biased and the corrected one does not
exist---the design cannot separate exposure from conditioning, and the
rank check says so before any number is reported.

\section*{Appendix D: Deterministic design
checks}\label{appendix-d-deterministic-design-checks}
\addcontentsline{toc}{section}{Appendix D: Deterministic design checks}

The results of Section~\ref{sec-impossibility},
Section~\ref{sec-recovery}, and Proposition~\ref{prp-sensitivity} are
algebraic, and the replication package checks them by computation on
stated designs rather than by simulation: no randomness and no data. Two
kinds of arithmetic are used, and the scripts say which.

\begin{itemize}
\tightlist
\item
  \emph{Exact rational} rank and null-space computations (Python
  \texttt{fractions}, no tolerance): \texttt{check\_dose\_sharpness.py};
  \texttt{check\_recovery\_support.py}; blocks (A) and (D\('\)) of
  \texttt{check\_interrupted.py}; the named designs of
  \texttt{check\_lemma1\_components.py}.
\item
  \emph{Floating-point} rank and residual computations (numpy SVD at
  tolerance \(10^{-10}\); R's \texttt{qr} and \texttt{svd}): the
  10,934-design sweeps in block (B) of \texttt{check\_interrupted.py}
  and in \texttt{check\_lemma1\_components.py};
  \texttt{check\_thm1\_support.py};
  \texttt{check\_lemma1\_trapezoid.py}; \texttt{design\_rank.R}.
\end{itemize}

All scripts are in \texttt{sims/}. Every figure quoted in the text from
a sweep is a count of designs, for which the floating-point rank of a
small 0--1 matrix is reliable; every specific design quoted in the text
is also checked exactly.

\textbf{Support conditions and reference categories
(Theorem~\ref{thm-linear}(d), Section~\ref{sec-recovery},
Theorem~\ref{thm-correction}, Theorem~\ref{thm-pretrend}).}
\texttt{check\_recovery\_support.py} verifies by exact arithmetic: the
six-cell support of cohorts \(1, 4, 6\) (rank six of nine, full nullity
three, projected nullity one, C\(_1\) satisfied, three incidence
components), on which exact knowledge of \(g\) leaves \(\tau(2)\) free,
against the connected trapezoid on the same cohorts, on which it pins
the path, and the two-component support of cohorts \(1, 2, 5\), on which
knowledge of \(g\) pins \(\tau(2)\) although (CG) fails; the identity
\(\dim\mathcal{K} = \dim\mathcal{K}_\tau + (c-1)\) on 9,183 supports
(trapezoids with \(K \in \{2,3\}\) cohorts in \(\{1,\dots,6\}\),
\(T \le 8\), and up to two cells removed; 140 of them disconnected; no
violation); the event-study design of the proof of
Theorem~\ref{thm-pretrend}, in which a post-event reference category
produces \(\mathbb{E}[\hat\beta^{\circ}_{-2}] = +\tfrac12\) under
homogeneous effects and zero pre-period coefficients once both
references precede the event; the failure of (R) under
\(D = \mathbf{1}\{s \ge 4\}\) on the support of Section~\ref{sec-jlps},
with two exact fits at \(\beta = 0\) and \(\beta = 1\); the
residue-class structure of Corollary~\ref{cor-curvature} at \(d = 2\);
the centered plateau boundary of Proposition~\ref{prp-plateau}; the
population widths of design G4; and the uninterrupted rows of
Table~\ref{tbl-designs}.

\textbf{Dose designs (Theorem~\ref{thm-dose}, Corollary~\ref{cor-cps},
Table~\ref{tbl-designs}).} For each schedule,
\texttt{check\_interrupted.py} builds the cell design of
Equation~\ref{eq-dose}, computes its exact rank and a basis of the
kernel, and reports whether the kernel's \(\tau\)-component is spanned
by \(j(\cdot)\) and whether a dose-linear path lies in it. Its final
block recomputes the descriptive profiles of Table~\ref{tbl-cps} from
the published indices; consistent with Section~\ref{sec-cps}, it labels
them as transformations of published averages that retain the cohort
effects, not as identified contrasts. With the CPS pattern
\(J = \{0,1,2,3,12,13,14,15\}\) and monthly entry over ten or more
cohorts the design has nullity one and kernel
\(\propto (0,1,2,3,12,13,14,15)\)---that is, \(j(k)\)---and the
dose-linear path is \emph{not} in the kernel; with \(J = \{0,\dots,7\}\)
the kernel is \(\propto (0,\dots,7) = k-1\) and the dose-linear path is
in it, as Theorem~\ref{thm-dose}(d) requires. The pattern
\(\{0,1,4,5\}\) behaves like the CPS one and the equally spaced pattern
\(\{0,2,4,6,8,10\}\) like the uninterrupted one. With entry every second
period the CPS pattern gives nullity two, the second kernel direction
being \(\rho(j) = j \bmod 2\), the periodic direction of
Theorem~\ref{thm-dose}(a). The equally spaced pattern
\(\{0,2,4,\dots\}\) also illustrates the incidence-graph caveat of
Theorem~\ref{thm-linear}(d): cohorts of odd and even entry never share a
period, the incidence graph has two components, and the design carries a
kernel direction that moves \(\alpha\) and \(g\) against each other with
no movement in \(\tau\) at all.

\textbf{Condition (P\('\)) and the counterexample to the earlier
statement.} \texttt{check\_dose\_sharpness.py} computes, for any cohort
set and pattern, the exact rank of the \emph{conditioning projection} of
the kernel---not merely the kernel's nullity, since the projection can
have lower rank. It reproduces the design of cohorts \(\{1,2,4\}\),
offsets \(\{0,1,3\}\), \(T = 7\) (ten columns, rank eight, projected
dimension two, against the one asserted by the earlier version of
Theorem~\ref{thm-dose}(b)); sweeps eleven patterns, strides
\(d \in \{1,2,3\}\) and windows \(L \in \{1,\dots,16\}\) under
continuous recruitment, and finds the projected dimension equal to the
number of represented residue classes in all 326 designs where (P\('\))
holds, equal to it in 123 designs where (P\('\)) fails, and larger in
47---so (P\('\)) is sufficient, not necessary, and not vacuous; shows
that gapped cohort sets \(\{1,2,4\}\), \(\{1,3,4\}\), \(\{1,2,4,5\}\)
with the same pattern have projected dimension two while \(\{1,2,3,4\}\)
has one; and confirms that for the CPS pattern \(\Gamma(J, 1) = 9\), so
that the dimension is two for nine monthly cohorts and one from ten on.
It also records that the pattern's difference set omits \(4\) through
\(8\), which the withdrawn condition required and no recruitment window
can supply.

\textbf{Lemma 1 (o) and (iv).} A sweep of all entry sets
\(\mathcal{E} \subset \{1,\dots,13\}\) with \(e_1 = 1\),
\(3 \le K \le 5\), and \(T \le e_K + 13\)---10,934 designs---finds no
violation of \(\nu = \#\{\text{increment-graph components}\}\) and none
of \(\nu \le \max\{d, \Delta_2 - w\}\); the bound is attained in
\(9{,}746\) of them and strict in \(1{,}188\)
(\texttt{check\_lemma1\_components.py}, which repeats the sweep
standalone). The version of the bound computed from the \emph{third}
cohort's follow-up instead of the last cohort's fails in 72 of these
designs, the smallest being \(\mathcal{E} = \{1,5,9,11,12\}\) with
\(T = 12\).

\textbf{Sensitivity weights (Proposition~\ref{prp-sensitivity}).} For
the configuration of Simulation 4 the script solves for the
representation of the shift functional in the basis of second
differences and reports the weights, which are the ramps \(\{1\}\),
\(\{1,2\}\), \(\{1,2,3\}\) at \(k = -3, -4, -5\) and \(\{1,\dots,7\}\)
at \(k = 6\), with absolute sums \(1, 3, 6\) and \(28\)---the triangular
numbers \(m_k(m_k+1)/2\) of Equation~\ref{eq-sensitivity}---and confirms
Equation~\ref{eq-sensitivity} against the realized shifts of Appendix C
at every event time in the window. An equality test with \(C = .06\)
verifies that a path of constant centered curvature \(+C\) attains the
bound exactly at every \(k\) in the window (\(.36, .18, .06\) at
\(k = -5, -4, -3\); \(.06, .60, 1.68\) at \(k = 0, 3, 6\)) and that a
path whose curvature alternates in sign stays strictly below it at
\(|k| \ge 4\).

\section*{Data and code availability}\label{data-and-code-availability}
\addcontentsline{toc}{section}{Data and code availability}

The analysis uses simulated data and published aggregate indices, and
every input needed for reproduction is included in the public archive.
Every number the paper reports is produced either by a simulation that
generates its own data from a stated seed or by a deterministic
algebraic computation on a stated design; the illustration of
Section~\ref{sec-cps} is a descriptive transformation of month-in-sample
indices published in Bailar (\citeproc{ref-bailar1975}{1975}, Tab. 1),
Solon (\citeproc{ref-solon1986}{1986}, Tab. 1) and McIllece
(\citeproc{ref-mcillece2022}{2022}, Tab. 2), which are reproduced inside
the code. No restricted or licensed microdata are used, and none are
needed to reproduce any table or appendix of this paper.

The complete archive---the scripts, all of their outputs, the manuscript
source and its build script---is public at
\url{https://github.com/sokubo/paper-panel-conditioning-replication}.
The numerical results of this manuscript are those of tag
\texttt{paper-v1.0}, whose README gives the execution order, the
environment, the runtimes, and the map from each table and appendix
number to the script and output file that produce it.
\texttt{check\_manuscript\_values.R} asserts every number quoted in
Appendices A and C against a named output column---requiring each key to
select exactly one row, each compared vector to have its expected length
and only finite values, and each floating-point identity to hold at a
stated scale-aware tolerance---and exits with an error on any failure;
its \texttt{-\/-selftest} mode confirms that a deleted row, a duplicated
row, a perturbed value, and a non-finite value each make it fail.
\texttt{check\_dose\_sharpness.py} reproduces the counterexample to the
earlier form of Theorem~\ref{thm-dose} and verifies condition (P\('\))
by exact arithmetic, and \texttt{check\_recovery\_support.py} verifies
the support conditions of Appendix D. The manuscript in that tag is an
earlier version of this text. The present version, dated 24 September
2026, adds presentation and reference updates and
Figure~\ref{fig-support} and Figure~\ref{fig-cps}, which depict designs
and paths stated in the text, and changes no numerical result; the
script that draws the figures, \texttt{make\_figures.py}, first checks
the quantities the text states and is distributed with the arXiv version
as an ancillary file. The computational commit checked against the
numerical results is \texttt{f06283e}; the tag points at the commit that
adds the release record \texttt{RELEASE\_CHECK.md}, which documents an
anonymous download of the published snapshot, a comparison of its file
list against the archive's manifest, and a re-execution of the
documented sequence in a clean copy with the regenerated outputs
compared against the shipped ones by \texttt{compare\_outputs.py}, which
is itself part of the archive. That comparison is file by file and cell
by cell, not a test of selected quantities: for each CSV it requires
identical column names, identical numbers of rows and columns, every
text cell identical, every missing cell missing in both, and every
numeric cell---integer counts included---equal to within \(10^{-8}\);
for the console transcript it applies the same rules token by token on
each line. An empty cell, \texttt{NA} and \texttt{NaN} are treated as
the same missing value; beyond that, and beyond whitespace (line
endings, trailing and inter-token spacing, blank lines) and numeric
spellings that agree to within the tolerance, no print-format difference
is tolerated. Its \texttt{-\/-selftest} mode confirms that a changed
scenario label, a changed integer count (with or without an accompanying
label change), a changed identification verdict in the transcript, an
added, deleted or reordered row, a dropped or renamed column, a
perturbed value, and a value replaced by a missing entry each make it
fail, and that a perturbation of \(10^{-12}\), an
\texttt{NA}/\texttt{NaN} swap, and trailing whitespace do not. The
release script exits with an error if any step fails, any file is
missing or unlisted, or the comparison fails on any file. That record is
the author's own verification. The tags \texttt{paper-v0.7} and
\texttt{paper-v0.9}, with computational commits \texttt{a48bd80} and
\texttt{ea044e6}, remain available and correspond to earlier versions of
this manuscript, dated 21 and 22 September 2026; no \texttt{paper-v0.8}
tag was published, the v0.8 manuscript having been superseded before its
release.

\section*{References}\label{references}
\addcontentsline{toc}{section}{References}

\phantomsection\label{refs}

\end{document}